\documentclass[11pt]{article}

\usepackage{mathtools}
\usepackage[utf8]{inputenc}
\usepackage{todonotes}

\usepackage[top=4cm,bottom=4cm,left=3cm,right=3cm]{geometry}

\usepackage[square,numbers]{natbib}
\setcitestyle{authoryear,open={(},close={)}}
\usepackage[dvipsnames]{xcolor}        

\RequirePackage[colorlinks,citecolor=blue,urlcolor=blue]{hyperref}

\usepackage{graphicx,graphics}
\usepackage{subcaption}
\usepackage{enumitem}

\usepackage{url}
\usepackage{supertabular}
\usepackage{amsmath,amscd,amsfonts,amssymb}
\usepackage{fancyhdr}
\usepackage{color}
\usepackage{verbatim}

\usepackage{array,supertabular}
\usepackage{float}
\usepackage{rotating}
\usepackage{lscape}

\usepackage{booktabs,caption,fixltx2e}
\usepackage[flushleft]{threeparttable}
\usepackage{amsthm}

\usepackage{dcolumn}
\newcolumntype{L}{D{.}{.}{2,5}}

\newtheorem{theorem}{Theorem}[section]
\newtheorem{lemma}[theorem]{Lemma}

\newtheorem{corollary}[theorem]{Corollary}

\usepackage[title]{appendix}

\theoremstyle{remark}
\newtheorem{definition}{Definition}

\long\def\symbolfootnote[#1]#2{\begingroup%
\def\thefootnote{\fnsymbol{footnote}}\footnotetext[#1]{#2}\footnotemark[#1]\endgroup}

\newcommand{\cor}{\operatorname{Corr}}

\makeatletter
\@addtoreset{equation}{section}

\makeatother

\newcounter{Fig}[figure]

\newcounter{Tab}[table]

  {\refstepcounter{Tab}
   \addtocounter{table}{1}
   \samepage\vspace{0.2cm}
   \centerline{#1} \nobreak
   \begin{verse}{Table \thesection.\arabic{Tab}:}
  }%
  {\end{verse} \vspace{0.2cm}}

\relax

\usepackage{footnote}
\usepackage{adjustbox}

\def\1{{\mathbf 1}}
\def\0{{\mathbf 0}}

\newtheorem{condition}{Condition}
\newcommand{\reels}{\mathbb{R}}
\newcommand{\naturels}{\mathbb{N}}

\newcommand{\esp}{\mathbb{E}}
\newcommand{\proba}{\mathbb{P}}
\newcommand{\filt}{\mathcal{F}}
\newcommand{\sigalg}{\mathcal{G}}

\newcommand{\ind}{\mathbf{1}}
\newcommand{\var}{\operatorname{Var}}

\newcommand{\be}{\begin{eqnarray}}
\newcommand{\ee}{\end{eqnarray}}
\newcommand{\bes}{\begin{eqnarray*}}
\newcommand{\ees}{\end{eqnarray*}}
\newcommand{\cvdistrib}{\overset{\mathcal{D}}{\rightarrow}}

\newcommand{\cvproba}{\overset{\proba}{\rightarrow}}

\allowdisplaybreaks

\newcommand{\simpanel}[5]{%
\begin{minipage}[t]{0.242\textwidth}
\centering
\includegraphics[
    width=\linewidth,
    height=0.108\textheight,
    keepaspectratio
]{#1}

\vspace{-0.45em}

{\fontsize{5.3}{5.8}\selectfont
\renewcommand{\arraystretch}{0.82}
\begin{tabular}{@{}c@{}}
$n_h=#2$\\[-0.1ex]
$\alpha^*=(#3)'$\\[-0.1ex]
$b^*=(#4)'$\\[-0.1ex]
$\alpha^*/b^*=(#5)'$
\end{tabular}
}
\end{minipage}%
}

\newcommand{\simheading}[2]{%
\begin{minipage}[t]{0.242\textwidth}
\centering
{\small
\begin{tabular}{@{}c@{}}
\textbf{#1}\\[-0.35em]
\textit{#2}
\end{tabular}%
}
\end{minipage}%
}

\begin{document}
 \title{Parametric estimation of Hawkes processes based on ordinary least squares}
\author{Benjamin Poignard\footnote{Faculty of Science and Technology, Keio University. 3-14-1 Hiyoshi, Kohoku-ku, Yokohama, Kanagawa 223-8522, Japan. Jointly affiliated at Riken-AIP. Email: bpoignard@math.keio.ac.jp} \, and Yoann Potiron\footnote{(CONTACT) Faculty of Business and Commerce, Keio University. 2-15-45 Mita, Minato-ku, Tokyo, 108-8345, Japan. E-mail: potiron@fbc.keio.ac.jp website: \url{https://www.fbc.keio.ac.jp/~potiron/}} }
\date{}

\maketitle

\begin{abstract}


We develop a parametric estimation framework for self-exciting Hawkes processes whose intensity functions admit a parametric form. The estimation procedure is based on ordinary least squares. To apply the least squares estimation, we restrict to a kernel class that can be expressed as a sum of the product of a parameter and a function. We first establish the central limit theorem for the proposed estimator. We then show the consistency of the asymptotic variance estimator. Finally, we introduce a Wald test statistic and derive its asymptotic distribution. We apply the proposed methodology to trade times from high-frequency financial asset data. Our empirical results provide evidence for three heterogeneous trader types. In particular, we identify two trader types as high-frequency traders and one trader type as fundamental trader.

\end{abstract}
\textbf{Keywords}: Ordinary least squares; Parametric estimation; Self-exciting Hawkes processes; Trader type heterogeneity; Wald test.
\\
\section{Introduction}

Simple point processes are widely used in statistics to characterize event times. The main stylized fact in this strand of literature is the presence of event clustering in time, which motivates the use of the so-called  Hawkes self-exciting process (see \cite{hawkes1971point} and \cite{hawkes1971spectra}). For any time $t \in \reels^+$, we define the simple point process $N_t$ as the cumulative number of events from the starting time $0$ to the final time $t$ and $\lambda_t$ is its intensity, namely the probability of having one event per unit of time. Then, a standard definition of a Hawkes self-exciting process is given by
\begin{equation*}
	\lambda_t = \nu + \int_0^t h(t-s)\, dN_s.
\end{equation*}
Here, the parameter $\nu > 0$ is the Poisson baseline and the function $h$ is the kernel, which is non-negative, namely $h(t)\geq 0$ for any time $t \in \reels^+$. The kernel drives the self-excitation of the Hawkes process. In particular, this accommodates the presence of event clustering in time, which is the main stylized fact observed with financial data. The particular case $h= 0$ corresponds to a classical Poisson process, and we can view Hawkes processes as a natural extension of Poisson processes.

An early application of Hawkes processes concerns modeling in seismology (see \cite{rubin1972regular}, \cite{vere1978earthquake}, \cite{ozaki1979maximum}, \cite{vere1982some}, \cite{ogata1978asymptotic}, \cite{ogata1988statistical},  \break \cite{ikefuji2022earthquake} and \cite{kling2026goodness}). There are also applications in financial econometrics (see \cite{yu2004empirical}, \cite{bowsher2007modelling}, \cite{embrechts2011multivariate}, \cite{AitSahalia2014mutual},  and also \cite{clinet2018statistical}, \cite{corradi2020testing}, \cite{cavaliere2023bootstrap}, \cite{potiron2026mutually}), in finance (see \cite{Large2007},  \cite{ait2015modeling} and \cite{fulop2015self}) and in quantitative finance (\cite{chavez2005estimating},  \cite{bacry2013some} and \cite{morariu2022state}). Some applications also concern biology (see  \cite{reynaud2010adaptive}, \cite{donnet2020nonparametric}, \cite{cai2024latent} and \cite{kang2025analyzing}), social studies (see \cite{fox2016modeling} and \cite{fang2024group}), epidemiology (see \cite{cheysson2022spectral}) and management (see \cite{ait2026saddlepoint}).

In this paper, we focus on a Hawkes self-exciting process in which the kernel is parametric. More specifically, for any parameter $\theta \in \Theta$ and any time $t \in \reels^+$, we introduce the family of intensities $\lambda_t (\theta) = \nu  + \int_0^t  h(t-s,\alpha)\,dN_{s}$. Here, the baseline parameter $\nu$ is of dimension $1$ while the kernel parameter $\alpha$ is a vector of dimension $n_h$, in which $n_h$ will denote the number of functions in the kernel. Denoting by $v'$ the transpose operation for any vector or any matrix $v$, we assume that the parameter $\theta$ has the form $\theta=(\nu,\alpha')'$. We also assume that the number of parameters $n=1+n_h$ satisfies $n \geq 2$, and so $\theta \in \mathbb{R}^n$.
Finally,
we assume the existence of the true parameter $\theta^* \in \Theta$ such that $\lambda_t = \lambda_t (\theta^*) = \nu^*  + \int_0^t  h(t-s,\alpha^*)\,dN_{s}$. Here, the true parameter $\theta^*$ has the form $\theta^*=(\nu^*,(\alpha^{*})')'$. 

The goal of this paper is to retrieve the true parameter $\theta^*$. The parametric estimation procedure is based on ordinary least squares. The ordinary least squares estimator is the value that minimizes the sum of squared residuals. To apply least squares estimation, we restrict to the kernel which has the form $h(t,\alpha) = \sum_{k=1}^{n_h} \alpha_k h_k(t)$. Here, we have $\alpha=(\alpha_1, \ldots, \alpha_{n_h})'$ and the kernel is the sum of $n_h$ functions $h_k$. Thus, this introduces heterogeneity in the kernel. The reason why we have to restrict the class of kernels is due to the parametric estimation procedure based on ordinary least squares. Then, the kernel form that  we have introduced is the most general form we could accommodate within our theoretical framework. In addition, we assume that the functions $h_k$ are exponential. This is the kernel considered in \cite{bacry2020sparse} (see Equations (1) and (2)) who study parametric estimation of Hawkes processes based on penalized least squares. A more general form of the kernel is considered in \cite{reynaud2010adaptive}, \cite{hansen2015lasso} (see Section 1.2.3) and \cite{cai2024latent} (see p. 98) for nonparametric estimation of Hawkes processes based on penalized least squares. See also \cite{kirchner2017estimation} and \cite{fujimori2026sparse} who consider nonparametric estimation of Hawkes processes based on  conditional least squares.

The functions in the kernel introduce heterogeneity. This heterogeneity is observed in financial market trades, as it is well-known that there are different types of traders. Thus, we can naturally interpret each function of the kernel as a trader type. In the literature based solely on trade times but not on price processes which is the case of this paper, there is also some evidence that there are different trader types. \cite{gagnon2010multi} show that price parity deviations relate positively to proxies for holding costs that can limit arbitrage. The empirical application from \cite{Hasbrouck2013low} suggests that high-frequency trading is beneficial to market quality. In \cite{hoffmann2014dynamic}, fast traders can revise their quotes quickly after news arrivals to reduce market risks. \cite{budish2015high}, \cite{biais2015equilibrium}, \cite{foucault2016news} and \cite{pagnotta2018competing} also consider trading speed.

The main idea of this paper is to rewrite the continuous-time Hawkes process as a standard linear regression model based on time series. Then, we can use asymptotic theory designed for linear regression model based on time series. More specifically, we apply the technology of \cite{white2001asymptotic} (see Theorem 5.17, p. 126). In Section \ref{sectheory}, we first give the consistency and central limit theorem for parametric estimation of Hawkes processes based on ordinary least squares (see Theorem \ref{thclt}). Then, we show the consistency of the variance matrix estimation procedure (see Theorem \ref{thconsistency}). In addition, we establish the central limit theorem of the parametric estimation procedure when we estimate the variance matrix (see Theorem \ref{thcltest}). Finally, we introduce a Wald test statistic and show that the statistic converges in distribution to a chi-squared distribution with $q$ degrees of freedom under the null hypothesis $H_0$ and is consistent under the alternative hypothesis $H_1$ (see Corollary \ref{corwald}). The main idea of the proofs is to apply the technology of \cite{white2001asymptotic} (see Theorem 5.17, p. 126). All these theoretical results are novel to the literature on Hawkes processes as we are not aware of any work using ordinary least squares for parametric estimation.

In Section \ref{secempirical}, we apply our parametric estimation procedure based on ordinary least squares to real financial high-frequency data. We first investigate the prediction performances of the number of transactions based on ordinary least squares and maximum likelihood procedures. Overall, we slightly improve out-of-sample predictions. Our empirical results on predictions also provide evidence for three heterogeneous trader types. In particular, we identify two types of high-frequency traders and one type of fundamental trader. Finally, we consider a statistical Wald test to confirm the results on heterogeneity of trader types.

The remainder of this paper is organized as follows. We outline the setting in Section \ref{secsetting}. We introduce the estimation and testing strategy in Section \ref{secestimation}. We give the theoretical results in Section \ref{sectheory}. We detail implementation rules in Section \ref{secimplementation}. In Section \ref{secempirical}, we give real data experiments. We provide concluding remarks in Section \ref{secconclusion}. In  Section \ref{secsimulation} from the Appendices, we perform simulation experiments. Additional real data experiments are available in Section \ref{secempiricalstudyinsample} from the Appendices. All the proofs of the theoretical results can be found in Section \ref{secproofs} from the Appendices.

\section{Setting: self-exciting Hawkes processes with parametric kernel}
\label{secsetting}

In this section, we give an introduction to the simple point processes. In particular, we rely on one particular family of simple point processes, namely self-exciting Hawkes processes for which its intensity has a parametric form. In addition, we assume that the kernel has a specific form for estimation based on ordinary least squares. Finally, we introduce a Wald test of $q$ linear hypotheses on the parameter vector $\theta^*=(\nu^*,(\alpha^{*})')'$. 

We start with the definition of the probabilistic tools. First, we define $\reels^+ = [0,\infty)$ as the space of real nonnegative numbers. We consider that the stochastic processes evolve on the space $\reels^+$. We introduce the stochastic basis $\mathbf{B} = \big(\Omega, \sigalg, \filt=\{\filt_t\}_{t \in \reels^+}, \proba\big)$ which is a probability space equipped with a filtration. Namely, we have that $\Omega$ is the sample space, $\sigalg$ is the space of measured sets called $\sigma$-algebra and $\proba$ is a probability measure from the sample space $\Omega$ to the interval $[0,1]$. The filtration $\filt$ is a family of $\sigma$-algebras $\filt_t$ which represents the information available to the statistician at the time $t \in \reels^+$ and satisfies $\filt_t \subset \sigalg$. We denote the natural filtration generated by some stochastic process $X=\{X_t\}_{t \in \reels^+}$ as $\filt^X = \{\filt_t^X\}_{t \in \reels^+}$ in which $\filt_t^X = \mathcal{B} \{ X_s: 0 \le s \le t\}$ for any time $t \in \reels^+$. Here, $\mathcal{B} \{ X \}$ denotes the smallest Borel $\sigma$-algebra generated by the collection of random variables $X$. We assume that $\filt_t^N = \filt_t$ for any time $t \in \reels^+$. We also assume that the stochastic basis $\mathbf{B}$ satisfies the usual conditions. Moreover, we denote by $\esp[X]$ the expected value with respect to the probability measure $\proba$ of a random variable $X$. Finally, a measured set is said to happen almost surely if it happens with probability $1$ with respect to the probability measure $\proba$.

We now give an introduction to the simple point processes. These simple point processes are a special type of counting processes and more generally jump processes. They can be interpreted as random measures. In particular, simple point processes have simple increments which are necessarily equal to unity. Simple point processes deal with the number of occurrences of something over time, which we refer to as events. Examples of simple point processes include Poisson processes and renewal processes. Another example of a simple point process is the number of job arrivals to a queue over time. In this paper, we will consider that the events correspond to financial trades observed in the market. We assume that the simple point processes $N=\{N_t\}_{t \in \reels^+}$ are of dimension $1$. Then, the random variable $N_t$ counts the cumulative number of events between the starting time $0$ and the final time $t$. In the standard literature, a simple point process is a stochastic process $N$ with values that are non-negative, integer, non-decreasing and with increments which are equal to unity which is typically
characterized by the following definition. In what follows, $\ind_C$ is the indicator function from the space $C \subset \reels^+$ to the space $\{0,1\}$ defined as $\ind_C (x) = 1 \text{ if } x \in C$ and $\ind_C (x) = 0 \text{ if } x \notin C$. 

\begin{definition}[simple point process]
\label{defspp}
We say that the stochastic process $N=\{N_t\}_{t \in \reels^+}$ is a simple point process if it is a point process of the particular form
$N_t = \sum_{k \in \naturels} \ind_{[0,t]}(T_k)$. 
 Here, the sequence of event times $\{T_k\}_{k \in \naturels}$ takes its values in the space of nonnegative real numbers $\reels^+$ and is random. Finally, the sequence of event times $\{T_k\}_{k \in \naturels}$ is increasing, namely 
$\proba( T_k <  T_{k+1} \text{ for any } k\in \naturels )=1.$
\end{definition}

As a consequence of Definition \ref{defspp}, we can deduce some properties on the simple point process $N$. More specifically, we can reinterpret the simple point process $N$ as a random measure on the space of nonnegative real numbers $\reels^+$, namely a family 
$\{ N(C) \}_{C \in \mathcal{B}(\reels^+)}$ 
of random variables with values in the space $\overline{\naturels} = \naturels  \cup \{ + \infty \}$. Here, $\naturels = \{ 0, 1, \ldots \}$ is the space of nonnegative integers and $\mathcal{B} ( X )$ denotes the Borel $\sigma$-algebra on the space $X$. Moreover, the random measure $N(C)$ can be expressed as 
$N(C) = \sum_{k \in \naturels} \ind_C(T_k)$ and in particular we have $N([0,t])= N_t$ for any time $t \in \reels^+$. 

We need to make some stronger assumption on the time of the first event $T_0$. We assume that the time of the first event $T_0$ is larger than $0$ and smaller than the real positive number $T_0^+ \in \reels_*^+$ almost surely, namely $\proba(0 < T_0 < T_0^+)=1$. Here, for any space $S$ such that $0 \in S$, we define the space without zero as $S_*$. In practice, this assumption on the simple point process is harmless since this only requires that the time of the first event must belong to a compact space. This is more flexible than imposing that $T_0=0$.

Under some slightly stronger assumptions on the simple point process $N$, there exists an intensity by the Doob–Meyer decomposition theorem (see Theorem I.3.17 (p. 32) in \cite{jacod2003limit}). We now introduce the definition of the intensity $\lambda=\{\lambda_t\}_{t \in \reels^+}$ with respect to the filtration $\filt$ for the simple point process $N$.

\begin{definition}[intensity]
\label{defppint}
Any stochastic process $\lambda=\{\lambda_t\}_{t \in \reels^+}$ defined on the space of real nonnegative numbers $\reels^+$ and satisfying the following properties is called an intensity with respect to the filtration $\filt$ of the simple point process $N$. First, the stochastic process $\lambda$ is adapted to the filtration $\filt$. Secondly, the stochastic process $\lambda$ is of dimension $1$ and takes its values in the space of nonnegative real numbers $\reels^+$.
Moreover, we have for any interval $(t_0, t_1]\subset \reels^+$ that almost surely
\be
\label{defppinteq}
\esp[N((t_0, t_1]) \mid \filt_{t_0}] = \esp\Big[\int_{t_0}^{t_1} \lambda_s ds \mid \filt_a\Big].
\ee  
Here, we implicitly assume that the stochastic process $\lambda$ is such that the integral from Equation (\ref{defppinteq}) is well-defined almost surely for any interval $(t_0, t_1]\subset \reels^+$.
\end{definition}

Intuitively, the intensity corresponds to the expected number of events given the past information. Namely, we have almost surely and for any time $t \in \reels^+$ that
$\lambda_t = \lim_{u \rightarrow 0} \esp\Big[\frac{N_{t+u} - N_t}{u} \Big| \filt_{t^-}\Big].$
Here, $\filt_{t^-}$ denotes the limit from the left at the point $t \in \reels^+$ of the filtration $\filt$. Moreover, we have by the Doob–Meyer decomposition theorem (see Theorem I.3.17 (p. 32) in \cite{jacod2003limit}) that the compensated simple point process $A=\{A_t\}_{t \in \reels^+}$ defined for any time $t \in \reels^+$ as
\begin{eqnarray}
\label{defA}
A_{t} = N_{t} - \int_0^t \lambda_s ds
\end{eqnarray}
is a martingale with respect to the filtration $\filt$ almost surely. Then, the intensity $\lambda$ measures the rate of change for the predictable part of the simple point process $N$ based on the martingale representation (\ref{defA}). Finally, we note that $N((t_0, t_1])$ is almost surely finite if and only if $\int_{t_0}^{t_1} \lambda_s ds$ is almost surely finite. For a background on simple point processes, the reader can consult \cite{jacod1975multivariate}, \cite{jacod2003limit},  \cite{daley2003introduction} and \cite{daley2008introduction}.

The present work is concerned with the simple point process $N$ admitting an intensity which has a parametric form. More specifically, we introduce the parameter space $\Theta$ which consists of $n$ parameters. We also introduce the family of intensities $\lambda_t (\theta)$ for any parameter $\theta \in \Theta$ and any time $t \in \reels^+$. We assume that the intensity process $\lambda_t (\theta)$ takes its value in the space of nonnegative real numbers $\reels^+$ for any parameter $\theta \in \Theta$, any $\omega \in \Omega$ and any time $t \in \reels^+$. Finally,
we assume the existence of the true vector parameter $\theta^* \in \Theta$ of dimension $n$ such that for any time $t \in \reels^+$ and any $\omega \in \Omega$ we have $\lambda_t = \lambda_t (\theta^*)$.

From the introduced class of simple point processes, we focus on Hawkes self-exciting processes (see \cite{hawkes1971spectra} and \cite{hawkes1971point}) in which the kernel is parametric. More specifically, we introduce for any parameter $\theta \in \Theta$ and any time $t \in \reels^+$ the family of intensities 
\be \label{defhawkesintensityfamily}
   \lambda_t (\theta) = \nu  + \int_0^t  h(t-s,\alpha)\,dN_{s}.
\ee
Here, the function $h$ is the exciting kernel of dimension $1$. Moreover, the baseline parameter $\nu$ is of dimension $1$ while the kernel parameter $\alpha$ is a vector of dimension $n_h$. We assume that the parameter $\theta$ has the form $\theta=(\nu,\alpha')'$ and belongs to the parameter space $\Theta=(\Theta_\nu,\Theta_\alpha)$. We also assume that the number of parameters satisfies $1+n_h \geq 2$.
Finally,
we assume the existence of the true vector parameter $\theta^* \in \Theta$ of dimension $n$ such that
\begin{equation}\label{defhawkesintensity}
   \lambda_t = \lambda_t (\theta^*) = \nu^*  + \int_0^t  h(t-s,\alpha^*)\,dN_{s}.
\end{equation}
Here, the true parameter $\theta^*$ has the form $\theta^*=(\nu^*,(\alpha^{*})')'$ in which $\nu^* \in \Theta_\nu$ and $\alpha^* \in \Theta_\alpha$. In Equation (\ref{defhawkesintensityfamily}), the Hawkes processes are self-exciting in the sense that the realization of an event raises the probability of a future realization. In particular, this accommodates the presence of event clustering in time, which is the main stylized fact observed with financial data.

The goal of this paper is to retrieve the true parameter $\theta^*=(\nu^*,(\alpha^{*})')'$. To apply least square estimation, we restrict to the kernel which has the following form
\begin{equation} \label{defkernel}
   h(t,\alpha) = \sum_{k=1}^{n_h} \alpha_k h_k(t).
\end{equation}
Here, we denote by $n_h$ the number of functions in the kernel. So, there are $n=1+n_h$ parameters in the parametric model by definition. In addition, we have that $\alpha=(\alpha_1, \ldots, \alpha_{n_h})'$ in which $\alpha_k \in \Theta_{\alpha,k}$ for any index $k=1,\ldots,n_h$ and $\Theta_{\alpha} = \Theta_{\alpha,1} \times \ldots \times \Theta_{\alpha,n_h}$. In the theoretical analysis, we assume that the number of functions $n_h$ is fixed and known to the statistician. The case in which the number of functions $n_h$ is unknown is above the scope of this paper as it uses the parametric estimation procedure based on maximum likelihood. Nonetheless, we detail the implementation rules to choose the number of functions $n_h$ in Section \ref{secimplementation} for applications when these quantities are unknown to the statistician.

Equation (\ref{defkernel}) introduces heterogeneity in the kernel. This heterogeneity is observed on trades from financial market as we have empirical evidence that there are different types of traders (see \cite{gagnon2010multi} , \cite{Hasbrouck2013low}, \cite{hoffmann2014dynamic}, \cite{budish2015high}, \break \cite{biais2015equilibrium}, \cite{foucault2016news} and \cite{pagnotta2018competing}). The reason why we restrict the class of kernels is due to the parametric estimation procedure based on ordinary least squares. Then, the form (\ref{defkernel}) is the most general expression we could deal with in our theoretical analysis. In addition, we assume that the functions $h_k$ are exponential. This is the kernel considered in \cite{bacry2020sparse} (see Equations (1) and (2)) who study parametric estimation of Hawkes processes with penalized least squares. A more general form of the kernel is considered in \cite{reynaud2010adaptive}, \cite{hansen2015lasso} (see Section 1.2.3) and \cite{cai2024latent} (see p. 98) for nonparametric estimation of Hawkes processes based on penalized least squares. Investigating if this kernel form is allowed is above the scope of our paper. See also \cite{kirchner2017estimation} and \cite{fujimori2026sparse} who consider nonparametric estimation of Hawkes processes based on conditional least squares.

In addition, we introduce a Wald test $S$ of $q$ linear hypotheses on the parameter vector $\theta^*=(\nu^*,(\alpha^*)')'$. This Wald test $S$ is based on the matrix $R$ of dimension $q \times n$. Namely, we define for a real number $r \in \reels$ the null hypothesis as $H_0: \{ R \theta^*=r\}$ and the alternative hypothesis as 
$H_1: \{ R \theta^*\neq r \}.$ In general, the Wald test assesses constraints on parameters based on the weighted distance between the unrestricted estimate and its value under the null hypothesis. Intuitively, the larger this weighted distance, the less likely it is that the constraint is true. This test is very similar to the Wald test proposed in \cite{erdemlioglu2025latency}.

\section{Parametric estimation  based on ordinary least squares}
\label{secestimation}

This section details the parametric estimation procedure for Hawkes processes based on ordinary least squares. For that purpose, we rewrite the continuous-time Hawkes process as a standard linear regression model with time series. We consider the standard linear regression relationship with time series $y=\{y_{t\Delta}\}_{ t \in \naturels_* }$ for any time index $t\in \naturels_*$, namely
\be
\label{deflr}
y_{t\Delta} =  x_{t\Delta}' \theta^{*} + u_{t\Delta}.
\ee
Here, the time series $y=\{y_{t\Delta}\}_{ t \in \naturels_* }$ is the sequence of observed random variable of dimension $1$ which is adapted to the discrete-time filtration $\filt^d = \{ \filt_{t\Delta} \}_{ t \in \naturels_* }$. Also, the time series  $x=\{x_{t\Delta}\}_{ t \in \naturels_* }$ is the sequence of the explaining random vector of dimension $n$ and adapted to the discrete-time filtration $\filt^d$. Moreover, the time series $u=\{u_{t\Delta}\}_{ t \in \naturels_* }$ is the sequence of the error term which is a random variable of dimension $1$ adapted to the discrete-time filtration $\filt^d$. In Equation (\ref{deflr}), we have that the time series consists of observations recorded at successive equally spaced points in time. Thus, it represents a form of discrete-time data. Since the original observations come from a continuous-time stochastic process, the time between two consecutive observations $\Delta \in \reels^+_*$ is a tuning parameter to be set by the statistician. Thus, this is a crucial step and a limitation of our novel approach. We detail implementation rules for the time increment  $\Delta$ in Section \ref{secimplementation}. However, we insist that the theory of this paper holds for any time increment $\Delta \in \reels^+_*$ and does not affect the asymptotic variance matrix. In matrix notation, this relationship
is written as $Y_{T\Delta} = X'_{T\Delta} \theta^{*} + U_{T\Delta}$. Here, we introduce the vector $Y_{T\Delta}=(y_{\Delta}, y_{2\Delta}, \ldots, y_{T\Delta})'$ of dimension $T$, the matrix $X_{T\Delta}=(x_{\Delta}, x_{2\Delta} \ldots, x_{T\Delta})$ of dimension $n \times T$ and the vector $U_{T\Delta}=(u_{\Delta}, u_{2\Delta}, \ldots, u_{T\Delta})'$ of dimension $T$. The main idea of this paper is to rewrite the continuous-time Hawkes process $N$ as a standard linear regression model with time series of the form (\ref{deflr}). More specifically, we first set the observed variable for any time index $t\in \naturels_*$ as $y_{t\Delta}= N_{(t\Delta)^-} -N_{(t-1)\Delta}$. Here, $N_{(t\Delta)^-}$ denotes the limit from the left at the point $t\Delta$ of the stochastic process $N$. Since the Hawkes process $N$ is not stationary, it is natural to differentiate it to obtain a time series representation. We set the explaining vector $x_{t\Delta}$ for any time index $t\in \naturels_*$ as
\be
\label{defx}
x_{t\Delta}=\Big(1, \int_{(t-1)\Delta}^{t\Delta} h_1(t-s) dN_s, \ldots, \int_{(t-1)\Delta}^{t\Delta} h_{n_h}(t-s) dN_s \Big)'.
\ee
The reason for that choice is the intensity form (\ref{defhawkesintensity}) and its martingale representation (\ref{defA}). In practice, we have that the explaining vector $x_{t\Delta}$ for any time index $t\in \naturels_*$ can be re-expressed as 
$x_{t\Delta} =\Big(1, \sum_{(t-1)\Delta \leq T_i < t\Delta}h_1(t-T_i), \ldots, \sum_{(t-1)\Delta \leq T_i < t\Delta}h_{n_h}(t-T_i) \Big)'.$
Thus, the explaining vector $x_{t\Delta}$ corresponds to a sum of a discrete and random number of terms, namely the number of events on the interval $[(t-1)\Delta , t\Delta)$. Then, we define the error term $u_{t\Delta}$ for any time index $t\in \naturels_*$ as
$u_{t\Delta} = y_{t\Delta} - x_{t\Delta}'\theta^*$. We propose to estimate the unknown parameter $\theta^*$ by minimizing, for any time index $T\in \naturels_*$, the ordinary least squares objective function
\be
\label{defSSR}
SSR_{T\Delta}(\theta)= \sum_{t=1}^T \big(y_{t\Delta} - x_{t\Delta} ' \theta \big)^2 = \big( Y_{T\Delta} - X'_{T\Delta} \theta \big) ' \big( Y_{T\Delta} - X'_{T\Delta} \theta \big).
\ee
More specifically, the least squares estimator is defined as the solution to the minimization problem (\ref{defSSR}) between the starting time $0$ and the final time $T\Delta$, namely
$\widehat{\theta}_{T\Delta} = \operatorname{argmin}_{\theta \in \Theta} SSR_{T\Delta} (\theta).$
The first-order conditions for a minimum of the objective function $SSR_{T\Delta}(\theta)$ are
$\frac{\partial SSR_{T\Delta}(\theta)}{
\partial \theta} = -2 \sum_{t=1}^T x_{t\Delta} \big(y_{t\Delta} - x_{t\Delta} ' \theta \big) = -2 X_{T\Delta} \big( Y_{T\Delta} - X'_{T\Delta} \theta \big).$
If the matrix $X_{T\Delta} X_{T\Delta}' = \sum_{t=1}^T x_{t\Delta} x_{t\Delta} '$ is non singular, this system of $n$ equations in $n$
unknowns can be uniquely solved for the ordinary least squares 
estimator
$\widehat{\theta}_{T\Delta} = \Big( \sum_{t=1}^T x_{t\Delta} x_{t\Delta}' \Big)^{-1} \sum_{t=1}^T x_{t\Delta} y_{t\Delta} = \big( X_{T\Delta} X'_{T\Delta} \big)^{-1} X_{T\Delta} Y_{T\Delta}.$
This is the objective function and the estimator discussed in \cite{white2001asymptotic} (p. 2).

We now introduce the quantities required for the asymptotic variance-covariance matrix. They correspond to the quantities given in \cite{white2001asymptotic} (Theorem 5.17, p. 126). First, we define the variance matrix $V_{T\Delta}$ at the final time $T\Delta$ of dimension $n \times n$ for any time index $T\in \naturels_*$ as $V_{T\Delta} = \var \big((T \Delta)^{-1/2} X_{T\Delta}U_{T\Delta}\big)$.
Here, $\var(v)$ denotes the variance-covariance matrix of dimension $p \times p$ for a random vector $v$ of dimension $p$. Then, we define its limiting variance matrix when the time index $T \rightarrow \infty$ as
\be
\label{defV}
V = \lim_{T \rightarrow \infty} V_{T\Delta}.
\ee
We also define the matrix $M_{T\Delta}$ at the final time $T\Delta$ of dimension $n \times n$ for any time index $T\in \naturels_*$ as
$M_{T\Delta} = \esp [x_{T\Delta} x_{T\Delta}']$.
Then, we define its limiting matrix when the time index $T \rightarrow \infty$ as $M = \lim_{T \rightarrow \infty} M_{T\Delta}$. Finally, we define the asymptotic variance matrix $D$ as $D=M^{-1} V M^{-1}$.

We also introduce the quantities required for estimation of the asymptotic variance-covariance matrix. First, the error term $u_{t\Delta}$ is not observable for any time index $t\in \naturels_*$.
However, it can be estimated for any time index $t\in \naturels_*$ by $\widehat{u}_{t\Delta} = y_{t\Delta} - x_{t\Delta}'\widehat{\theta}_{T\Delta}$. Secondly, we define the estimator of the variance matrix $V_{T\Delta}$ at the final time $T\Delta$ for any time index $T\in \naturels_*$ as
\be
\label{defVThat}
\widehat{V}_{T\Delta} & = & w_{T,0} T^{-1} \sum_{t=1}^T x_{t\Delta} \widehat{u}_{t\Delta} \widehat{u}_{t\Delta}' x_{t\Delta}' \\ \nonumber && + T^{-1} \sum_{\tau=1}^{m_T} w_{T,\tau} \sum_{t=\tau+1}^T \big( x_{t\Delta} \widehat{u}_{t\Delta} \widehat{u}_{(t-\tau)\Delta}' x_{(t-\tau)\Delta}' + x_{(t-\tau)\Delta} \widehat{u}_{(t-\tau)\Delta} \widehat{u}_{t \Delta}' x_{t \Delta}' \big).
\ee
This corresponds to the variance estimator $\widetilde{V}_n$ given in \cite{white2001asymptotic} (p. 156). The estimation procedure requires the introduction of the deterministic sequence $m_T \in  \naturels_*$. The reason is that the terms satisfying $m_T < \tau < T$ will
be small in absolute value if $m_T$ is sufficiently large and thus can be neglected. If the sequence $m_T$ is simply kept fixed as the final index $T$ grows, the number of neglected terms may grow in such a way that the sum of the neglected terms does not remain negligible. This suggests that the sequence $m_T$ will have to grow with the time index $T$, so that
the terms in the variance $V_{T\Delta}$ ignored by its estimator $\widehat{V}_{T\Delta}$ remain negligible. We will write $m_T$ to
make this dependence explicit. In addition, we have introduced deterministic weights $w_{T,\tau} \in \reels^+_*$ for any index $\tau=0,1,\ldots,m_T$ in the variance estimator $\widehat{V}_{T\Delta}$. The reason is that the variance estimator $\widehat{V}_{T\Delta}$ can
fail to be positive definite and can be even indefinite. This is clearly inconvenient as negative variance
estimates can lead to results that are useless for testing hypotheses
or for any other purpose where variance estimates are required. We detail implementation rules for the variance estimator $\widehat{V}_{T\Delta}$ in Section \ref{secimplementation}. Finally, we define the estimator of the asymptotic variance matrix $D$ at the final time $T\Delta$ of dimension $n \times n$ for any time index $T\in \naturels_*$ as $\widehat{D}_{T\Delta} = \big( X_{T\Delta} X'_{T\Delta}/T \big)^{-1} \widehat{V}_{T\Delta} \big( X_{T\Delta} X'_{T\Delta}/T \big)^{-1}$. This corresponds to the estimator $\widehat{D}_n$ given in \cite{white2001asymptotic} (Theorem 5.17, p. 126).

Since we have proposed estimation of the asymptotic variance matrix $D$, we can define our Wald test statistic at the final time $T\Delta$ for any time index $T\in \naturels_*$ as
\begin{eqnarray}
\label{defW}
    S_{T\Delta} =
T\big(R \widehat{\theta}_{T\Delta} -r\big)' \big(R\widehat{D}_{T\Delta} R' \big)^{-1}\big(R\widehat{\theta}_{T\Delta}-r\big).
\end{eqnarray}
The Wald test statistic $S_{T\Delta}$ relies on two approximations, namely the variance matrix estimator $\widehat{D}_{T\Delta}$ and the parameter estimator $\widehat{\theta}_{T\Delta}$.

\section{Asymptotic theorems}
\label{sectheory}
Since we have rewritten the continuous-time Hawkes process as a standard linear regression model with time series in Section \ref{secestimation}, we can use asymptotic theory designed for linear regression model based on time series. More specifically, we apply the technology of \cite{white2001asymptotic} (see Theorem 5.17, p. 126). In this section, we first give the consistency and central limit theorem for parametric estimation of Hawkes processes based on ordinary least squares (see Theorem \ref{thclt}). Then, we show the consistency of the variance matrix estimation procedure (see Theorem \ref{thconsistency}). In addition, we establish the central limit theorem of the parametric estimation procedure when we estimate the variance matrix (see Theorem \ref{thcltest}). Finally, we show that the Wald test statistic $S_{T\Delta}$ converges in distribution to a chi-squared distribution with $q$ degrees of freedom under the null hypothesis $H_0$ and is consistent under the alternative hypothesis $H_1$ (see Corollary \ref{corwald}).
To the best of our knowledge, there is currently no existing work using ordinary least squares for parametric estimation of Hawkes processes.
We first give a condition required on the Hawkes process to establish the central limit theorem for parametric estimation based on ordinary least squares.

\begin{condition} 
\label{condh}
\begin{enumerate}[label=(\alph*)]

\item \label{condhparameterspace} There exist constants $\nu_- \in \reels^+_*$ and $\nu_+ \in \reels^+_*$ satisfying $\nu_-< \nu < \nu_+$ for any baseline parameter $\nu \in \Theta_\nu$. There exist constants $\alpha_- \in \reels^+_*$ and $\alpha_+ \in \reels^+_*$ satisfying $\alpha_-< \alpha_k < \alpha_+$ for any index $k = 1, \ldots ,n_h$ and any parameter $\alpha_k \in \Theta_{\alpha,k}$.

\item \label{condhhexp} For any index $k = 1, \ldots ,n_h$ and any time $t\in \reels^+$, there exists an inverse scale number $b_k^* \in \reels^+_*$ satisfying $h_k(t) = \exp(-b_k^* t)$.

\item \label{condhidentifiability} For any index $k = 1, \ldots ,n_h$ and any index $j = 1, \ldots ,n_h$ satisfying $k \neq l$, we have that the inverse scale numbers are unequal to each other, namely $b_k^* \neq b_l^*$. 

\item \label{condhstationary} There exists a constant $\rho_+ \in \reels^+_*$ satisfying $\rho_+ < 1$ and $\sum_{k=1}^{n_h} \frac{\alpha_k}{b_k^*} < \rho_+$ for any parameter $\alpha \in \Theta_\alpha$. 

\end{enumerate}
\end{condition}

Condition \ref{condh} \ref{condhparameterspace} implies that there exists a positive constant $\nu_- \in \reels^+_*$ satisfying $\nu > \nu_-$ for any baseline parameter $\nu \in \Theta_\nu$ and thus the simple point processes are well-defined. This is also required in the simpler case of heterogeneous Poisson processes without a kernel (see \cite{daley2003introduction}). More generally, Condition \ref{condh} \ref{condhparameterspace} assumes that the parameters are uniformly bounded away from $0$ and bounded. 

In Condition \ref{condh} \ref{condhhexp}, we assume that the functions $h_k$ are exponential functions in which the inverse scale number $b_k^*$ are known to the statistician. Thus, the inverse scale numbers $b_k^*$ are not parameters of the model to be estimated. However, the inverse scale numbers $b_k^*$ may be unknown in practice. In such a situation, we use parametric estimation based on maximum likelihood. This is especially useful as the inverse scale numbers $b_k^*$ are used to choose the time increment $\Delta$ (see Section \ref{secimplementation}). Moreover, the kernel $h$ is a sum of exponential functions and this introduces heterogeneity in the kernel. Most of the works on parametric estimation for Hawkes processes assume that the kernel function $h$ is exponential. The main reason is that Hawkes processes are more tractable in that particular case where the intensity $\lambda$ is a Markov process. Then, the proofs rely heavily on the Markov property of the exponential distribution (see \cite{clinet2017statistical}). However, we are slightly more flexible by allowing the kernel to be the sum of exponential kernels but we are not as flexible as \cite{potiron2026mutually} and \cite{potiron2026hawkes}. 

Condition \ref{condh} \ref{condhidentifiability} is naturally required for identifiability purposes. Without Condition \ref{condh} \ref{condhidentifiability}, we cannot retrieve the parameters $\alpha_k^*$ and $\alpha_j^*$ when $b_k^* = b_j^*$. We define the $L^1$ norm for the kernel function $h$ of dimension $1$ as $\| h \|_1 = \int_0^{\infty} h(s) ds$.
Condition \ref{condh} \ref{condhstationary} states that the kernel integral $\| h \|_1$ is strictly smaller than unity. This is exactly the condition which is necessary to obtain a stationary intensity with finite first moment (see Lemma 1 (p.\ 495) in \cite{hawkes1974cluster} and Theorem 1 (p.\ 1567) in \cite{bremaud1996stability}).

In the theorem that follows, we state the consistency and the
central limit theorem of the parametric estimation procedure for Hawkes processes in which its intensity has a parametric form. The parametric estimation procedure is based on ordinary least squares. We consider asymptotics when the final time diverges to infinity, namely $T \rightarrow + \infty$. As expected, we obtain the convergence rate $\sqrt{T\Delta}$ which is affected by the time increment $\Delta$. The main idea of the proofs is to apply the technology of \cite{white2001asymptotic} (Theorem 5.17, p. 126). Hereafter, $\xi$ is defined as a standard normal vector of dimension $n$. Finally, we denote the standard convergence in distribution of the family of random vectors $X_T$ to the random vector $Y$ as the final time $T \rightarrow +\infty$ with respect to the probability measure $\proba$ by $X_T \cvdistrib Y$.

\begin{theorem}[Central limit theorem]
\label{thclt}
We assume that Condition \ref{condh} holds. Then, we have the central limit theorem of the parametric estimation procedure based on ordinary least squares as the final time $T \rightarrow + \infty$, namely $D^{-1/2} \sqrt{T\Delta}(\widehat{\theta}_{T\Delta} - \theta^*) \cvdistrib  \xi$.
\end{theorem}

We now give a condition required for the consistency of the variance matrix estimator $\widehat{V}_{T\Delta}$.

\begin{condition} 
\label{condconsistency}
\begin{enumerate}[label=(\alph*)]

\item \label{condconsistencym} We have the sequence $m_T \rightarrow +\infty$ and $m_T T^{-1/4} \rightarrow 0$ as the final time $T \rightarrow +\infty$.

\item \label{condconsistencyweights}We have $w_{T,\tau} < C_\delta$ in which $C_\delta \in \reels^+$ is a deterministic constant for any time index $T\in \naturels_*$ and any index  $\tau \in \naturels_*$. We also have $w_{T,\tau} \rightarrow 1$ as the final time $T \rightarrow +\infty$ for any index $\tau \in \naturels_*$.

\end{enumerate}
\end{condition}

In Condition \ref{condconsistency} \ref{condconsistencym}, we have that the deterministic sequence $m_T$ is growing with the final time index $T$, namely we have $m_T T^{-1/4} \rightarrow 0$. Thus, the growing rate required is explicitly given.
With that growing rate
the terms in the variance $V_{T\Delta}$ ignored by its estimator $\widehat{V}_{T\Delta}$ remain negligible. In Condition \ref{condconsistency} \ref{condconsistencyweights},
we require $w_{T,\tau} \rightarrow 1$ as the final time $T \rightarrow +\infty$ for any index $\tau \in \naturels_*$ for consistency of the variance matrix estimation procedure. This is natural given how the weights appear in Definition (\ref{defVThat}). For simplicity,
we consider only nonstochastic weights in what follows. Stochastic
weights can be treated straightforwardly in a similar manner.

In the following theorem, we show the consistency of the variance matrix estimation procedure. We consider asymptotics when the final time diverges to infinity, namely $T \rightarrow + \infty$. The main idea of the proofs is to apply the technology of \cite{white2001asymptotic} (Theorem 6.21, p. 159).  Finally, we denote the standard convergence in probability of the family of random vectors $X_T$ to the random vector $Y$ as the final time $T \rightarrow +\infty$ with respect to the probability measure $\proba$ by $X_T \cvproba Y$.

\begin{theorem}[Variance estimator consistency]
\label{thconsistency}
We assume that Conditions \ref{condh} and \ref{condconsistency} hold. Then, we have the variance estimator consistency of the parametric estimation procedure based on ordinary least squares as the final time $T \rightarrow + \infty$, namely $\widehat{V}_{T\Delta} - V_{T\Delta} \cvproba 0$.
\end{theorem}

In the theorem that follows, we establish the consistency and central limit theorem of the parametric estimation procedure when we estimate the variance matrix. As in Theorem \ref{thclt}, we obtain the same convergence rate $\sqrt{T\Delta}$ which is affected by the time increment $\Delta$. We consider asymptotics when the final time diverges to infinity, namely $T \rightarrow + \infty$. The main idea of the proofs is to apply the technology of \cite{white2001asymptotic} (Theorem 5.17, p. 126) with Theorem \ref{thclt} and Theorem \ref{thconsistency}. 

\begin{theorem}[Central limit theorem with variance estimator]
\label{thcltest}
We assume Conditions \ref{condh} and \ref{condconsistency} hold. Then, we have the central limit theorem with variance estimator of the parametric estimation procedure based on ordinary least squares as the final time $T \rightarrow + \infty$, namely $\widehat{D}_{T\Delta}^{-1/2} \sqrt{T\Delta}(\widehat{\theta}_{T\Delta} - \theta^*) \cvdistrib  \xi$.
\end{theorem}

Now, we consider the Wald test based on the parameter $\theta^*$. The following corollary shows that the Wald test statistic $S_{T\Delta}$ converges in distribution to a chi-squared distribution with $q$ degrees of freedom under the null hypothesis $H_0$ and is consistent under the alternative hypothesis $H_1$. We consider asymptotics when the final time diverges to infinity, namely $T \rightarrow + \infty$. This result complements Corollary 2 in \cite{erdemlioglu2025latency}. The difference is that we use least square estimation while \cite{erdemlioglu2025latency} relies on maximum likelihood estimation. The proof of the convergence under the null hypothesis $H_0$ is obtained by applying the delta method (see Theorem 20.8 (p. 297) in \cite{van1998asymptotic}) to Theorem \ref{thclt} and applying Slutsky's Theorem to Theorem \ref{thconsistency}. We define $Q(p)$ as the quantile function of the chi-squared distribution with one degree of freedom for any probability $p \in (0,1)$. Namely, the quantile function $Q$ is the inverse of the cumulative distribution function.

\begin{corollary}
\label{corwald}
We assume that Conditions \ref{condh} and \ref{condconsistency} hold. Then, the test statistic $S_{T\Delta}$ converges in distribution to a chi-squared random variable with $q$ degrees of freedom under the null hypothesis $H_0$ as the final time $T \rightarrow + \infty$. The test statistic $S_{T\Delta}$ is also consistent under the alternative hypothesis  $H_1$, namely we have $\proba ( S_{T\Delta} > Q(p) \mid H_1 ) \rightarrow 1$ for any probability $p \in (0,1)$ as the final time $T \rightarrow + \infty$.
\end{corollary}

\section{Practical implementation}\label{secimplementation}

In this section, we discuss the practical implementation of the parametric estimation procedure based on ordinary least squares. We first describe the discrete-time linear regression used for the parametric estimation. Then, we consider the choice of the time increment $\Delta$ and the number of functions $n_h$ in the kernel. Finally, we specify the number of terms in the sum $m_T$ and the weights $w_{T,\tau}$ appearing in the estimator of the variance-covariance matrix $V_{T\Delta}$.

Our estimation procedure requires discretizing the Hawkes process with the time increment $\Delta$. For that purpose, we introduce the response variable which is the event rate $\widetilde{y}_{t}=\frac{y_{t\Delta}}{\Delta}$ for any time index $t = 1,\ldots, T$. We also define the discrete variable $h_{t,k}=\sum_{s<t}y_{s\Delta} \exp(-b^*_k\Delta(t-s))$ for each exponential function $k=1,\ldots,n_h$.
This discrete variable is updated recursively by $h_{t+1,k}=\exp(-b^*_k\Delta)h_{t,k}+y_{t\Delta}$ in which $h_{1,k}=0$. Then, the parametric estimation procedure reduces to the linear regression $\widetilde{y}_t = \nu+\sum_{k=1}^{n_h} \alpha_k h_{t,k}+u_{t\Delta}$ for any time index $t = 1,\ldots, T$. Moreover, we solve 
\begin{equation}\label{stat_crit}
\widehat{\theta}_{T\Delta} = \underset{(\nu,\alpha_1,\cdots,\alpha_{n_h}) \in \Theta}{\arg\min}\sum_{t=1}^T\left(\widetilde{y}_t-\nu-\sum_{k=1}^{n_h}\alpha_k h_{t,k}\right)^2, \; \Theta = \Bigg\{ \nu_-<\nu<\nu_+, \alpha_-< \alpha_k < \alpha_+, \sum_{k=1}^{n_h} \alpha_k/b^*_k < \rho_+\Bigg\}.
\end{equation}
Finally, we will set $\nu_-=10^{-8}$, $\nu_+=100$, $\alpha_-=10^{-8}$, $\alpha_+=b^*_{\max}$ in which $b^*_{\max}=\max_{k=1,\ldots,n_h} b^*_k$ and $\rho_+=0.99$ throughout our numerical experiments. In particular, this framework satisfies Condition \ref{condh} from the theory. In practice, one must select a convenient value for the time increment $\Delta$ from which an estimator $\widehat{\theta}_{T\Delta}$ based on ordinary least squares is deduced. Ideally, the selected $\Delta$ should provide an estimator that should be as close as possible to the true parameter $\theta^*$. A small value of the time increment $\Delta$ provides a more accurate time series approximation of the continuous-time Hawkes process. However, a large value of the time increment $\Delta$ leads to a coarser approximation and may deteriorate the estimation accuracy. Rather than selecting the time increment $\Delta$ directly, we propose to parameterize it as $\Delta=c/b^*_{\max}$ in which $c$ is a tuning parameter. By parametrization, we mean that the inverse scale numbers $b^*_k$ are now considered as parameters from the model. This choice is motivated by the fact that the normalization expresses the size of the time increment relative to the shortest characteristic time scale of the Hawkes process, namely $c/b^*_{\max}$. Therefore, the same values of the tuning parameter $c$ correspond to comparable discretization levels across Hawkes models with substantially different decay numbers. Hence, this normalization provides a common and interpretable tuning scale across different model specifications. In settings when the true parameter $b^*$ is unknown, we will replace $b^*$ with its maximum likelihood estimator $\widehat{b}^{\text{mle}}_{T\Delta}$. Namely, the procedure relies on the log likelihood process defined for any parameter $(\theta,b) \in \Theta \times \Theta_b$ and any time index $T \in \naturels_*$ as
$l_{T\Delta}(\theta,b) =  \int_0^{T\Delta} \log (\lambda_t(\theta,b)) dN_t - \int_0^{T\Delta} \lambda_t (\theta,b) dt$.
Then, the maximum likelihood estimator is defined as a maximizer of the log likelihood process, namely
$\big(\widehat{\theta}_{T\Delta}^{\text{mle}},\widehat{b}^{\text{mle}}_{T\Delta}\big) \in \operatorname{argmax}_{\theta \in \Theta, b \in \Theta_b} l_{T\Delta} (\theta,b)$. Parametric estimation of Hawkes processes based on maximum likelihood was introduced in \cite{ogata1978asymptotic} and recently used in \cite{potiron2026mutually}. Finally, we will estimate the time increment by $\widehat{\Delta} = \widehat{c}/\widehat{b}^{\text{mle}}_{\max,T\Delta}$ in which $\widehat{b}^{\text{mle}}_{\max,T\Delta} = \max_{k=1,\ldots,n_h} \widehat{b}^{\text{mle}}_{k,T\Delta}$. Since the number of functions $n_h$ of the kernel is also unknown in practice, we propose selecting both the tuning parameter $c$ and $n_h$ in a data-driven way. More specifically, we use the Akaike Information Criterion (AIC) based on the log likelihood process. For a given pair $(c,n_h)$, AIC is computed as
\be
\label{defaic}
\text{AIC}\big(\widehat{\theta}_{T\Delta},\widehat{b}^{\text{mle}}_{T\Delta}\big) = 2(2n_h+1)+2 l_{T\Delta}\big(\widehat{\theta}_{T\Delta},\widehat{b}^{\text{mle}}_{T\Delta}\big).
\ee
Here, the optimal pair $(\widehat{c},\widehat{n}_h)$ is chosen by $(\widehat{c},\widehat{n}_h)=\arg\min_{(c,n_h) \in \mathcal{C} \times \naturels_*}\text{AIC}\big(\widehat{\theta}_{T\Delta},\widehat{b}^{\text{mle}}_{T\Delta}\big)$. 
Throughout this paper, we consider a fixed logarithmically equally spaced grid of candidate values for the time increment parameter $c$, from which the corresponding candidate values of the time increment $\Delta$ are obtained through the aforementioned scaling. More precisely, we will work in Section \ref{secempirical} with the grid $c \in \mathcal{C} = \{10^{-6.5},\ldots,10^2\}$ with 96 logarithmically equally spaced grid points. The sensitivity of the estimation accuracy with respect to the time increment $\Delta$ is thoroughly assessed in Subsection \ref{appendix_subsec:sim_l2error} of the Appendices through simulation experiments.

Finally, the estimator of the variance-covariance matrix $\widehat{V}_{T\Delta}$ from Definition (\ref{defVThat}) requires specifying the number of terms in the sum $m_T$ and the weights $w_{T,\tau}$. Throughout this paper, we follow Section 6.3 of \cite{white2001asymptotic} and employ exponentially decaying weights $w_{T,\tau}=\exp(-\gamma \tau)$ in which $\gamma = 0.01$. Moreover, Condition~\ref{condconsistency}\ref{condconsistencym} requires $m_T T^{-1/4} \rightarrow 0$ as the final time $T \rightarrow +\infty$. Therefore, we set $m_T$ as the nearest integer to $T^{1/4-0.01}$ throughout this paper.

The code for implementation and the numerical experiments detailed hereafter are publicly available in the Github repository: \url{https://github.com/Benjamin-Poignard/parametricHawkes_OLS}. 

\section{Real data experiments}
\label{secempirical}

In this section, we apply our parametric estimation procedure based on ordinary least squares to real financial high-frequency data. We first investigate the prediction performances of the number of transactions based on ordinary least squares and maximum likelihood procedures. Overall, we slightly improve out-of-sample predictions. Our empirical results on predictions also provide evidence for three heterogeneous trader types. In particular, we identify two trader types as high-frequency traders and one trader type as fundamental trader. Finally, we consider a statistical Wald test to confirm the results on heterogeneity of trader types.

The dataset is obtained on  Wharton Research Data Services which is generated from Trades and
Quotes datasets. The dataset contains information about transactions of stock prices and the transaction times. In particular, we restrict to the transactions that correspond to a trade. This means that we rely solely on the trade times and do not use any information about prices of the assets. By duration, we mean the time increment between two consecutive trades in
the dataset. We analyze intraday trade data for Apple (AAPL), Facebook (META), and EOG Resources, Inc. (EOG). AAPL and META are traded on the Nasdaq Stock Exchange, while EOG is traded on the New York Stock Exchange. The sample spans January 3 to December 29, 2017, excluding July 3 and November 24, which were early-closing trading days with shortened regular trading sessions, leaving a total of 249 trading days. We only consider trades executed between 9:30 am and 3:30 pm to avoid opening and closing effect. Table \ref{sumstat} reports the summary statistics about durations of the three stocks.

\begin{table}[H]
\caption{Summary statistics (January 2017 - December 2017).}\label{sumstat}
\begin{centering}
\begin{tabular}{lccccc}
\hline \hline
  Asset        & Observations & Mean (seconds) & Std dev (seconds) & Skewness & Kurtosis \tabularnewline \hline


AAPL trade durations & 31,309,451 & 0.172 & 0.455 & 5.04 & 42.49 \\
META trade durations  & 20,714,199 & 0.260 & 0.695 & 5.08 & 42.72 \\
EOG trade durations   & 4,877,706  & 1.102 & 3.260 & 5.60 & 51.43 \\\hline \hline
\end{tabular}
\par\end{centering}
\end{table}

\subsection{Out-of-sample prediction}\label{subsec:oos}

In this subsection, we investigate the prediction performances of the number of transactions based on ordinary least squares and maximum likelihood procedures. We consider the prediction of the number of transactions for AAPL, META and EOG at an $s$-minute horizon with $s\in\{1,5,10\}$. The prediction starts from $T_0=12600$ in which $T_0$ represents the 12600-th second after 9:30 am, namely 1:00 pm. We define the sequence of prediction origins $\tau_j=T_0+(j-1)60s$ for any index $j=1,\ldots,J$. Here, $J$ is the largest integer such that $\tau_j+60s\leq T$ with $T=21600$ corresponding to 3:30 pm. For each prediction origin $\tau_j$, the prediction interval is $[\tau_j,\tau_j+60s]$. In addition, the in-sample training interval is $I_j=[\tau_j-W,\tau_j]$ in which $W=7200$ seconds. This corresponds to the previous two hours of trades. Thus, the in-sample trade times are $\mathcal{T}_j =\{T_i:\tau_j-W\leq T_i<\tau_j\}$. The realized future trade count is $Y_{j,s} = N_{\tau_j+60s}-N_{\tau_j}$.  To distinguish between the parameters estimated by maximum likelihood and ordinary least squares, we set the $1+n_h$-dimensional parameter $\theta=(\nu,\alpha')'$. We also introduce the $1+2n_h$-dimensional parameter $\vartheta=(\theta',b')'$ of the extended Hawkes model. At each prediction origin $\tau_j$, we apply the following rolling-window procedure:
\begin{itemize}
    \item[(i)] We re-estimate the Hawkes model by maximum likelihood and ordinary least squares using the most recent two hours in-sample training interval $I_j=[\tau_j-W,\tau_j]$. This yields $\widehat{\vartheta}^{\text{mle}}_{\tau_j}=((\widehat{\theta}^{\text{mle}}_{\tau_j})',(\widehat{b}^{\text{mle}}_{\tau_j})')'$ in which $\widehat{\theta}^{\text{mle}}_{\tau_j}
    =(\widehat{\nu}^{\text{mle}}_{\tau_j},(\widehat{\alpha}^{\text{mle}}_{\tau_j})')'$.
    The maximum likelihood estimator jointly estimates the baseline intensity, the excitation parameters and the decay numbers.
    Given the maximum likelihood estimated parameter vector $\widehat{b}^{\text{mle}}_{\tau_j}$, we apply the procedure based on ordinary least squares to estimate $\nu$ and $\alpha$. For each time increment parameter $c$, the time increment is set as $\widehat{\Delta}
    =c/\widehat{b}^{\text{mle}}_{\max,\tau_j}$ in which $\widehat{b}^{\text{mle}}_{\max,\tau_j}=
    \max_{1\leq k\leq n_h}\widehat{b}^{\text{mle}}_{\tau_j,k}$ and the resulting estimator based on ordinary least squares is denoted by $\widehat{\theta}_{\tau_j}
    =(\widehat{\nu}_{\tau_j},(\widehat{\alpha}_{\tau_j})')'$.
    In the regression and predictions based on ordinary least squares,
    the unknown decay vector $b^*$ is replaced by $\widehat{b}^{\text{mle}}_{\tau_j}$. We consider $c\in \mathcal{C}=\{10^{-6.5},\ldots,10^2\}$, using $96$ logarithmically equally spaced grid points.
    
    \item[(ii)] Given the estimated parameters, we compute a new forecast
    $\widehat{Y}_{j,s}$ over $[\tau_j,\tau_j+60s]$.
    \item[(iii)] We shift the estimation window forward by $60s$.
\end{itemize}

This procedure is repeated until the end of the trading day, namely 3:30 pm.
The rolling-window estimation allows the prediction procedure to adapt to intraday non-stationarity or possible regime shifts. 
We denote by $\widehat{\vartheta}_{\tau_j}=((\widehat{\theta}_{\tau_j})',(\widehat{b}^{\text{mle}}_{\tau_j})')'$ the extended Hawkes parameter vector with $\widehat{\theta}_{\tau_j}$ the estimator based on ordinary least squares. In contrast, $\widehat{\vartheta}^{\text{mle}}_{\tau_j}=((\widehat{\theta}^{\text{mle}}_{\tau_j})',(\widehat{b}^{\text{mle}}_{\tau_j})')'$ stands for the extended Hawkes parameter vector with its components estimated by maximum likelihood.
In the ordinary least squares-based estimation equipped with $\widehat{\vartheta}_{\tau_j}$, the fitted intensity is
$\lambda_t\bigl(\widehat{\vartheta}_{\tau_j}\bigr)=\widehat{\nu}_{\tau_j}+\sum_{k=1}^{n_h}\widehat{\alpha}_{\tau_j,k}R_{j,k}(t)$ in which $R_{j,k}(t)=\sum_{T_i<t}\exp\big(-\widehat{b}^{\text{mle}}_{\tau_j,k}(t-T_i)\big)$. For any time $u\geq0$, $m_{j,k}(u)=\esp_{\widehat{\vartheta}_{\tau_j}}
\Big[R_{j,k}(\tau_j+u)\,\big|\,\mathcal{F}_{\tau_j}\Big]$ is the conditional expected state variable. The conditional expected intensity is $\esp_{\widehat{\vartheta}_{\tau_j}}\Big[\lambda_{\tau_j+u}\bigl(\widehat{\vartheta}_{\tau_j}\bigr)\,\big|\,\mathcal{F}_{\tau_j}\Big]=\widehat{\nu}_{\tau_j}+\sum_{k=1}^{n_h}\widehat{\alpha}_{\tau_j,k}m_{j,k}(u)$. The predicted number of trades over the next $s$ minutes is
$\widehat{Y}_{j,s}=\int_0^{60s}\Big[\widehat{\nu}_{\tau_j}+\sum_{k=1}^{n_h}\widehat{\alpha}_{\tau_j,k}m_{j,k}(u)\Big]du$.
The same procedure is applied to produce $\widehat{Y}_{j,s}$ based on $\widehat{\vartheta}^{\text{mle}}_{\tau_j}$.
The corresponding prediction error is $\widehat{u}_{j,s}=
Y_{j,s}-\widehat{Y}_{j,s}$.
Furthermore, we will compute the in-sample Akaike Information Criterion (AIC) for $s$-prediction horizon as $\text{AIC}_{j,s} = 2(2n+1)+2 l_{\tau_j}(\widehat{\vartheta}_{\tau_j})$. Here, $l_{\tau_j}(\widehat{\vartheta}_{\tau_j}) = -\underset{\tau_j-W\leq T_i<\tau_j}{\sum}\log(\lambda_{T_i}(\widehat{\vartheta}_{\tau_j}))+\int^{\tau_j}_{\tau_j-W}\lambda_t(\widehat{\vartheta}_{\tau_j})dt$ for any index $j$ and then we compute the AIC averaged over $J$, yielding the daily averaged $\text{AIC}_s$. In particular, the criterion $\text{AIC}_s$ is a function of $c$ in the ordinary least squares case. For a given $s$, prediction accuracy is assessed using the mean absolute error $\mathrm{MAE}_s=\frac{1}{J}\sum_{j=1}^J|\widehat{u}_{j,s}|$, the root mean square error $\mathrm{RMSE}_s=\sqrt{\frac{1}{J}\sum_{j=1}^J\widehat{u}_{j,s}^2}$ and the in-sample AIC. We repeat the estimation and prediction procedure over the full year and report the corresponding metrics averaged over the $249$ trading days, denoted by $\overline{\mathrm{MAE}}_s$, $\overline{\mathrm{RMSE}}_s$ and $\overline{\mathrm{AIC}}_s$. 

Figures \ref{fig:Apple_prediction}-\ref{fig:EOG_prediction} display the prediction metrics $\overline{\text{MAE}}_s$, $\overline{\text{RMSE}}_s$ and $\overline{\text{AIC}}_s$ at an $s$-minute horizon with $s\in\{1,5,10\}$ for different $n_h$, and for the three stocks. The ordinary least squares-based method frequently matches or even outperforms maximum likelihood in terms of $\overline{\text{MAE}}_s$ and $\overline{\text{RMSE}}_s$. Across all three assets and horizons $s$, these prediction errors exhibit a U-shaped pattern with respect to the parameter $c$, with their minimal often below the corresponding maximum likelihood values. Small values of the parameter $c$ yield relatively stable prediction performance, a behavior consistent with the flat $\ell_2$-error curves observed in the simulation experiments of Subsection \ref{appendix_subsec:sim_l2error} of the Appendices. Large values of the parameter $c$ lead to an excessively coarse discretization and deteriorate prediction accuracy. The $\overline{\text{AIC}}_s$ exhibits a similar pattern. Although its minimum typically occurs at a slightly larger value of the parameter $c$ than those of $\overline{\text{MAE}}_s$ and $\overline{\text{RMSE}}_s$, it identifies a similar range of time increments. This supports its use as a data-driven criterion for selecting the time increment $\Delta$.



Furthermore, increasing the number of trader types $n_h$ generally improves prediction accuracy. This is reflected by lower $\overline{\text{MAE}}_s$, $\overline{\text{RMSE}}_s$ and $\overline{\text{AIC}}_s$ for both estimation methods. However, this improvement no longer holds for large $n_h$. In particular, the prediction performance of maximum likelihood deteriorates markedly when $n_h=5$ for AAPL and META. However, their $\overline{\text{AIC}}_s$ remains close to that obtained for $n_h=3$ suggesting possible overfitting. The predictions based on ordinary least squares are improved when $n_h=4$. However, they become unstable with respect to the parameter $c$ when $n_h=5$ especially in terms of $\overline{\text{RMSE}}_s$. This increased sensitivity suggests that the regression problem becomes increasingly ill-conditioned as the number of parameters grows.
The prediction results confirm the results of Section \ref{secempiricalstudyinsample} of the Appendices, suggesting that there are three different types of traders in each asset. In Section \ref{secempiricalstudytest}, we will document the heterogeneity in the three trader types. However, we find that we can still improve the prediction results with the use of four traders. This shows that the use of four traders in practice could also be useful to the practitioner.

\begin{figure}[htbp]
\centering

\begin{subfigure}[b]{0.32\textwidth}
\centering
\includegraphics[width=.95\textwidth]{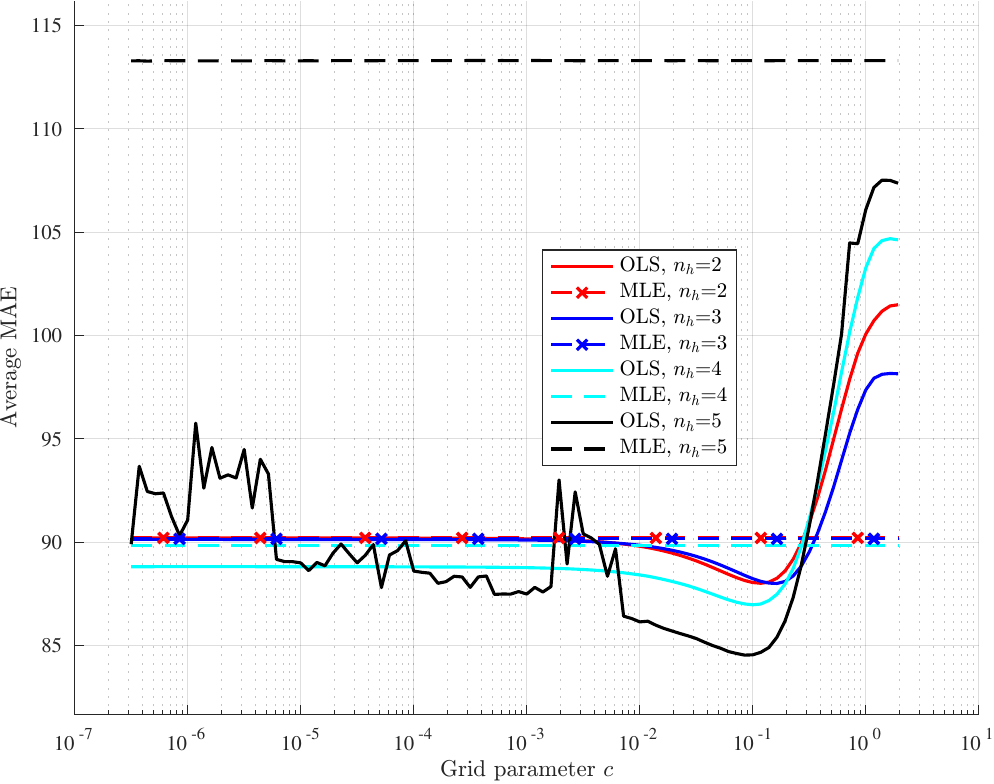}
\caption{$\overline{\text{MAE}}_s$ for $s=1$}
\end{subfigure}
\hfill
\begin{subfigure}[b]{0.32\textwidth}
\centering
\includegraphics[width=.95\textwidth]{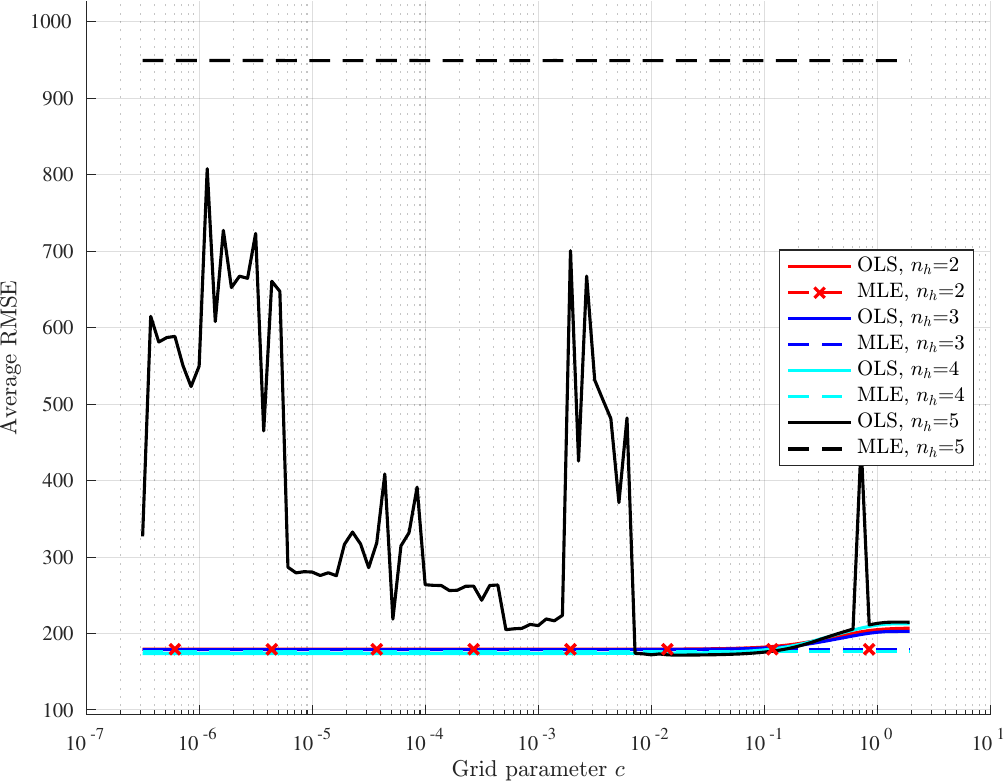}
\caption{$\overline{\text{RMSE}}_s$ for $s=1$}
\end{subfigure}
\hfill
\begin{subfigure}[b]{0.32\textwidth}
\centering
\includegraphics[width=.95\textwidth]{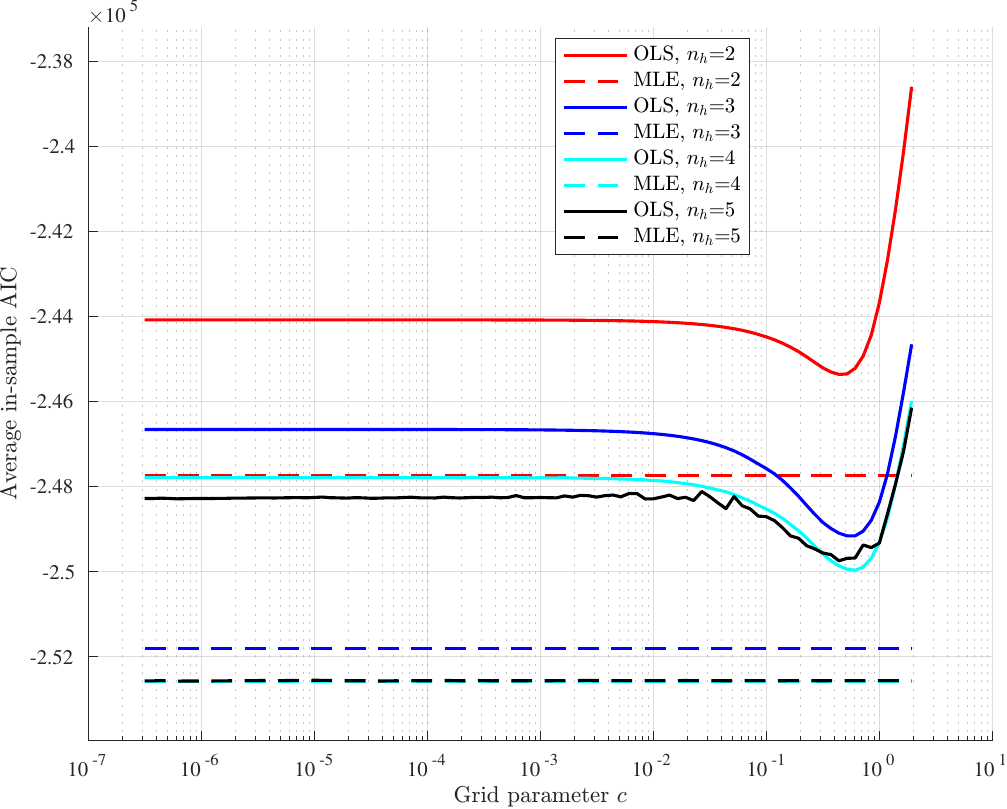}
\caption{$\overline{\text{AIC}}_s$ for $s=1$}
\end{subfigure}

\begin{subfigure}[b]{0.32\textwidth}
\centering
\includegraphics[width=.95\textwidth]{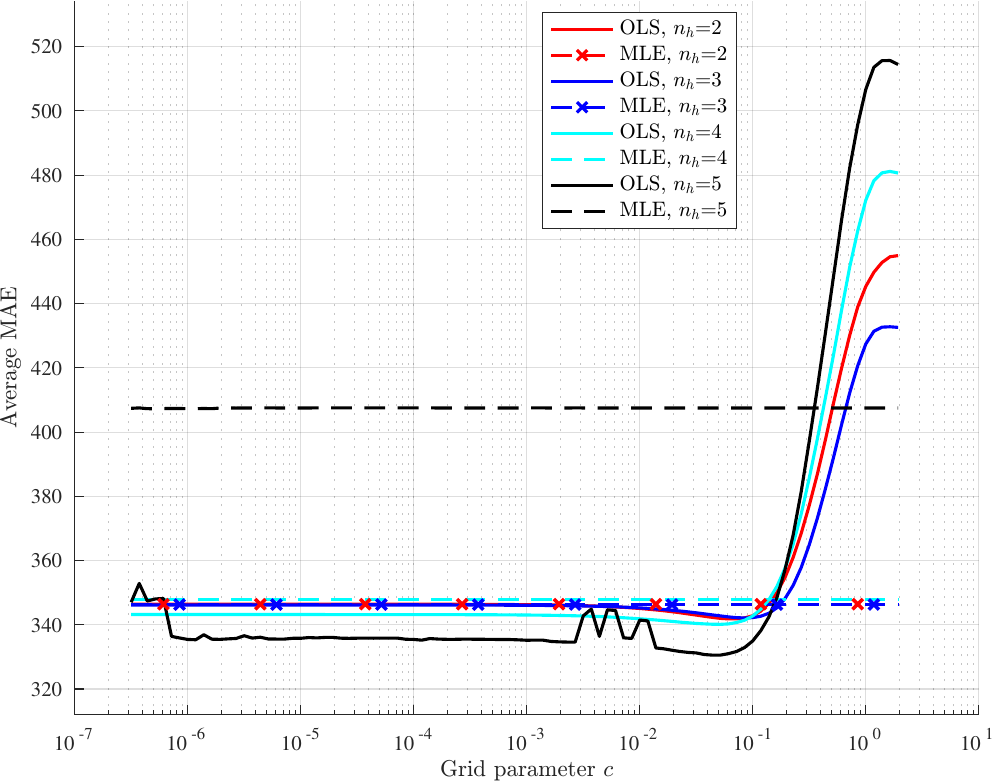}
\caption{$\overline{\text{MAE}}_s$ for $s=5$}
\end{subfigure}
\hfill
\begin{subfigure}[b]{0.32\textwidth}
\centering
\includegraphics[width=.95\textwidth]{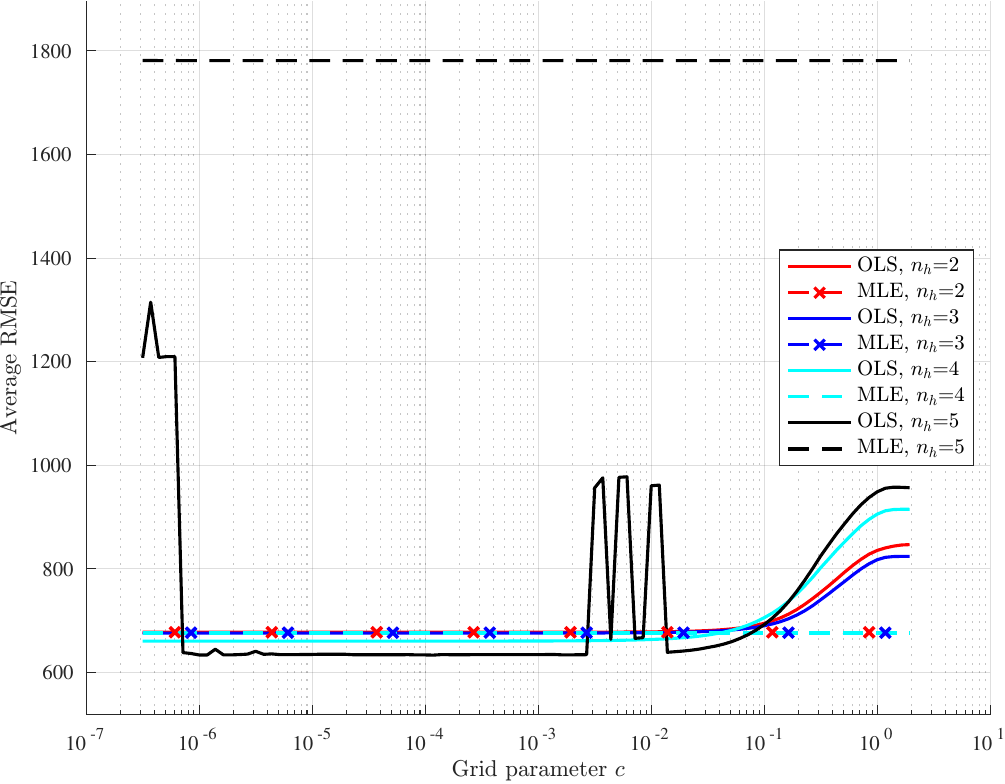}
\caption{$\overline{\text{RMSE}}_s$ for $s=5$}
\end{subfigure}
\hfill
\begin{subfigure}[b]{0.32\textwidth}
\centering
\includegraphics[width=.95\textwidth]{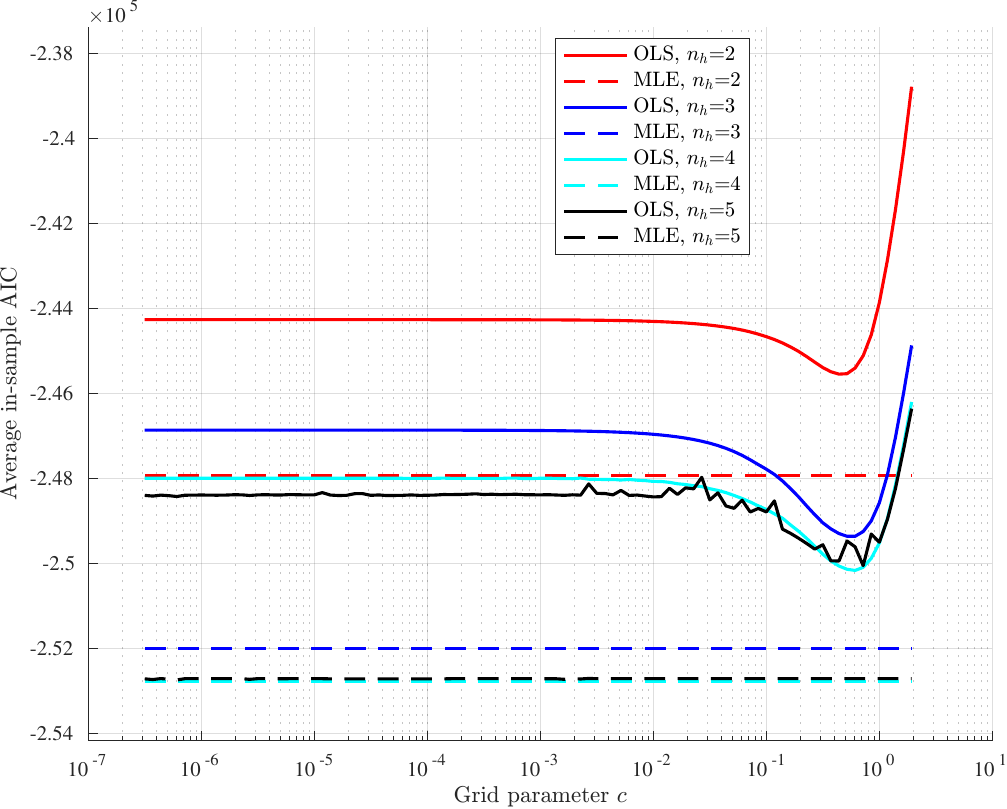}
\caption{$\overline{\text{AIC}}_s$ for $s=5$}
\end{subfigure}

\begin{subfigure}[b]{0.32\textwidth}
\centering
\includegraphics[width=.95\textwidth]{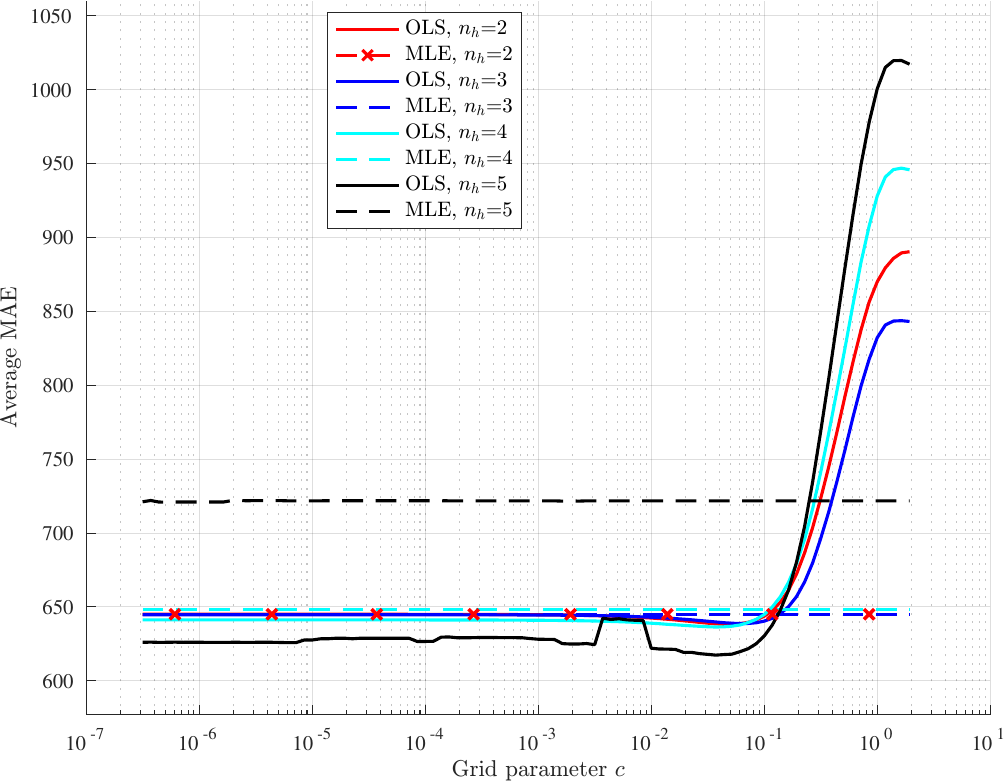}
\caption{$\overline{\text{MAE}}_s$ for $s=10$}
\end{subfigure}
\hfill
\begin{subfigure}[b]{0.32\textwidth}
\centering
\includegraphics[width=.95\textwidth]{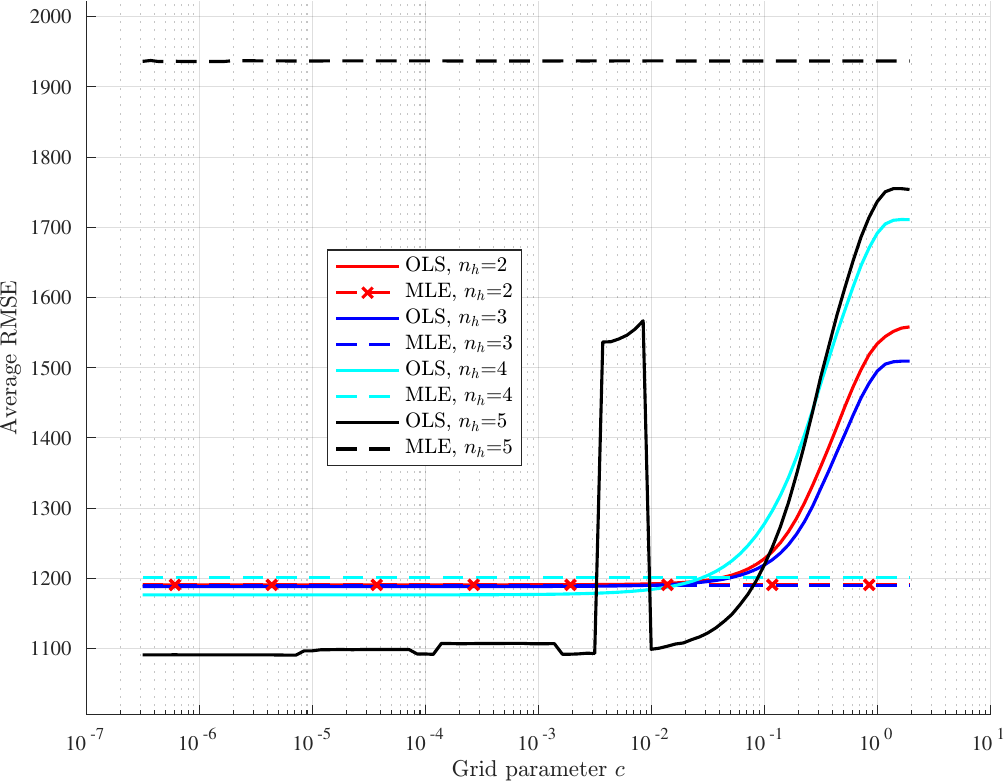}
\caption{$\overline{\text{RMSE}}_s$ for $s=10$}
\end{subfigure}
\hfill
\begin{subfigure}[b]{0.32\textwidth}
\centering
\includegraphics[width=.95\textwidth]{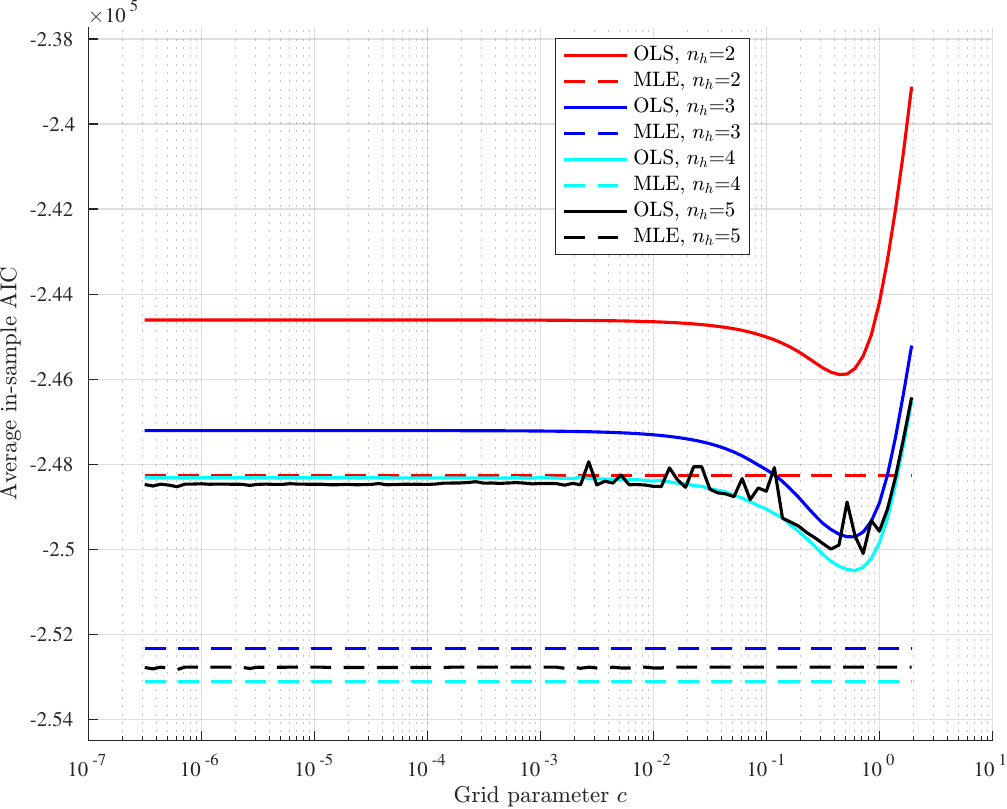}
\caption{$\overline{\text{AIC}}_s$ for $s=10$}
\end{subfigure}

\caption{Out-of-sample $\overline{\text{MAE}}_s$, out-of-sample $\overline{\text{RMSE}}_s$ and in-sample $\overline{\text{AIC}}_s$ at an $s$-minute horizon, $s\in\{1,5,10\}$. The dashed (resp. solid) lines represent the maximum likelihood estimator (resp. ordinary least squares). For readability only, markers "$\times$" are added to curves that nearly overlap. Asset: AAPL.}
\label{fig:Apple_prediction}
\end{figure}

\begin{figure}[htbp]
\centering

\begin{subfigure}[b]{0.32\textwidth}
\centering
\includegraphics[width=.95\textwidth]{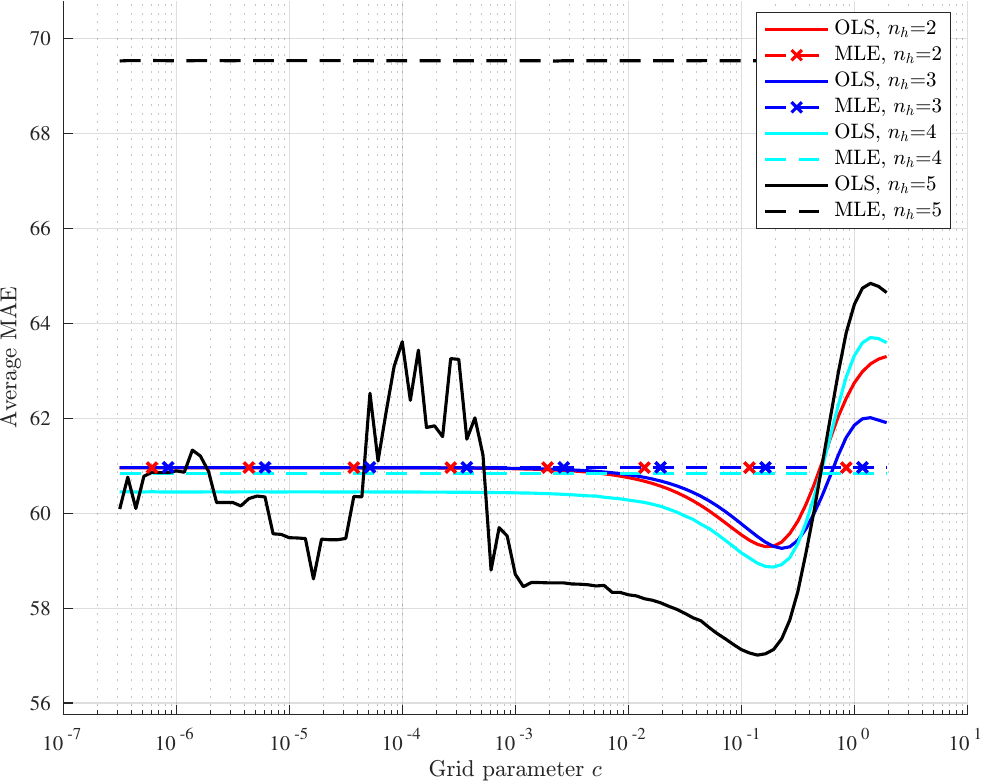}
\caption{$\overline{\text{MAE}}_s$ for $s=1$}
\end{subfigure}
\hfill
\begin{subfigure}[b]{0.32\textwidth}
\centering
\includegraphics[width=.95\textwidth]{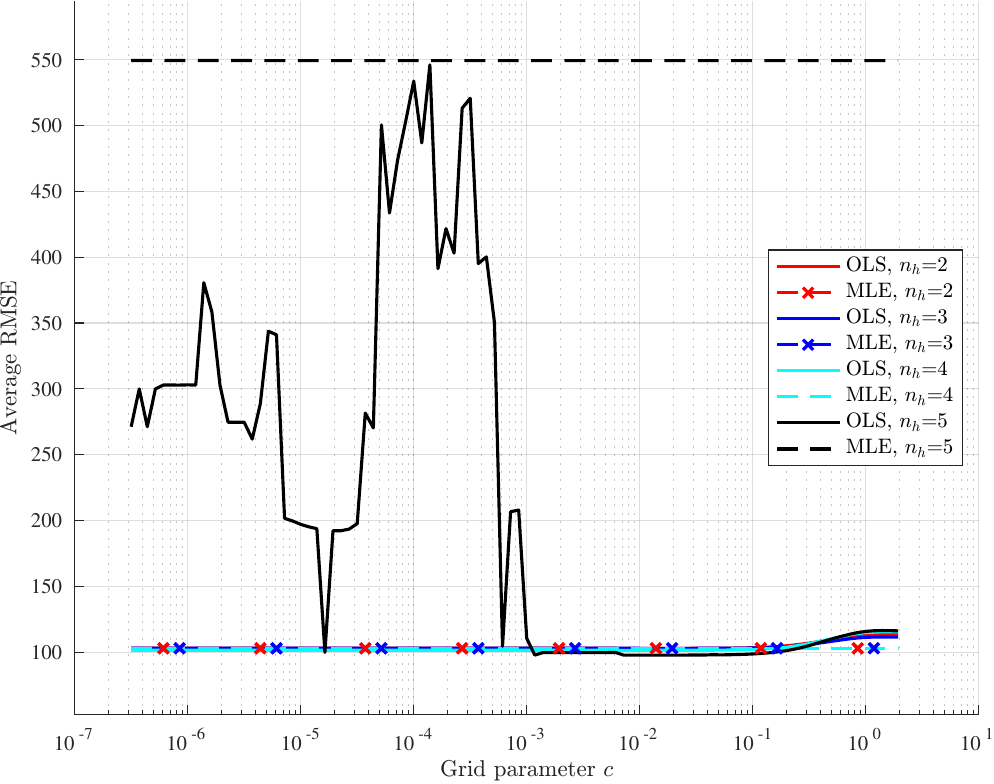}
\caption{$\overline{\text{RMSE}}_s$ for $s=1$}
\end{subfigure}
\hfill
\begin{subfigure}[b]{0.32\textwidth}
\centering
\includegraphics[width=.95\textwidth]{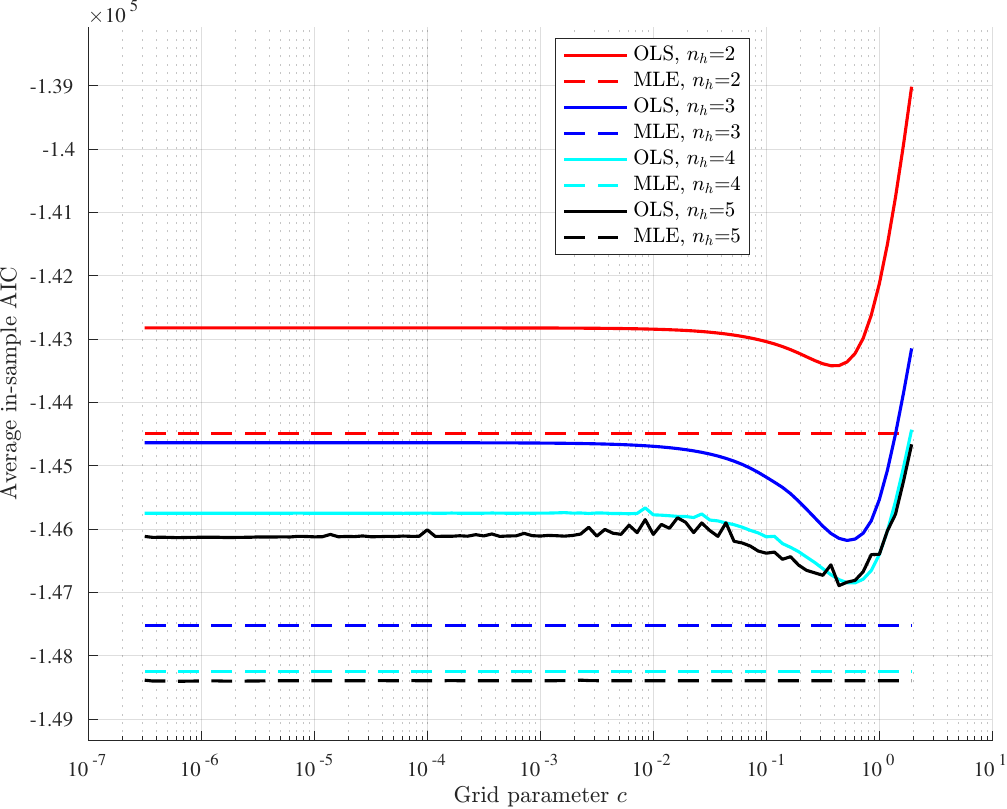}
\caption{$\overline{\text{AIC}}_s$ for $s=1$}
\end{subfigure}

\begin{subfigure}[b]{0.32\textwidth}
\centering
\includegraphics[width=.95\textwidth]{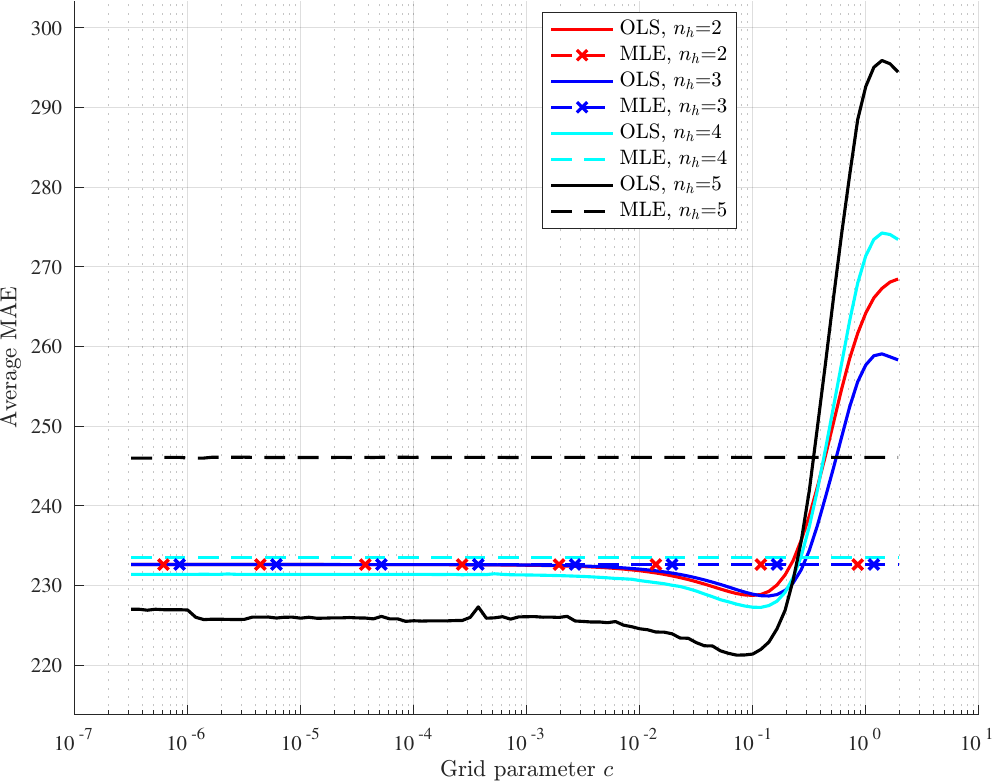}
\caption{$\overline{\text{MAE}}_s$ for $s=5$}
\end{subfigure}
\hfill
\begin{subfigure}[b]{0.32\textwidth}
\centering
\includegraphics[width=.95\textwidth]{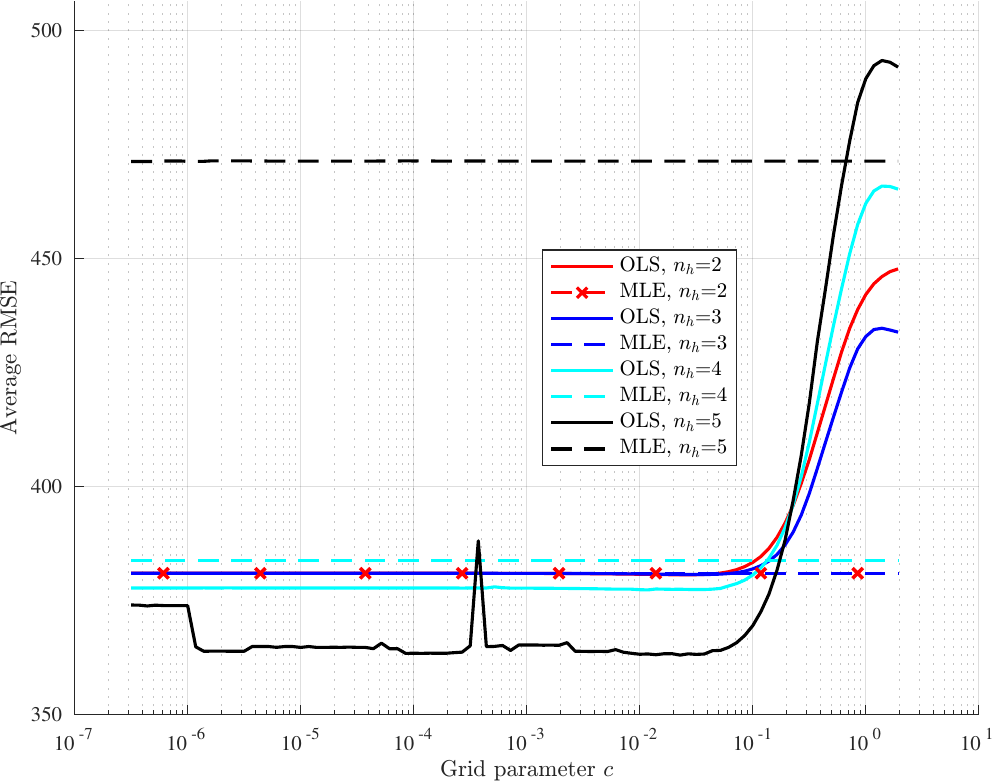}
\caption{$\overline{\text{RMSE}}_s$ for $s=5$}
\end{subfigure}
\hfill
\begin{subfigure}[b]{0.32\textwidth}
\centering
\includegraphics[width=.95\textwidth]{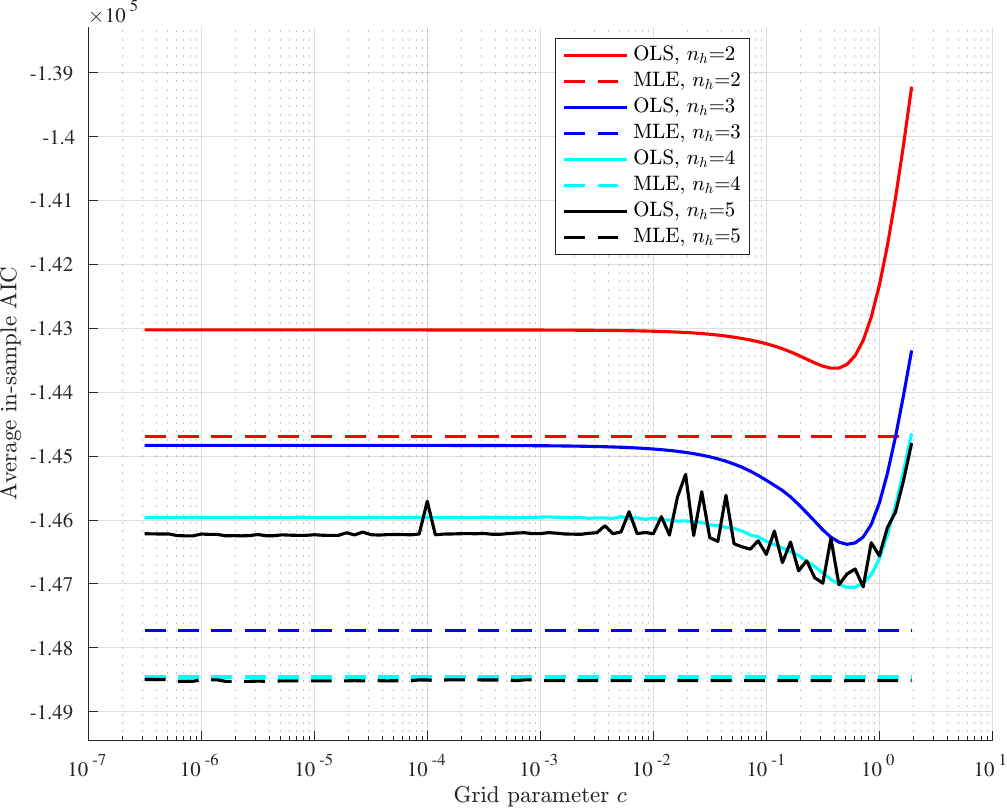}
\caption{$\overline{\text{AIC}}_s$ for $s=5$}
\end{subfigure}

\begin{subfigure}[b]{0.32\textwidth}
\centering
\includegraphics[width=.95\textwidth]{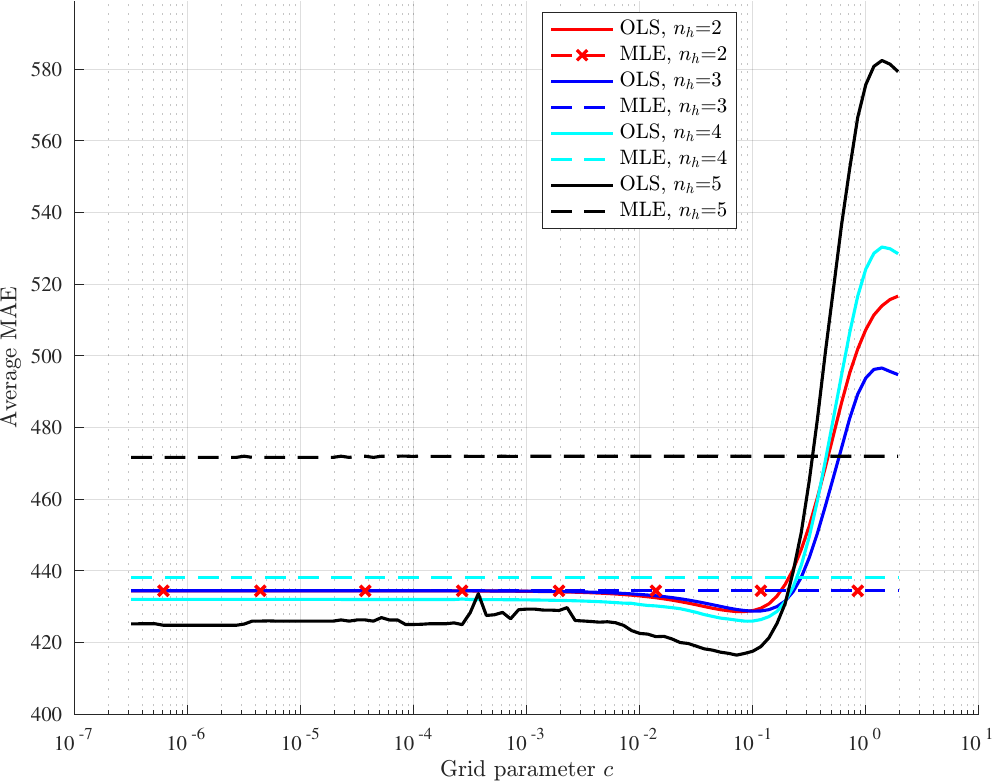}
\caption{$\overline{\text{MAE}}_s$ for $s=10$}
\end{subfigure}
\hfill
\begin{subfigure}[b]{0.32\textwidth}
\centering
\includegraphics[width=.95\textwidth]{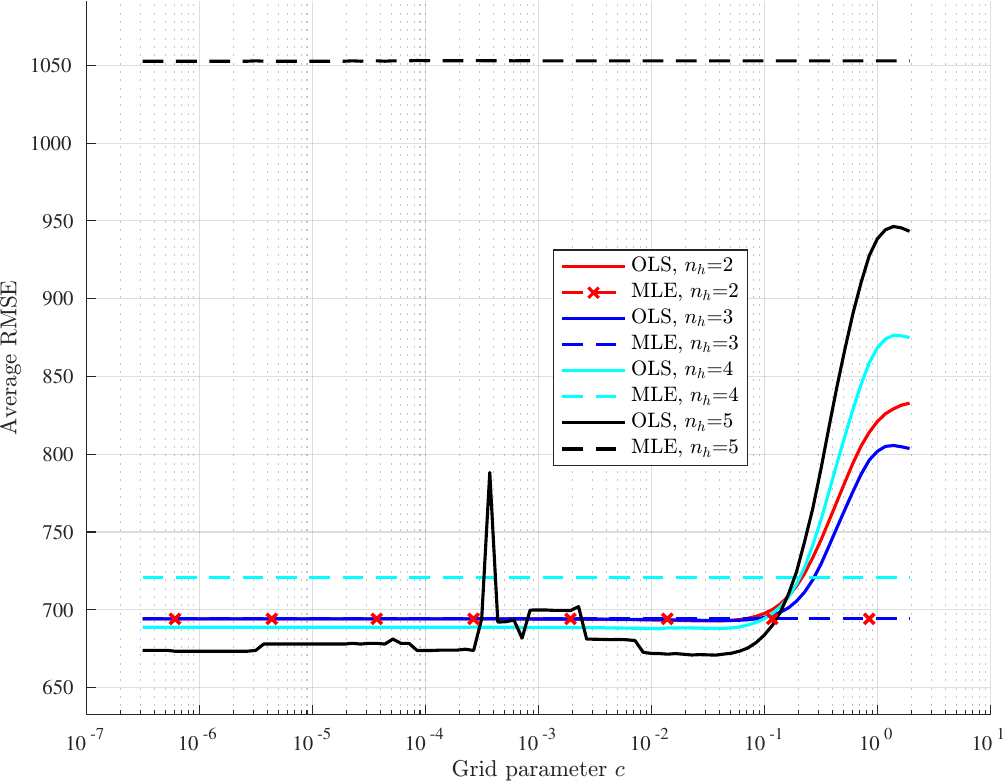}
\caption{$\overline{\text{RMSE}}_s$ for $s=10$}
\end{subfigure}
\hfill
\begin{subfigure}[b]{0.32\textwidth}
\centering
\includegraphics[width=.95\textwidth]{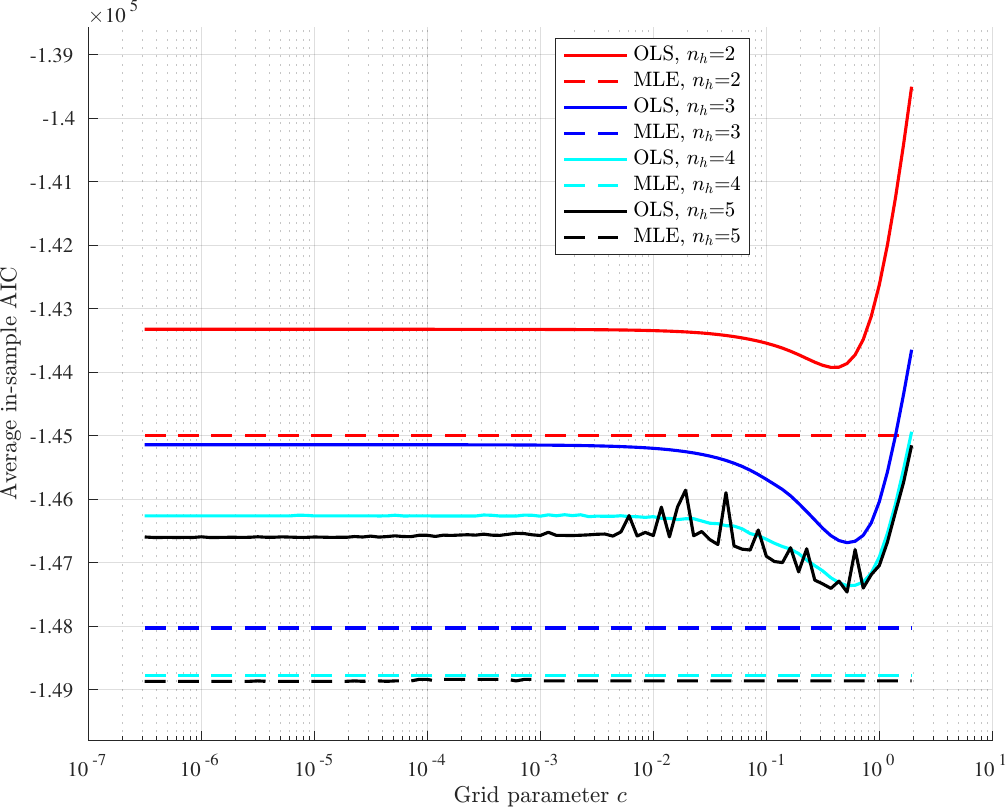}
\caption{$\overline{\text{AIC}}_s$ for $s=10$}
\end{subfigure}

\caption{Out-of-sample $\overline{\text{MAE}}_s$, out-of-sample $\overline{\text{RMSE}}_s$ and in-sample $\overline{\text{AIC}}_s$ at an $s$-minute horizon, $s\in\{1,5,10\}$. The dashed (resp. solid) lines represent the maximum likelihood estimator (resp. ordinary least squares). For readability only, markers "$\times$" are added to curves that nearly overlap.  Asset: META.}
\label{fig:FB_prediction}
\end{figure}

\begin{figure}[htbp]
\centering

\begin{subfigure}[b]{0.32\textwidth}
\centering
\includegraphics[width=.95\textwidth]{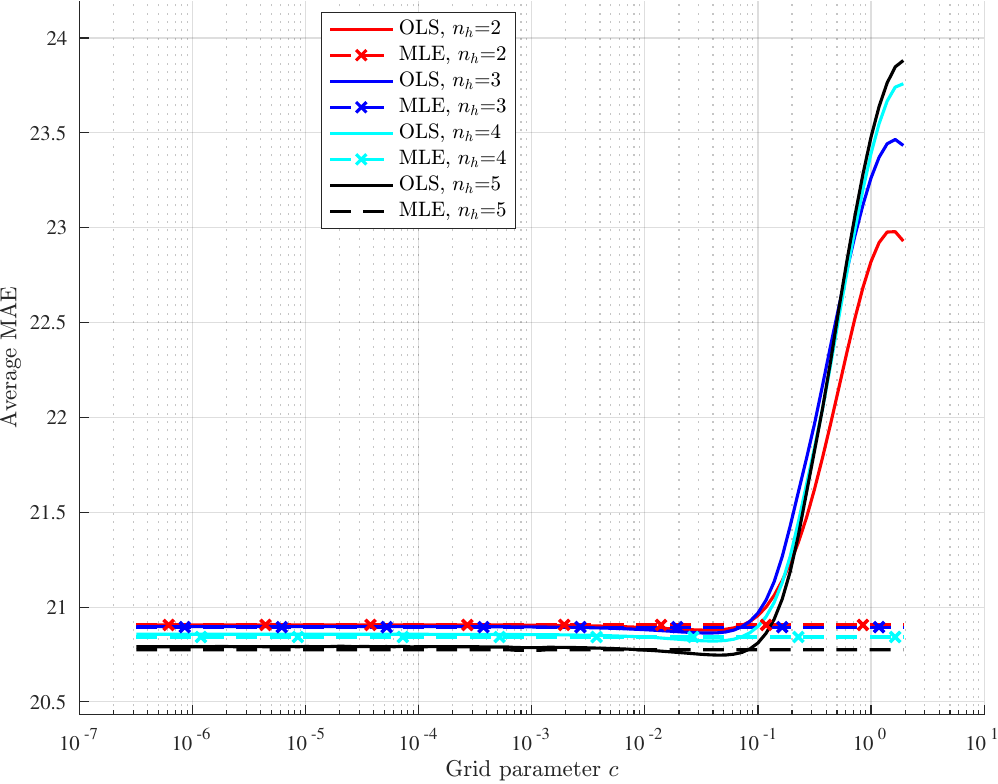}
\caption{$\overline{\text{MAE}}_s$ for $s=1$}
\end{subfigure}
\hfill
\begin{subfigure}[b]{0.32\textwidth}
\centering
\includegraphics[width=.95\textwidth]{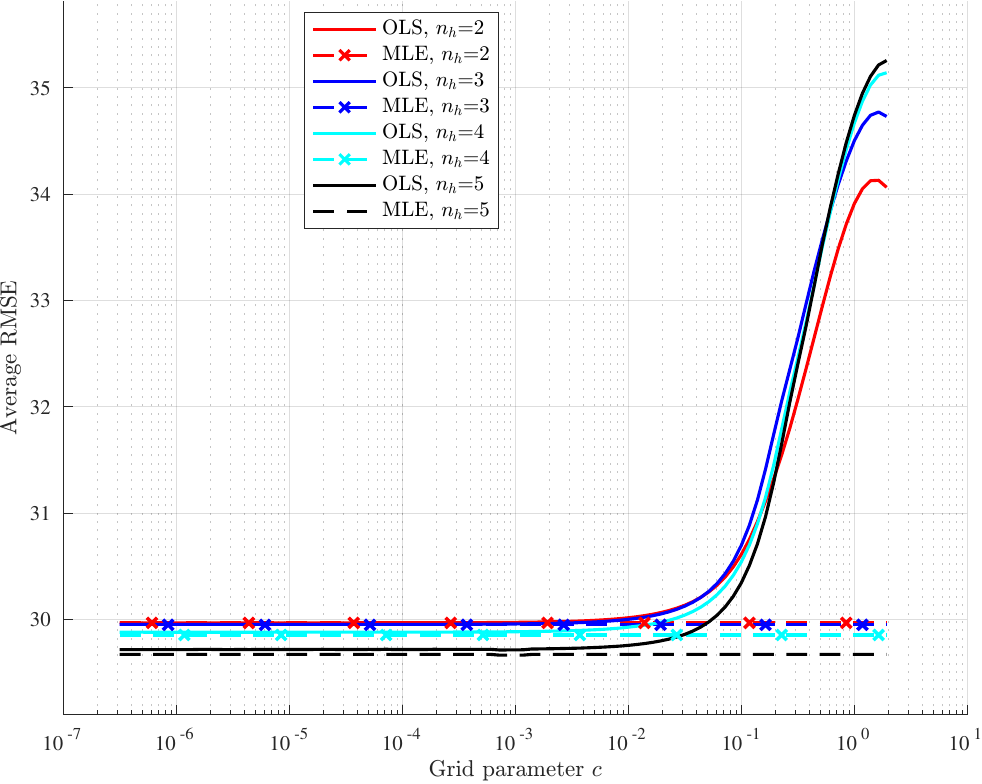}
\caption{$\overline{\text{RMSE}}_s$ for $s=1$}
\end{subfigure}
\hfill
\begin{subfigure}[b]{0.32\textwidth}
\centering
\includegraphics[width=.95\textwidth]{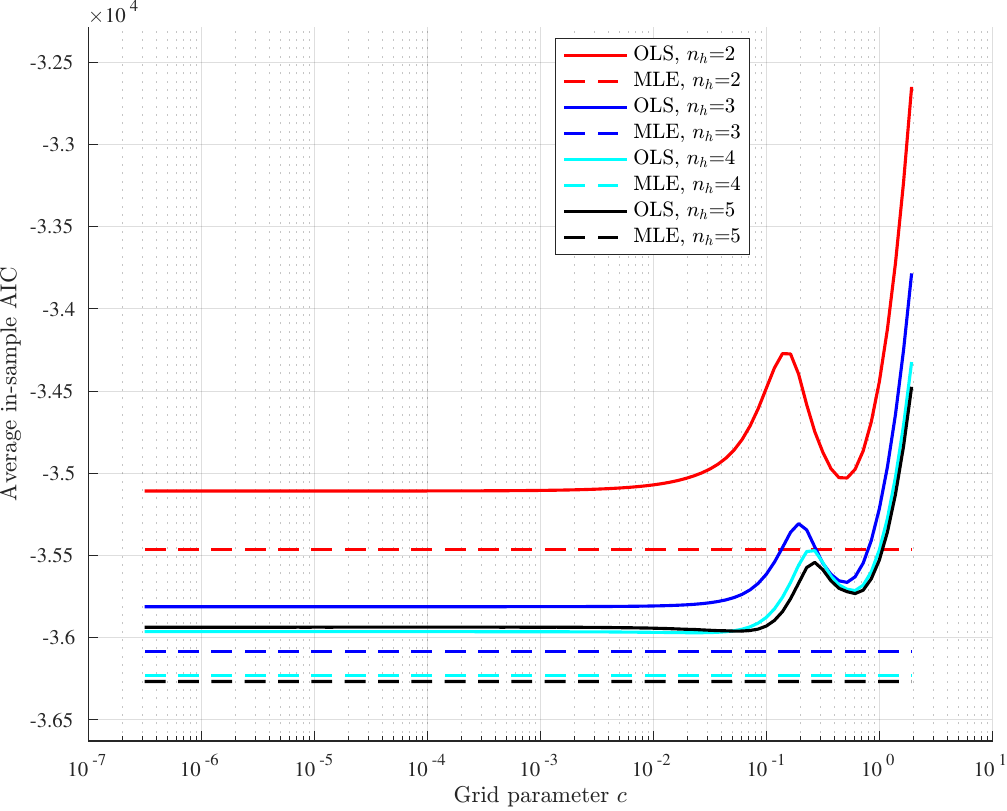}
\caption{$\overline{\text{AIC}}_s$ for $s=1$}
\end{subfigure}

\begin{subfigure}[b]{0.32\textwidth}
\centering
\includegraphics[width=.95\textwidth]{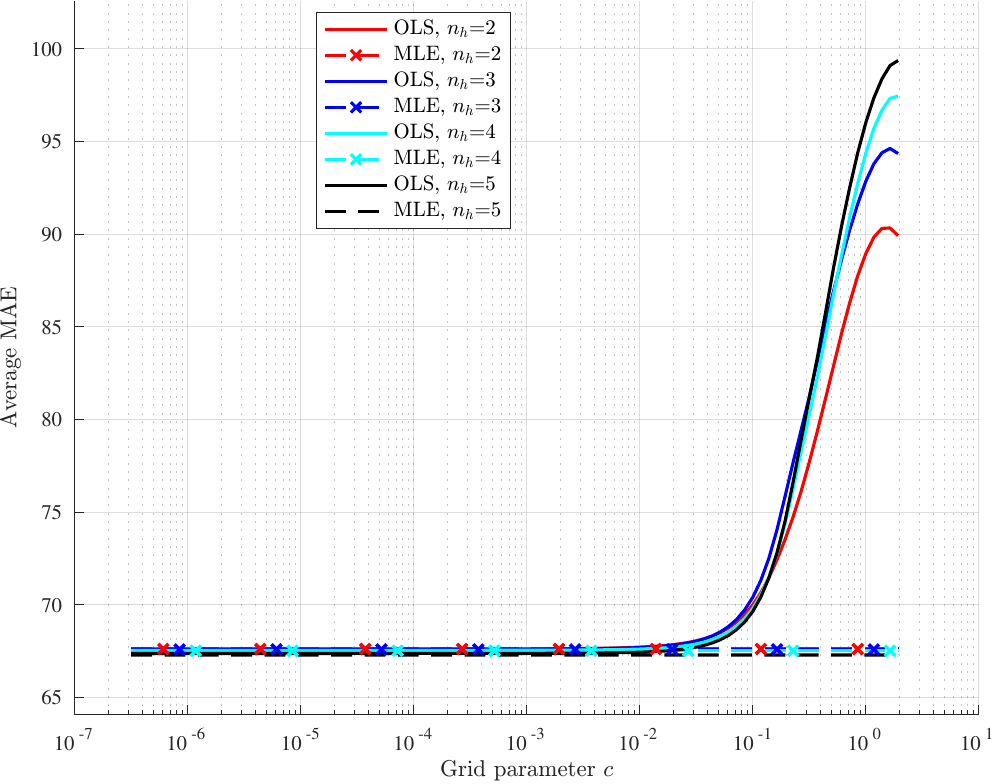}
\caption{$\overline{\text{MAE}}_s$ for $s=5$}
\end{subfigure}
\hfill
\begin{subfigure}[b]{0.32\textwidth}
\centering
\includegraphics[width=.95\textwidth]{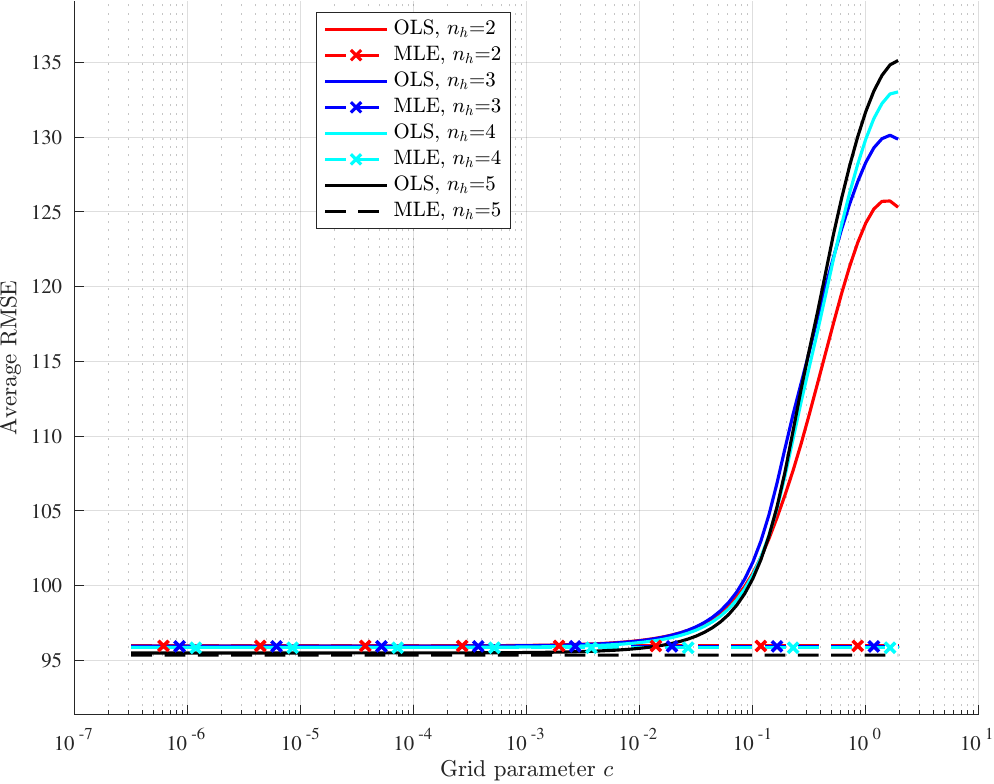}
\caption{$\overline{\text{RMSE}}_s$ for $s=5$}
\end{subfigure}
\hfill
\begin{subfigure}[b]{0.32\textwidth}
\centering
\includegraphics[width=.95\textwidth]{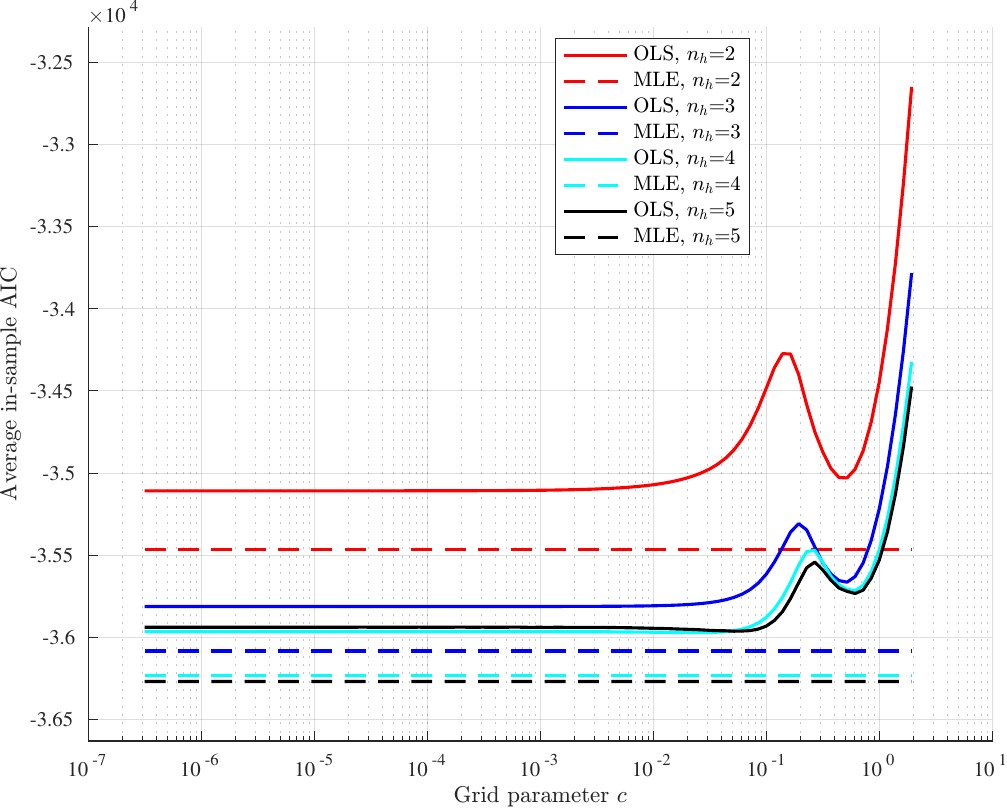}
\caption{$\overline{\text{AIC}}_s$ for $s=5$}
\end{subfigure}

\begin{subfigure}[b]{0.32\textwidth}
\centering
\includegraphics[width=.95\textwidth]{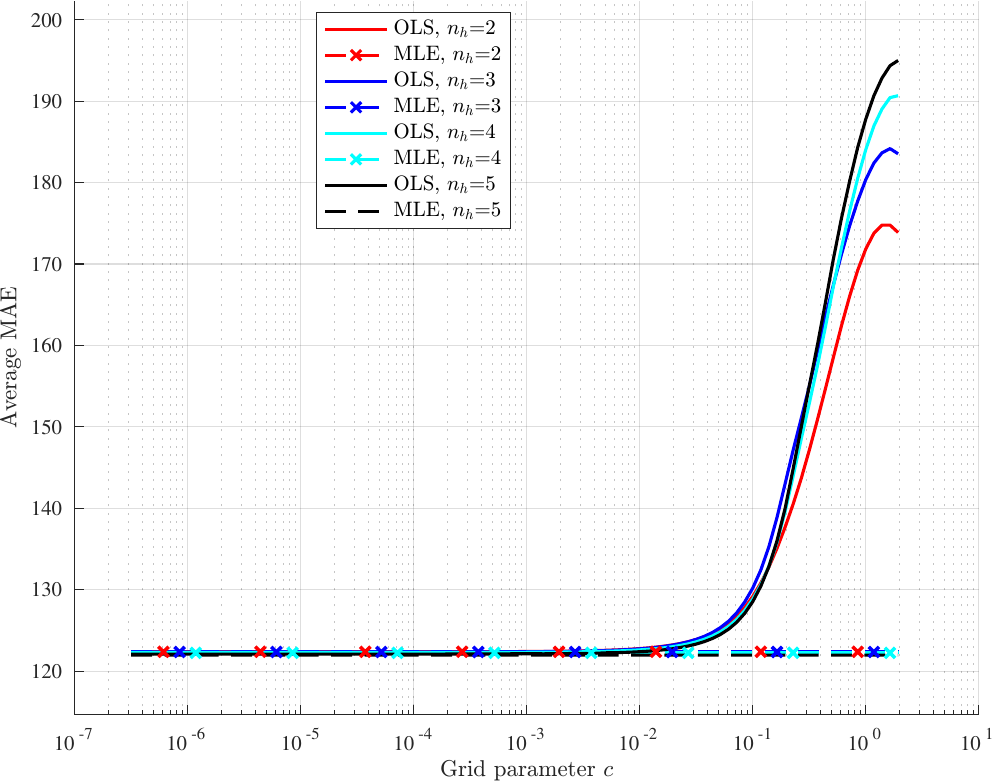}
\caption{$\overline{\text{MAE}}_s$ for $s=10$}
\end{subfigure}
\hfill
\begin{subfigure}[b]{0.32\textwidth}
\centering
\includegraphics[width=.95\textwidth]{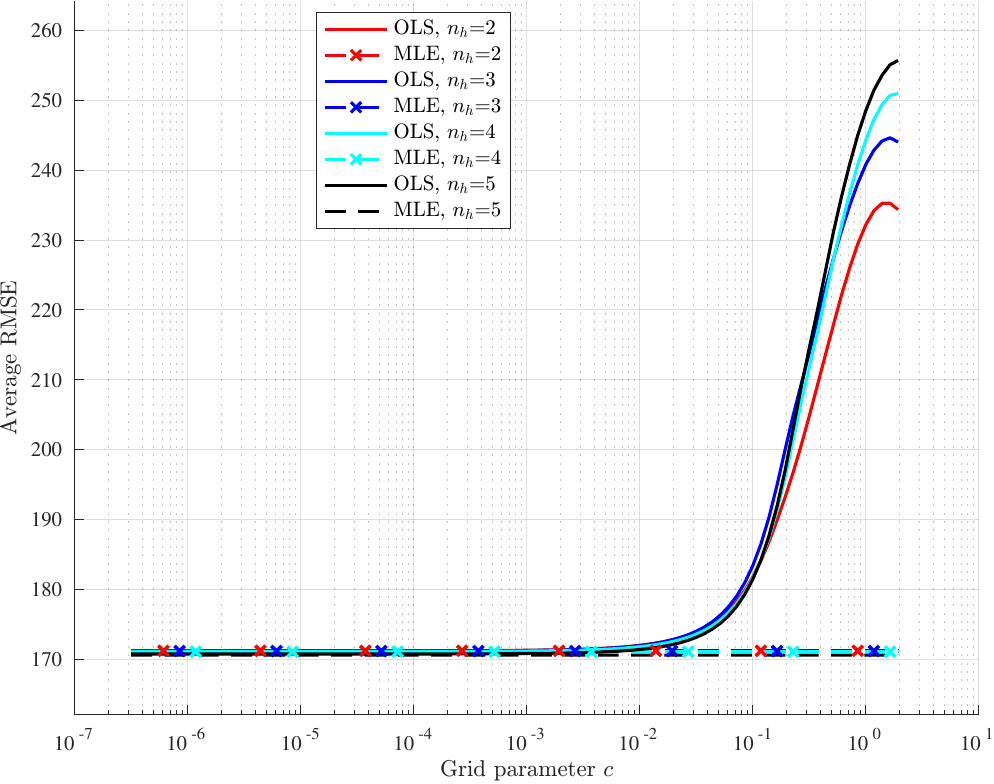}
\caption{$\overline{\text{RMSE}}_s$ for $s=10$}
\end{subfigure}
\hfill
\begin{subfigure}[b]{0.32\textwidth}
\centering
\includegraphics[width=.95\textwidth]{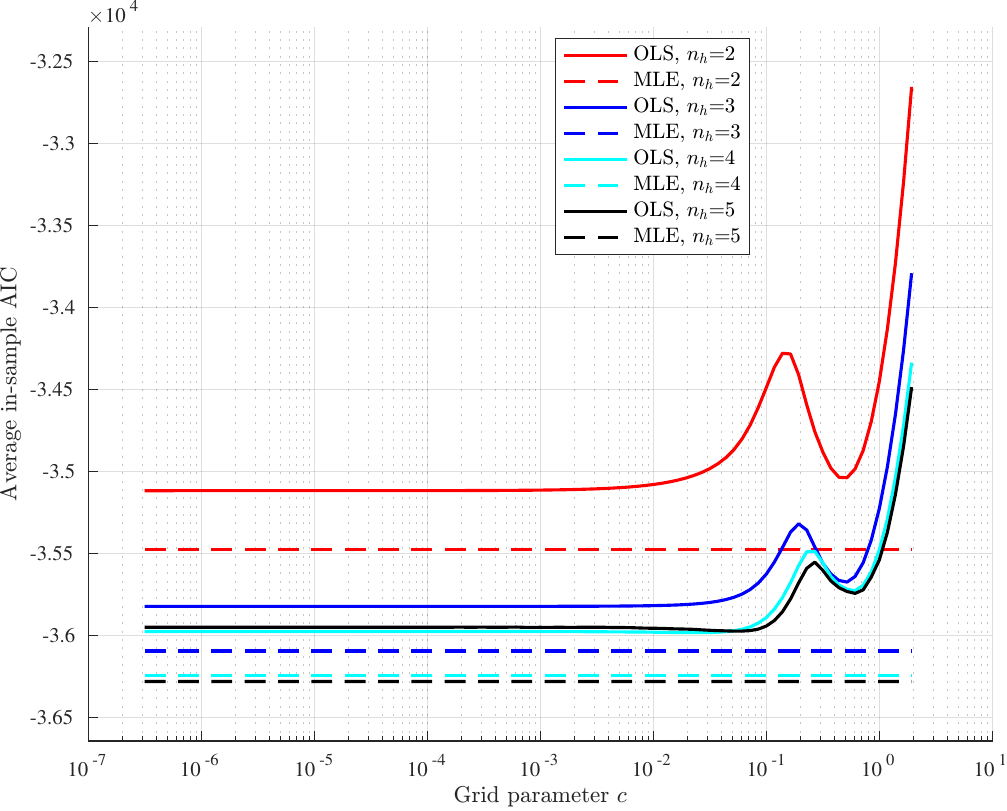}
\caption{$\overline{\text{AIC}}_s$ for $s=10$}
\end{subfigure}

\caption{Out-of-sample $\overline{\text{MAE}}_s$, out-of-sample $\overline{\text{RMSE}}_s$ and in-sample $\overline{\text{AIC}}_s$ at an $s$-minute horizon $s\in\{1,5,10\}$. The dashed (resp. solid) lines represent the maximum likelihood estimator (resp. ordinary least squares). For readability only, markers "$\times$" are added to curves that nearly overlap.  Asset: EOG.}
\label{fig:EOG_prediction}
\end{figure}

\subsection{Heterogeneity in trader types}
\label{secempiricalstudytest}

Table \ref{tab:average_parameters} reports the average estimated baseline parameter $\nu$, estimated exciting parameter $\alpha$ and estimated decay number $b$ for the maximum likelihood estimator and the estimator relying on ordinary least squares for particular values of time increment parameter $c$. The figures are obtained over the period January 2017 - December 2017 for the three stocks. Overall, there is heterogeneity in trader types. More specifically, we observe that the exciting parameter $\alpha_1$ of the first trader type is always larger than $1.5$ and smaller than $21$. In addition, the exciting parameter $\alpha_2$ of the second trader type is consistently larger than $100$ and smaller than $1000$. Moreover, the exciting parameter $\alpha_3$ of the third trader type is consistently larger than $3000$ and smaller than $30000$. Finally, there is heterogeneity in excitation parameter based on the liquidity of the asset. Namely, we have larger excitation parameter values for the two most traded assets AAPL and META and smaller values for the active asset EOG.

With small excitation values, we can interpret the first trader type as fundamental traders. Fundamental traders analyze economic data, financial statements and market conditions to make long-term investment decisions. With large excitation values, we can interpret the second and third trader type as high-frequency traders. High-frequency traders are institutional financial market participants who employ powerful computers, advanced algorithms, and ultra-fast telecommunications to execute a massive volume of orders in fractions of a second. The second trader type and third trader type are both high-frequency traders but they differ in their speed to use technology due to the difference in their equipment (see \cite{gagnon2010multi} , \cite{Hasbrouck2013low}, \cite{hoffmann2014dynamic}, \cite{budish2015high},\cite{biais2015equilibrium}, \cite{foucault2016news} and \cite{pagnotta2018competing}). Finally, this interpretation is coherent with the heterogeneity results due to the asset liquidity.

\begin{table}[H]
\centering
\caption{Average estimated baseline parameter $\nu$, estimated exciting parameter $\alpha$ and estimated decay number $b$ for the maximum likelihood estimator (MLE) and ordinary least squares (OLS) estimator over January 2017 - December 2017. }
\label{tab:average_parameters}
\begin{tabular}{ccccccccc}
\hline\hline & & MLE& OLS& OLS& OLS& OLS& OLS& OLS
\\
& & --& {\scriptsize $c=0.005$}& {\scriptsize $c=0.012$}& {\scriptsize $c=0.027$}& {\scriptsize $c=0.061$}& {\scriptsize $c=0.139$}& {\scriptsize $c=0.316$}
\\
\hline

AAPL &$\nu$
& 1.875 & 1.496 & 1.494 & 1.490 & 1.479 & 1.462 & 1.445 
\\

&$\alpha_1$
& 17.819 & 8.039 & 8.127 & 8.329 & 8.799 & 9.680 & 11.380 
\\
&$b_1$
& 106.735 & - & - & - & - & - & - 
\\

&$\alpha_2$
& 655.607 & 976.549 & 971.211 & 958.919 & 933.763 & 896.357 & 835.938 
\\
&$b_2$
& 2692.49 & - & - & - & - & - & - 
\\

&$\alpha_3$
& 22614.46 & 25888.40 & 25915.73 & 25979.83 & 26131.66 & 25997.56 & 24771.38 
\\
&$b_3$
& 96160.04 & - & - & - & - & - & - 
\\

\hline

META&$\nu$
& 1.231 & 1.058 & 1.057 & 1.055 & 1.049 & 1.039 & 1.026 
\\

&$\alpha_1$
& 20.227 & 11.616 & 11.690 & 11.855 & 12.264 & 13.079 & 14.401 
\\
&$b_1$
& 104.457 & - & - & - & - & - & - 
\\

&$\alpha_2$
& 668.146 & 1159.58 & 1140.81 & 1098.47 & 984.657 & 798.466 & 712.667 
\\
&$b_2$
& 3551.49 & - & - & - & - & - & - 
\\

&$\alpha_3$
& 22141.44 & 24466.98 & 24507.09 & 24598.58 & 24833.68 & 25001.42 & 24207.71 
\\
&$b_3$
& 100061.69 & - & - & - & - & - & - 
\\

\hline

EOG&$\nu$
& 0.241 & 0.263 & 0.262 & 0.260 & 0.255 & 0.244 & 0.260 
\\

&$\alpha_1$
& 2.292 & 1.743 & 1.752 & 1.770 & 1.845 & 1.967 & 1.973 
\\
&$b_1$
& 13.831 & - & - & - & - & - & - 
\\

&$\alpha_2$
& 103.099 & 81.798 & 80.533 & 77.618 & 60.549 & 36.867 & 40.643 
\\
&$b_2$
& 514.747 & - & - & - & - & - & - 
\\

&$\alpha_3$
& 3822.65 & 4444.37 & 4449.36 & 4461.00 & 4499.79 & 4555.54 & 3957.31 
\\
&$b_3$
& 8980.57 & - & - & - & - & - & -
\\

\hline\hline
\end{tabular}
\end{table}

To confirm the previous results on heterogeneity of trader types, we consider a statistical test based on the Wald test statistic $S_{T\Delta}$ defined in (\ref{defW}). We naturally restrict our analysis to the number of trader types $n_h=3$ and introduce the following set of null hypothesis: $H_{0,1}: \, \alpha_1^*=\alpha_2^*, \;\; H_{0,2}: \, \alpha_1^*=\alpha_3^*, \;\; H_{0,3}: \, \alpha_2^*=\alpha_3^*$. These three null hypothesis consider a test of equality between two distinct exciting parameters. We also introduce the associated alternative hypothesis: $H_{1,1}: \, \alpha_1^* \neq \alpha_2^*, \;\; H_{1,2}: \, \alpha_1^* \neq \alpha_3^*, \;\; H_{0,3}: \, \alpha_2^* \neq \alpha_3^*$. We express each restriction under $R\theta^*=r$ with $\theta^* = (\nu^*,\alpha_1^*,\alpha_2^*,\alpha_3^*)'$. We set $r=0$ and employ $R_1=(0,1,-1,0)$, $R_2=(0,1,0,-1)$ and $R_3=(0,0,1,-1)$ for the null hypothesis $H_{0,1}$, $H_{0,2}$ and $H_{0,3}$. By Corollary \ref{corwald}, the test statistic $S_{T\Delta}$ converges in distribution to a chi-square distribution with $1$ degree of freedom under the null hypothesis $H_{0,j}$ for any  index $j=1,2,3$. Since the three hypotheses are tested separately on each of the $249$ retained trading days, we account for multiple testing using a Bonferroni correction over the full family of $3 \times 249=747$ tests for each time increment $\Delta$. At the significance level $5\%$, a null hypothesis $H_{0,j}$ on day $t$ is rejected whenever its $p$-value $p_{j,t}$ satisfies $p_{j,t} \leq 0.05/747$. We use the Bonferroni corrected $p$-value $p^{\text{b}}_{j,t}=\min(747\,p_{j,t},1)$ and reject when $p^{\text{b}}_{j,t} \leq 0.05$. See \cite{bajgrowicz2016jumps} for further details on the Bonferroni correction.  

Table \ref{table:wald_test_real} reports the percentage of trading days for which the pairwise equality hypotheses between the excitation parameters are rejected for different parameter $c$. The rejection frequencies remain extremely high after applying the Bonferroni correction over the three hypotheses and the 249 trading days. For META and AAPL, all rejection frequencies exceed 97\% for any choice of time increment parameter $c$. The results for EOG are similarly strong, although the rejection frequency for the null hypothesis $H_{0,1}: \alpha_1^*=\alpha_2^*$ is somewhat lower when the time increment parameter  is fixed to $c=0.316$. The rejection frequency is equal to 93.98\%. The conclusions are highly stable across the parameter $c$. These results provide strong empirical evidence that the three excitation parameters capture distinct trader types. Finally, this supports the use of a Hawkes process with three types of traders.

\begin{table}[H]
\centering
\caption{Percentage of trading days on which the pairwise equality hypotheses for the excitation coefficients are rejected at the 5\% significance level. For each fixed value of the   time increment parameter $c$, the reported results use a Bonferroni correction over the three hypotheses and the 249 trading
days.}
\label{table:wald_test_real}
\begin{tabular}{ll|cccccc}
\hline\hline
Asset
& Null hypothesis
& $c=0.005$
& $c=0.027$
& $c=0.316$
& $c=0.720$
& $c=1.638$
& $c=3.728$
\\
\hline

AAPL
& $H_{0,1}:\alpha_1^*=\alpha_2^*$
& 98.80 & 97.99 & 98.80 & 98.80 & 100.00 & 100.00
\\

& $H_{0,2}:\alpha_1^*=\alpha_3^*$
& 100.00 & 98.80 & 100.00 & 98.80 & 99.20 & 99.20
\\

& $H_{0,3}:\alpha_2^*=\alpha_3^*$
& 100.00 & 99.60 & 100.00 & 99.60 & 99.60 & 99.20
\\
\hline

META
& $H_{0,1}:\alpha_1^*=\alpha_2^*$
& 97.19 & 97.99 & 98.39 & 97.59 & 99.20 & 98.80
\\

& $H_{0,2}:\alpha_1^*=\alpha_3^*$
& 99.60 & 98.39 & 99.20 & 99.20 & 99.20 & 98.80
\\

& $H_{0,3}:\alpha_2^*=\alpha_3^*$
& 99.60 & 98.80 & 99.20 & 98.39 & 99.20 & 98.80
\\
\hline

EOG
& $H_{0,1}:\alpha_1^*=\alpha_2^*$
& 98.80 & 98.39 & 93.98 & 97.59 & 99.20 & 99.60
\\

& $H_{0,2}:\alpha_1^*=\alpha_3^*$
& 98.80 & 95.58 & 98.80 & 98.39 & 99.20 & 99.60
\\

& $H_{0,3}:\alpha_2^*=\alpha_3^*$
& 98.80 & 96.39 & 98.80 & 98.39 & 99.20 & 99.20
\\

\hline\hline
\end{tabular}
\end{table}

\section{Conclusion}
\label{secconclusion}

We have developed a parametric estimation procedure of self-exciting Hawkes processes. The procedure was based on ordinary least squares. We restricted to the kernel form which is a sum of the product of a parameter and a function. First, we gave the central limit theorem of the parametric estimation procedure. Then, we showed the consistency of the variance matrix estimation procedure. Finally, we introduced a Wald test statistic and show its asymptotic properties as a  consequence. We have applied the proposed methodology to trade times from high-frequency financial asset data. Our empirical results provided evidence for three heterogeneous trader types. By extending this approach to the parametric Hawkes processes with time dependent parameters, we can perform estimation of Hawkes processes with structural changes. Other possibilities consist in high-dimensional Hawkes processes. 

\section*{Acknowledgments}
Financial support from the Japanese Society for the Promotion of Science under grant 23H00807 (Potiron) and grant 25K16617 (Poignard) is gratefully acknowledged.

\bibliography{biblio_arxiv}

@article{ogata1978asymptotic,
  title={The asymptotic behaviour of maximum likelihood estimators for stationary point processes},
  author={Ogata, Yoshiko},
  journal={Annals of the Institute of Statistical Mathematics},
  volume={30},
  number={1},
  pages={243--261},
  year={1978},
  publisher={Springer}
}

@article{hawkes1971point,
  title={Point spectra of some mutually exciting point processes},
  author={Hawkes, Alan G},
  journal={Journal of the Royal Statistical Society. Series B (Methodological)},
  pages={438--443},
  year={1971},
  publisher={JSTOR}
}

@article{kirchner2017estimation,
  title={An estimation procedure for the {H}awkes process},
  author={Kirchner, Matthias},
  journal={Quantitative Finance},
  volume={17},
  number={4},
  pages={571--595},
  year={2017},
  publisher={Taylor \& Francis}
}

@article{fujimori2026sparse,
  title={Sparse estimators for multivariate integer-valued autoregressive models with applications to inference for Hawkes processes},
  author={Fujimori, Kou and Shiraishi, Hiroshi and Hirukawa, Junichi and Fokianos, Konstantinos},
  journal={Stochastic Processes and their Applications},
  pages={104960},
  year={2026},
  publisher={Elsevier}
}

@article{shalen1993volume,
  title={Volume, volatility, and the dispersion of beliefs},
  author={Shalen, Catherine T},
  journal={The Review of Financial Studies},
  volume={6},
  number={2},
  pages={405--434},
  year={1993},
  publisher={Oxford University Press}
}

@article{daigler1999impact,
  title={The impact of trader type on the futures volatility-volume relation},
  author={Daigler, Robert T and Wiley, Marilyn K},
  journal={The Journal of Finance},
  volume={54},
  number={6},
  pages={2297--2316},
  year={1999},
  publisher={Wiley Online Library}
}

@article{ozaki1979maximum,
  title={Maximum likelihood estimation of {H}awkes' self-exciting point processes},
  author={Ozaki, Tohru},
  journal={Annals of the Institute of Statistical Mathematics},
  volume={31},
  number={1},
  pages={145--155},
  year={1979},
  publisher={Springer}
}

@article{bowsher2007modelling,
  title={Modelling security market events in continuous time: Intensity based, multivariate point process models},
  author={Bowsher, Clive G},
  journal={Journal of Econometrics},
  volume={141},
  number={2},
  pages={876--912},
  year={2007},
  publisher={Elsevier}
}

@article{hawkes1971spectra,
  title={Spectra of some self-exciting and mutually exciting point processes},
  author={Hawkes, Alan G},
  journal={Biometrika},
  volume={58},
  number={1},
  pages={83--90},
  year={1971},
  publisher={Biometrika Trust}
}

@article{embrechts2011multivariate,
  title={Multivariate {H}awkes processes: an application to financial data},
  author={Embrechts, Paul and Liniger, Thomas and Lin, Lu},
  journal={Journal of Applied Probability},
  volume={48},
  pages={367--378},
  year={2011},
  publisher={Applied Probability Trust}
}

@article{bacry2013some,
  title={Some limit theorems for {H}awkes processes and application to financial statistics},
  author={Bacry, Emmanuel and Delattre, Sylvain and Hoffmann, Marc and Muzy, Jean-Francois},
  journal={Stochastic Processes and their Applications},
  volume={123},
  number={7},
  pages={2475--2499},
  year={2013},
  publisher={Elsevier}
}

@article{fox2016modeling,
  title={Modeling e-mail networks and inferring leadership using self-exciting point processes},
  author={Fox, Eric W and Short, Martin B and Schoenberg, Frederic P and Coronges, Kathryn D and Bertozzi, Andrea L},
  journal={Journal of the American Statistical Association},
  number={just-accepted},
  year={2016},
  publisher={Taylor \& Francis}
}

@book{van1998asymptotic,
  title={Asymptotic statistics},
  author={Van der Vaart, Aad W},
  year={1998},
  publisher={Cambridge university press}
}

@article{clinet2018statistical,
  title={Statistical inference for the doubly stochastic self-exciting process},
  author={Clinet, Simon and Potiron, Yoann},
  journal={Bernoulli},
  volume={24},
  number={4B},
  pages={3469--3493},
  year={2018},
  publisher={Bernoulli Society for Mathematical Statistics and Probability}
}

@article{Hasbrouck2013low,
  title={Low-latency trading},
  author={Hasbrouck, J and Saar, G.},
  journal={Journal of Financial Markets},
  volume={16},
  pages={646-679},
  year={2013}
}

@article{gagnon2010multi,
  title={Multi-market trading and arbitrage},
  author={Gagnon, L and Karolyi, G},
  journal={Journal of Financial Economics},
  volume={97},
  pages={53-80},
  year={2010}
}

@article{hoffmann2014dynamic,
  title={A dynamic limit order market with fast and slow traders},
  author={Hoffmann, P.},
  journal={Journal of Financial Economics},
  volume={113},
  pages={156-169},
  year={2014}
}

@article{wang2002effect,
  title={The effect of net positions by type of trader on volatility in foreign currency futures markets},
  author={Wang, Changyun},
  journal={Journal of Futures Markets: Futures, Options, and Other Derivative Products},
  volume={22},
  number={5},
  pages={427--450},
  year={2002},
  publisher={Wiley Online Library}
}

@article{kang2025analyzing,
  title={Analyzing whale calling through Hawkes process modeling},
  author={Kang, Bokgyeong and Schliep, Erin M and Gelfand, Alan E and Yack, Tina M and Clark, Christopher W and Schick, Robert S},
  journal={Journal of the American Statistical Association},
  volume={120},
  number={552},
  pages={2040--2052},
  year={2025},
  publisher={Taylor \& Francis}
}

@article{ait2026saddlepoint,
  title={Saddlepoint approximations for {H}awkes jump-diffusion processes with an application to risk management},
  author={A{\"i}t-Sahalia, Yacine and Laeven, Roger JA},
  journal={Journal of the American Statistical Association},
  number={just-accepted},
  pages={1--27},
  year={2026},
  publisher={Taylor \& Francis}
}

@article{chang2000market,
  title={Market volatility and the demand for hedging in stock index futures},
  author={Chang, Eric and Chou, Ray Y and Nelling, Edward F},
  journal={Journal of Futures Markets: Futures, Options, and Other Derivative Products},
  volume={20},
  number={2},
  pages={105--125},
  year={2000},
  publisher={Wiley Online Library}
}

@article{chang1997interday,
  title={Interday variations in volume, variance and participation of large speculators},
  author={Chang, Eric C and Pinegar, J Michael and Schachter, Barry},
  journal={Journal of Banking \& Finance},
  volume={21},
  number={6},
  pages={797--810},
  year={1997},
  publisher={Elsevier}
}

@article{potiron2026mutually,
  title={Mutually exciting point processes with latency},
  author={Potiron, Yoann and Volkov, Vladimir},
  journal={Journal of the American Statistical Association},
  volume={121},
  number={553},
  pages={326--337},
  year={2026},
  publisher={Taylor \& Francis}
}

@article{erdemlioglu2025latency,
  title={Estimation of latency for {H}awkes processes with a polynomial periodic kernel},
  author={Erdemlioglu, Deniz and Potiron, Yoann and Xu, Taiyu and Volkov, Vladimir},
  journal={Working paper},
  url={https://www.fbc.keio.ac.jp/\sim potiron/Erdemlioglu2025workingpaperestimationlatency.pdf},
  year={2025}
}

@article{cavaliere2023bootstrap,
  title={Bootstrap inference for {H}awkes and general point processes},
  author={Cavaliere, G and Lu, Y and Rahbek, A and St{\ae}rk-{\O}stergaard, J},
  journal={Journal of Econometrics},
  volume={235},
  pages={133--165},
  year={2023},
  publisher={Elsevier}
}

@article{chavez2005estimating,
  title={Estimating value-at-risk: a point process approach},
  author={Chavez-Demoulin, V and Davison, AC and McNeil, AJ},
  journal={Quantitative Finance},
  volume={5},
  pages={227--234},
  year={2005},
  publisher={Taylor \& Francis}
}

@article{rubin1972regular,
  title={Regular point processes and their detection},
  author={Rubin, I},
  journal={IEEE Transactions on Information Theory},
  volume={18},
  pages={547--557},
  year={1972},
  publisher={IEEE}
}

@article{vere1978earthquake,
  title={Earthquake prediction-a statistician's view},
  author={Vere-Jones, D},
  journal={Journal of Physics of the Earth},
  volume={26},
  pages={129--146},
  year={1978},
  publisher={The Seismological Society of Japan, The Volcanological Society of Japan, The~…}
}

@article{vere1982some,
  title={Some examples of statistical estimation applied to earthquake data: I. Cyclic Poisson and self-exciting models},
  author={Vere-Jones, D and Ozaki, T},
  journal={Annals of the Institute of Statistical Mathematics},
  volume={34},
  pages={189--207},
  year={1982},
  publisher={Springer}
}

@article{ait2015modeling,
  title={Modeling financial contagion using mutually exciting jump processes},
  author={A{\"i}t-Sahalia, Y and Cacho-Diaz, J and Laeven, RJA},
  journal={Journal of Financial Economics},
  volume={117},
  pages={585--606},
  year={2015},
  publisher={Elsevier}
}

@article{ikefuji2022earthquake,
  title={Earthquake risk embedded in property prices: Evidence from five {J}apanese cities},
  author={Ikefuji, M and Laeven, RJA and Magnus, JR and Yue, Y},
  journal={Journal of the American Statistical Association},
  volume={117},
  pages={82--93},
  year={2022},
  publisher={Taylor \& Francis}
}

@article{corradi2020testing,
  title={Testing for jump spillovers without testing for jumps},
  author={Corradi, V and Distaso, W and Fernandes, M},
  journal={Journal of the American Statistical Association},
  volume={115},
  pages={1214--1226},
  year={2020},
  publisher={Taylor \& Francis}
}

@article{clinet2017statistical,
  title={Statistical inference for ergodic point processes and application to limit order book},
  author={Clinet, S and Yoshida, N},
  journal={Stochastic Processes and their Applications},
  volume={127},
  pages={1800--1839},
  year={2017},
  publisher={Elsevier}
}

@ARTICLE{Large2007,
  author = {Large, J.},
  title = {Measuring the resiliency of an electronic limit order book},
  journal = {Journal of Financial Markets},
  year = {2007},
  volume = {10},
  pages = {1-25}
}

@ARTICLE{AitSahalia2014mutual,
  author = {A{\"i}t-Sahalia, Y and Laeven, RJA and Pelizzon, L},
  title = {Mutual excitation in eurozone sovereign {CDS}},
  journal = {Journal of Econometrics},
  year = {2014},
  volume = {183},
  pages = {151-167}
}

@article{fulop2015self,
  title={Self-exciting jumps, learning, and asset pricing implications},
  author={Fulop, Andras and Li, Junye and Yu, Jun},
  journal={Review of Financial Studies},
  volume={28},
  number={3},
  pages={876--912},
  year={2015},
  publisher={Oxford University Press}
}

@article{yu2004empirical,
  title={Empirical characteristic function estimation and its applications},
  author={Yu, Jun},
  journal={Econometric {R}eviews},
  volume={23},
  number={2},
  pages={93--123},
  year={2004},
  publisher={Taylor \& Francis}
}

@article{kwan2023alternative,
  title={Alternative asymptotic inference theory for a nonstationary {H}awkes process},
  author={Kwan, TKJ and Chen, F and Dunsmuir, WTM},
  journal={Journal of Statistical Planning and Inference},
  volume={227},
  pages={75--90},
  year={2023},
  publisher={Elsevier}
}

@article{cheysson2022spectral,
  title={Spectral estimation of {H}awkes processes from count data},
  author={Cheysson, Felix and Lang, Gabriel},
  journal={Annals of Statistics},
  volume={50},
  number={3},
  pages={1722--1746},
  year={2022},
  publisher={Institute of Mathematical Statistics}
}

@article{donnet2020nonparametric,
author = {Donnet, Sophie and Rivoirard, Vincent and Rousseau, Judith},
title = {Nonparametric {B}ayesian estimation for multivariate {H}awkes processes},
volume = {48},
journal = {Annals of Statistics},
number = {5},
publisher = {Institute of Mathematical Statistics},
pages = {2698--2727},
year = {2020}
}

@book{daley2003introduction,
  title={An introduction to the theory of point processes: Elementary theory and methods},
  publisher={Springer New York, NY},
  author={Daley, Daryl J and Vere-Jones, David},
  year={2003},
  volume = {1},
  edition = {2nd}
}

@BOOK{daley2008introduction,
  title = {An introduction to the theory of point processes: General theory
	and structure},
  publisher = {Springer New York, NY},
  year = {2008},
  author={Daley, Daryl J and Vere-Jones, David},
  volume = {2},
  edition = {2nd}
}

@book{jacod2003limit,
  title={Limit theorems for stochastic processes},
  author={Jacod, Jean and Shiryaev, Albert},
  year={2003},
  publisher={Springer Berlin, Heidelberg},
  edition = {2nd}
}

@article{hawkes1974cluster,
  title={A cluster process representation of a self-exciting process},
  author={Hawkes, Alan G and Oakes, David},
  journal={Journal of Applied Probability},
  volume={11},
  number={3},
  pages={493--503},
  year={1974},
  publisher={Cambridge University Press}
}

@article{jacod1975multivariate,
  title={Multivariate point processes: predictable projection, {R}adon-{N}ikodym derivatives, representation of martingales},
  author={Jacod, Jean},
  journal={Zeitschrift f{\"u}r Wahrscheinlichkeitstheorie und verwandte Gebiete},
  volume={31},
  number={3},
  pages={235--253},
  year={1975},
  publisher={Springer}
}

@article{budish2015high,
  title={The high-frequency trading arms race: Frequent batch auctions as a market design response},
  author={Budish, Eric and Cramton, Peter and Shim, John},
  journal={The Quarterly Journal of Economics},
  volume={130},
  number={4},
  pages={1547--1621},
  year={2015},
  publisher={MIT Press}
}

@article{biais2015equilibrium,
  title={Equilibrium fast trading},
  author={Biais, Bruno and Foucault, Thierry and Moinas, Sophie},
  journal={Journal of Financial Economics},
  volume={116},
  number={2},
  pages={292--313},
  year={2015},
  publisher={Elsevier}
}

@article{foucault2016news,
  title={News trading and speed},
  author={Foucault, Thierry and Hombert, Johan and Ro{\c{s}}u, Ioanid},
  journal={Journal of Finance},
  volume={71},
  number={1},
  pages={335--382},
  year={2016},
  publisher={Wiley Online Library}
}

@article{pagnotta2018competing,
  title={Competing on speed},
  author={Pagnotta, Emiliano S and Philippon, Thomas},
  journal={Econometrica},
  volume={86},
  number={3},
  pages={1067--1115},
  year={2018},
  publisher={Wiley Online Library}
}

@article{potiron2026hawkes,
  title={Parametric inference for {H}awkes processes with a general kernel},
  author={Potiron, Yoann},
  journal={Working paper},
  url={https://www.fbc.keio.ac.jp/\sim potiron/potiron2026hawkesworkingpaper.pdf},
  year={2026}
}

@book{white2001asymptotic,
  title={Asymptotic Theory for Econometricians},
  author={White, H.},
  isbn={9780127466521},
  lccn={00107735},
  series={Economic Theory, Econometrics, and Mathematical Economics},
  url={https://books.google.co.jp/books?id=FyG3AAAAIAAJ},
  year={2001},
  publisher={Emerald Group Publishing Limited}
}

@book{doob1953stochastic,
  title={Stochastic processes},
  author={Doob, J.L.},
  year={1953},
  publisher={Wiley New York}
}

@article{morariu2022state,
  title={State-dependent {H}awkes processes and their application to limit order book modelling},
  author={Morariu-Patrichi, Maxime and Pakkanen, Mikko S},
  journal={Quantitative Finance},
  volume={22},
  number={3},
  pages={563--583},
  year={2022},
  publisher={Taylor \& Francis}
}

@article{ogata1988statistical,
  title={Statistical models for earthquake occurrences and residual analysis for point processes},
  author={Ogata, Yosihiko},
  journal={Journal of the American Statistical Association},
  volume={83},
  number={401},
  pages={9--27},
  year={1988},
  publisher={Taylor \& Francis}
}

@article{bajgrowicz2016jumps,
  title={Jumps in high-frequency data: Spurious detections, dynamics, and news},
  author={Bajgrowicz, Pierre and Scaillet, Olivier and Treccani, Adrien},
  journal={Management Science},
  volume={62},
  number={8},
  pages={2198--2217},
  year={2016},
  publisher={INFORMS}
}

@article{cai2024latent,
  title={Latent network structure learning from high-dimensional multivariate point processes},
  author={Cai, Biao and Zhang, Jingfei and Guan, Yongtao},
  journal={Journal of the American Statistical Association},
  volume={119},
  number={545},
  pages={95--108},
  year={2024},
  publisher={Taylor \& Francis}
}

@article{kling2026goodness,
  title={On goodness-of-fit testing for self-exciting point processes},
  author={Kling, Jos{\'e} Carlos Fontanesi and Vetter, Mathias},
  journal={To appear in the Scandinavian Journal of Statistics},
  year={2026},
  publisher={Wiley Online Library}
}

@article{reynaud2010adaptive,
  title={Adaptive estimation for {H}awkes processes; application to genome analysis},
  author={Reynaud-Bouret, Patricia and Schbath, Sophie},
  journal={Annals of Statistics},
  volume={38},
  number={5},
  pages={2781--2822},
  year={2010}
}

@article{fang2024group,
  title={Group network {H}awkes process},
  author={Fang, Guanhua and Xu, Ganggang and Xu, Haochen and Zhu, Xuening and Guan, Yongtao},
  journal={Journal of the American Statistical Association},
  volume={119},
  number={547},
  pages={2328--2344},
  year={2024},
  publisher={Taylor \& Francis}
}

@article{bacry2020sparse,
  title={Sparse and low-rank multivariate {H}awkes processes},
  author={Bacry, Emmanuel and Bompaire, Martin and Ga{\"i}ffas, St{\'e}phane and Muzy, Jean-Francois},
  journal={Journal of Machine Learning Research},
  volume={21},
  number={50},
  pages={1--32},
  year={2020}
}

@article{hansen2015lasso,
  title={Lasso and probabilistic inequalities for multivariate point processes},
  author={Hansen, Niels Richard and Reynaud-Bouret, Patricia and Rivoirard, Vincent},
  year={2015},
  journal={Bernoulli},
  volume={21},
  number={1},
  pages={83--143},
  publisher={Bernoulli Society for Mathematical Statistics and Probability}
}

@article{rosenblatt1956central,
  title={A central limit theorem and a strong mixing condition},
  author={Rosenblatt, Murray},
  journal={Proceedings of the national Academy of Sciences},
  volume={42},
  number={1},
  pages={43--47},
  year={1956}
}

@book{rosenblatt1978dependence,
  title={Dependence and asymptotic independence for random processes. In Studies in Probability Theory},
  author={Rosenblatt, Murray},
  year={1978},
  publisher={Mathematical
Association of America, Washington, D.C.}
}

@book{stout1974almost,
  title={Almost Sure Convergence},
  author={Stout, W F},
  year={1974},
  publisher={Academic Press, New York}
}

@article{rosenblatt1972uniform,
  title={Uniform ergodicity and strong mixing},
  author={Rosenblatt, Murray},
  journal={Zeitschrift f{\"u}r Wahrscheinlichkeitstheorie und Verwandte Gebiete},
  volume={24},
  number={1},
  pages={79--84},
  year={1972},
  publisher={Springer}
}

@article{mcleish1974dependent,
  title={Dependent central limit theorems and invariance principles},
  author={McLeish, Donald L},
  journal={Annals of Probability},
  volume={2},
  number={4},
  pages={620--628},
  year={1974},
  publisher={Institute of Mathematical Statistics}
}

@book{doukhan1995mixing,
  title={Mixing: Properties and Examples},
  author={Doukhan, Paul},
  year={1995},
  publisher={Springer-Verlag New York}
}

@book{bradley2007introduction,
  title={Introduction to Strong Mixing Conditions},
  author={Bradley, Richard C.},
  year={2007},
  volume={1},
  publisher={Kendrick Press}
}

@article{bremaud1996stability,
  title={Stability of nonlinear {H}awkes processes},
  author={Br{\'e}maud, Pierre and Massouli{\'e}, Laurent},
  journal={Annals of Probability},
  pages={1563--1588},
  year={1996},
  publisher={JSTOR}
}

\newpage

\appendix

\Huge \noindent \textbf{Appendices}

\normalsize
\noindent This part corresponds to the Appendices of "Parametric estimation of Hawkes processes based on ordinary least squares". In Section \ref{secsimulation}, we perform simulation experiments. Additional real data experiments are available in Section \ref{secempiricalstudyinsample}. All the proofs of the theory can be found in Section \ref{secproofs}. 

\section{Simulation experiments}
\label{secsimulation}

In this section, we evaluate the performance of the proposed estimation method through a series of simulation studies. In Subsection \ref{appendix_subsec:sim_l2error}, we assess the estimation accuracy of the Hawkes process parameters obtained by maximum likelihood and ordinary least squares. In particular, we examine the impact of the tuning parameter $\Delta$ on the ordinary least squares estimation performance. 
In Subsection \ref{appendix_subsec:sim_wald}, we investigate the finite-sample performance of the Wald test for linear constraints to illustrate Corollary \ref{corwald} of the main text. An additional numerical experiment illustrating the asymptotic distribution derived in Theorem \ref{thcltest} of the main text is provided in Subsection \ref{appendix_subsec:asymptotic_dist}. 

\subsection{Precision of estimation}\label{appendix_subsec:sim_l2error}

We consider the setting where the true underlying data generating process (DGP) for $\{N_t\}_{1\leq t\leq T}$ has intensity given by (\ref{defhawkesintensity}). We generate $200$ independent batches of observations $\{N_t\}_{1\leq t\leq T}$ for given true values $\nu^*$, $\alpha^*=(\alpha^*_1,\ldots,\alpha^*_{n_h})'$ and $b^*=(b^*_1,\ldots,b^*_{n_h})'$. We set $T=10000$, which represents approximately half a day trading in seconds excluding the first and last 30 minutes. 
For each batch, once a path $\{N_t\}_{1\leq t\leq T}$ is generated from the true DGP, we solve (\ref{stat_crit}) for a given time increment $\Delta$, yielding $\widehat{\theta}_{T\Delta}$. For the sake of comparison, we also compute the maximum likelihood estimator $\widehat{\theta}^{\text{mle}}_{T\Delta}$. 
In both estimation procedures, the true decay $b^*$ is fixed and assumed known, so that the ordinary least squares and maximum likelihood procedures estimate $\nu^*$ and the vector $\alpha^*$. 

Obviously, the choice of the time increment $\Delta$ is a key ingredient in the estimation of $\theta$ based on ordinary least squares, as discussed in Section \ref{secimplementation}. We will assess the sensitivity of the estimation accuracy with respect to the choice of the time increment $\Delta$. We set the time increment as $\Delta = c/b^*_{\max}$ in which the increment parameter $c \in \mathcal{C}= \{10^{-6.5},\ldots,10^{-1.57}\}$,
with 70 logarithmically equally spaced grid points. Denoting by $\|v\|_2=\sqrt{\sum^p_{i=1}v^2_i}$ for $v \in \mathbb{R}^p$ the $\ell_2$ norm, we will evaluate the estimation accuracy as the $\ell_2$-error $\|\widehat{\theta}_{T\Delta}-\theta^*\|_2$ for a given $\Delta$ and $\|\widehat{\theta}^{\text{mle}}_{T\Delta}-\theta^*\|_2$. This distance will be averaged over the $200$ batches. 
For each $n_h$, the true parameters $\alpha^*,b^*$ satisfy $\sum^{n_h}_{k=1}\alpha^*_k/b^*_k<1$. We will consider two scenarios for the branching ratio $\text{BR}_{n_h}=\| h \|_1=\sum^{n_h}_{k=1}\alpha^*_k/b^*_k$, namely $\text{BR}_{n_h}\approx 0.8$ and $\text{BR}_{n_h} \approx 0.98$. 
And for each $\text{BR}_{n_h}$ case, we will consider settings with very large $b^*_{\max}$ and capped $b^*_{\max}$ values. Finally, for each ($\text{BR}_{n_h}$, $b^*_{\max}$) design, we will consider two settings with respect to the ratios $\alpha^*_k/b^*_k$: (i) homogeneous, namely the situation where the ratios $\alpha^*_k/b^*_k$ have similar levels; (ii) heterogeneous, namely the situation allowing different ratios among $\alpha^*_k/b^*_k$. For all these settings, we consider two situations: $\nu^* = 5$ and $\nu^* =1$.
The true maximum value $b^*_{\max}$ of $b^*$ is set as $5000$ in the capped decay number design. In the large decay number case, $b^*_{\max}=16000$ when $n_h=2$, $b^*_{\max}=20000$ when $n_h=3,4$ and $b^*_{\max}=30000$ when $n_h=5$. This range of scenarios is realistic and adequate with the empirical results. 

Figures~\ref{fig:mu5-br08} and \ref{fig:mu5-br098} display $\|\widehat{\theta}_{T\Delta}-\theta^*\|_2$ as a function of the increment parameter $c$ and $\|\widehat{\theta}^{\text{mle}}_{T\Delta}-\theta^*\|_2$ when $\nu^* = 5$. For the sake of comparison with the estimator based on ordinary least squares for each time increment, the $\ell_2$-error resulting from the maximum likelihood estimation which is averaged over the $200$ batches is displayed as a constant black solid line. Each point represents an average of $200$ $\ell_2$-errors. Across all simulation settings, the ordinary least squares estimator exhibits a similar pattern. For sufficiently small values of the parameter $c$, the estimation error remains nearly constant over a broad range of discretization levels. This suggests that the estimator is not particularly sensitive to the choice of the time increment $\Delta$ provided that the discretization is sufficiently fine. However, the estimation error rises rapidly as the parameter $c$ increases. Indeed, larger values of the parameter $c$ correspond to coarser time discretizations as the time increment is given by $\Delta=c/b^*_{\max}$. This results in a less accurate approximation of the continuous-time Hawkes dynamic. The transition from the stable region to the rapidly increasing region occurs at slightly different values of the parameter $c$ depending on the branching ratio and the kernel configuration. In particular, the same overall behavior is consistently observed across all simulation designs. For $n_h \leq 3$, the ordinary least squares estimator achieves an estimation accuracy comparable to that of the maximum likelihood estimator over a wide range of discretization parameters. For $n_h=4$, the relative performance depends on the simulation design, with the ordinary least squares estimator sometimes outperforming and slightly under performing the maximum likelihood estimator. Finally, the ordinary least squares estimator generally yields smaller average $\ell_2$-errors than the maximum likelihood estimator for $n_h=5$, particularly in the capped decay number $b^*_{\max}$ designs. Overall, these results suggest that the proposed regression-based estimator remains competitive across a broad range of Hawkes specifications and exhibits limited sensitivity to the precise choice of the parameter $c$ over a sufficiently large interval of time increments.

The situation is altered when $\nu^*=1$ as suggested by Figures~\ref{fig:mu1-br08} and \ref{fig:mu1-br098} which report the corresponding results when $\nu^*=1$ while all other parameters remain unchanged. The same behavior with respect to the choice of the time increment $\Delta$ can be noticed. For sufficiently small values of $c$, the $\ell_2$-error remains relatively stable before increasing rapidly as the discretization becomes coarser. However, the maximum likelihood method consistently outperforms the ordinary least squares procedure in terms of $\ell_2$-error compared with $\nu^*=5$. This deterioration can be analyzed through the stationary mean intensity given by $\nu^*/(1-\text{BR})$ with $\text{BR}=\| h \|_1$ the branching ratio. Reducing $\nu^*$ from $5$ to $1$ decreases the expected number of observed events by approximately a factor of five. The resulting regression problem contains substantially more empty intervals and the regressors $h_{t,k}$ defined in Section \ref{secimplementation} of the main text exhibit less variation. This reduces the amount of information available for estimating the excitation parameters. The maximum likelihood method exploits the exact event times through the continuous-time likelihood and therefore suffers less from the reduced sample of events. Nevertheless, the estimator based on ordinary least squares retains the same qualitative robustness with respect to the choice of the discretization parameter. This suggests that the scaling $\Delta=c/b^*_{\max}$ remains robust even in lower-intensity settings.

\begin{figure}[p]
\centering

\simheading{Large $b^*_{\max}$}{Homogeneous}
\hfill
\simheading{Large $b^*_{\max}$}{Heterogeneous}
\hfill
\simheading{Capped $b^*_{\max}$}{Homogeneous}
\hfill
\simheading{Capped $b^*_{\max}$}{Heterogeneous}

\vspace{0.25em}

\simpanel
{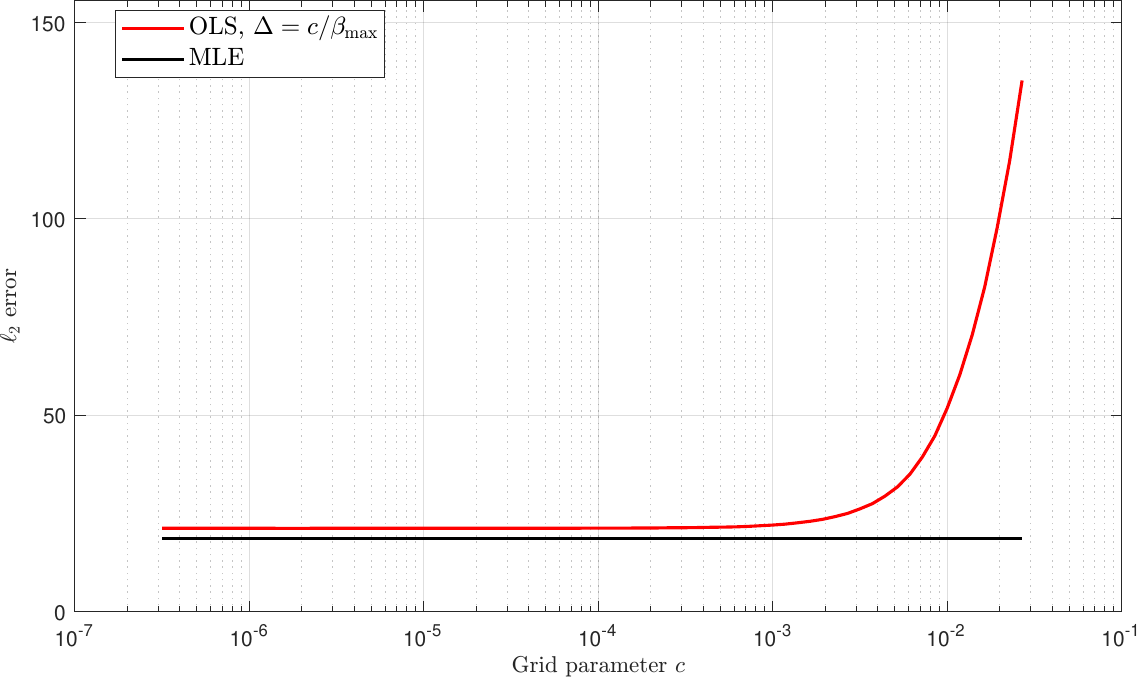}
{2}
{72,6400}
{180,16000}
{0.40,0.40}
\hfill
\simpanel
{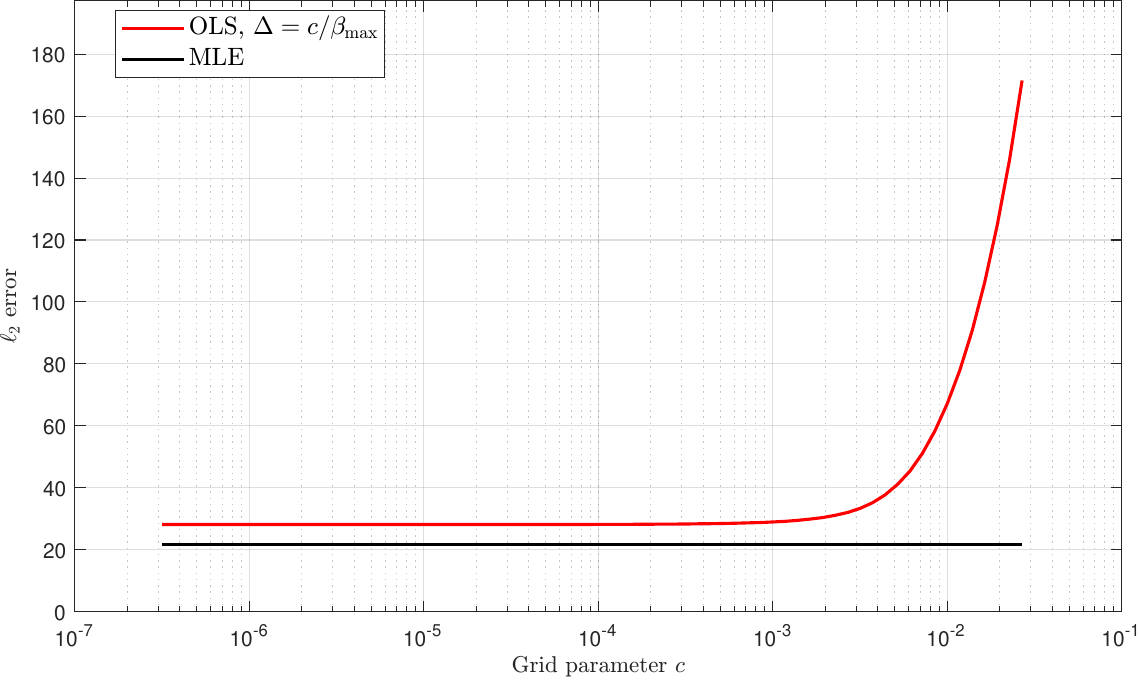}
{2}
{45,8800}
{180,16000}
{0.25,0.55}
\hfill
\simpanel
{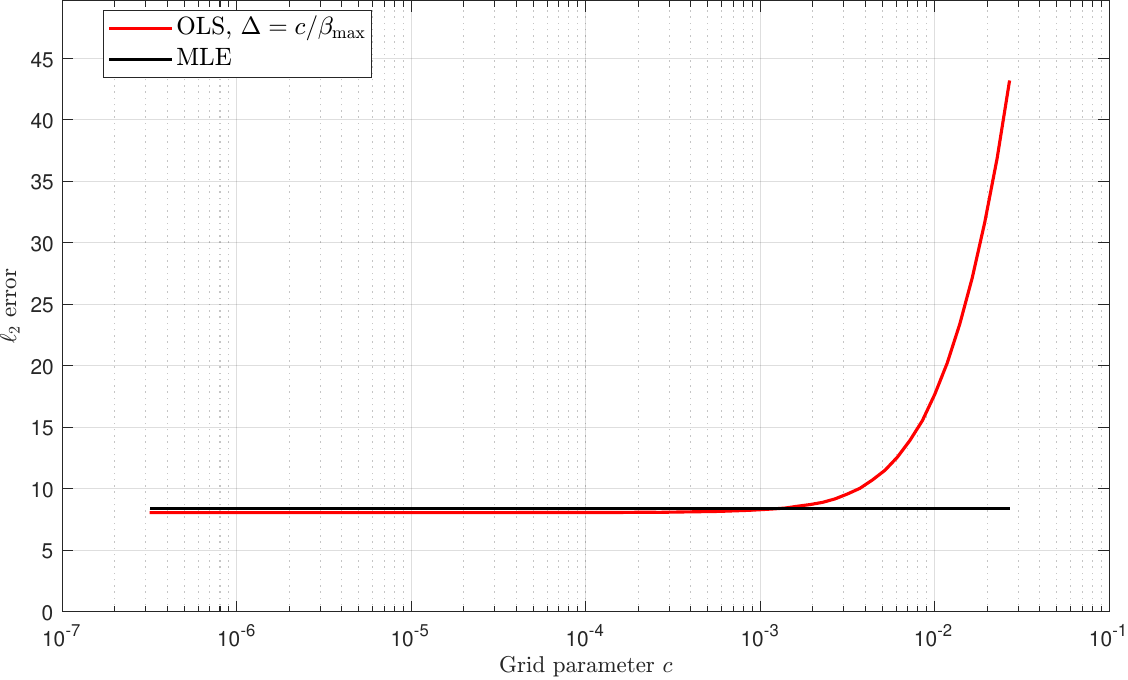}
{2}
{80,2000}
{200,5000}
{0.40,0.40}
\hfill
\simpanel
{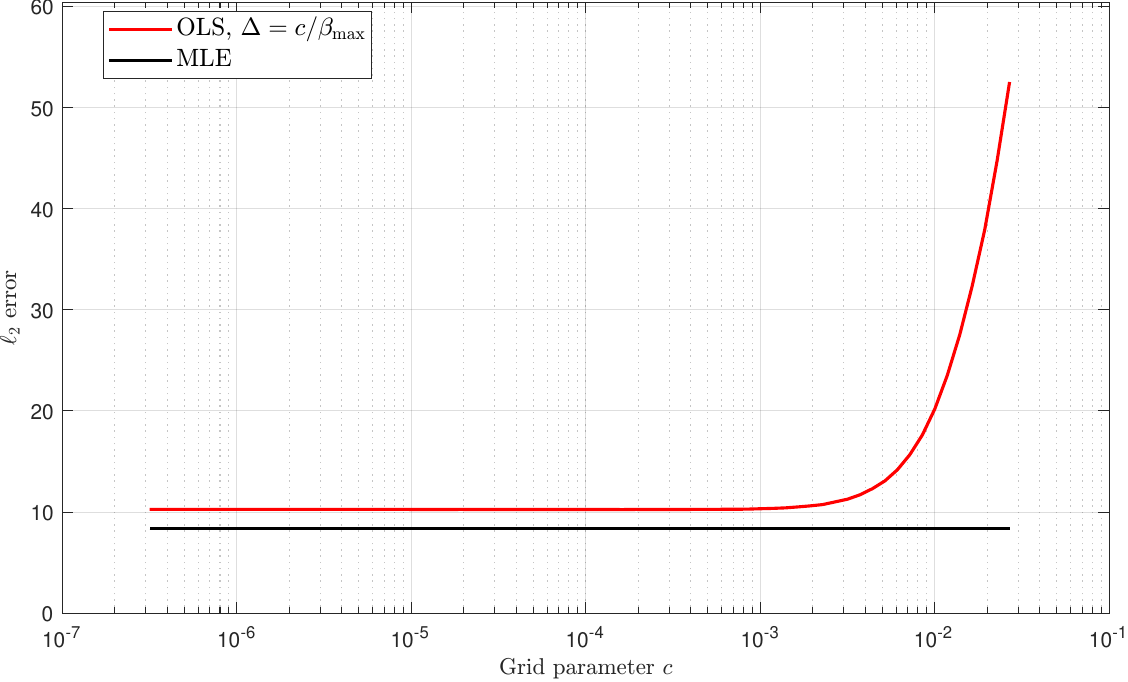}
{2}
{50,2750}
{200,5000}
{0.25,0.55}

\par\vspace{0.35em}

\simpanel
{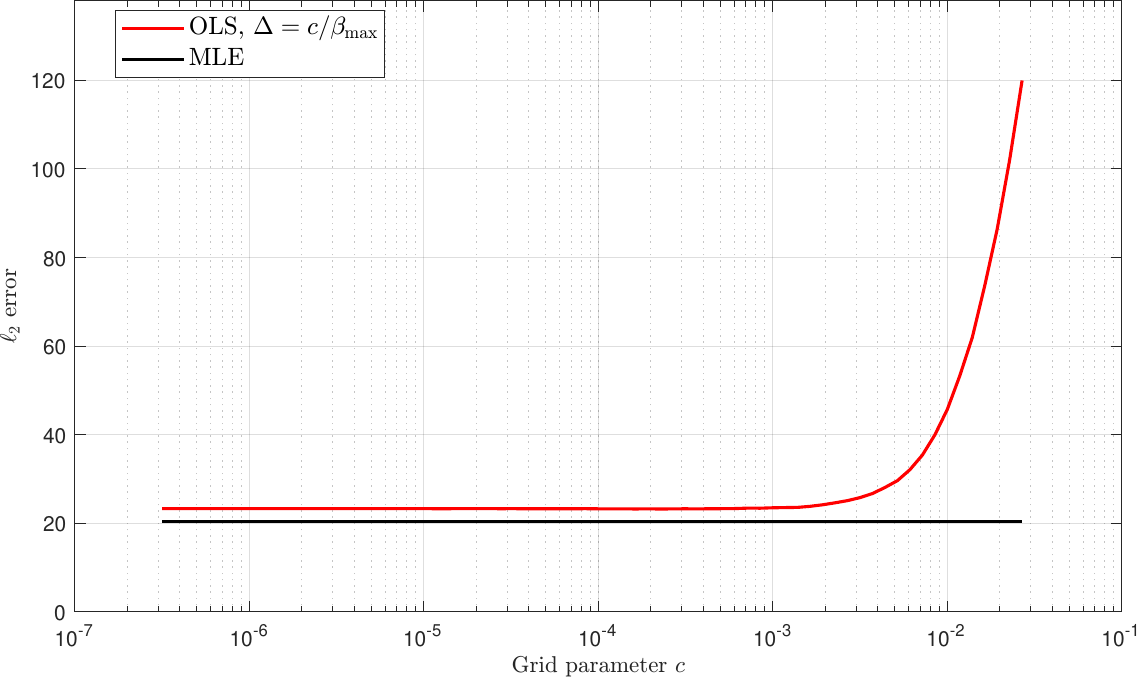}
{3}
{19,200,5332}
{70,750,20000}
{0.27,0.27,0.27}
\hfill
\simpanel
{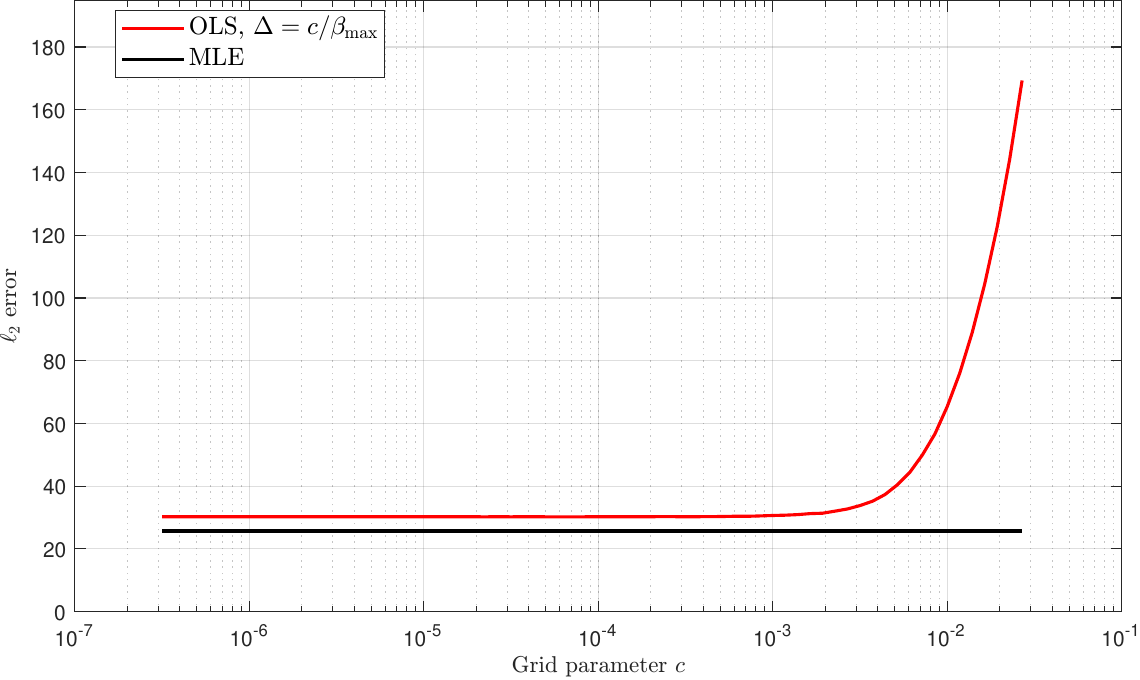}
{3}
{11,188,8000}
{70,750,20000}
{0.16,0.25,0.40}
\hfill
\simpanel
{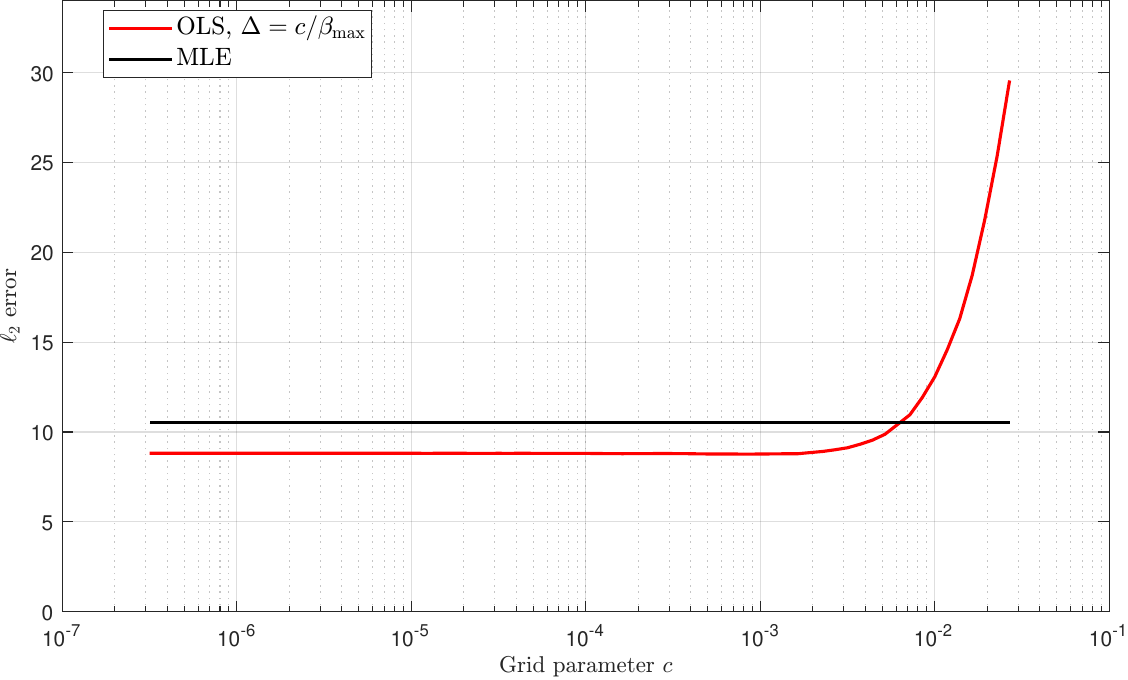}
{3}
{19,200,1333}
{70,750,5000}
{0.27,0.27,0.27}
\hfill
\simpanel
{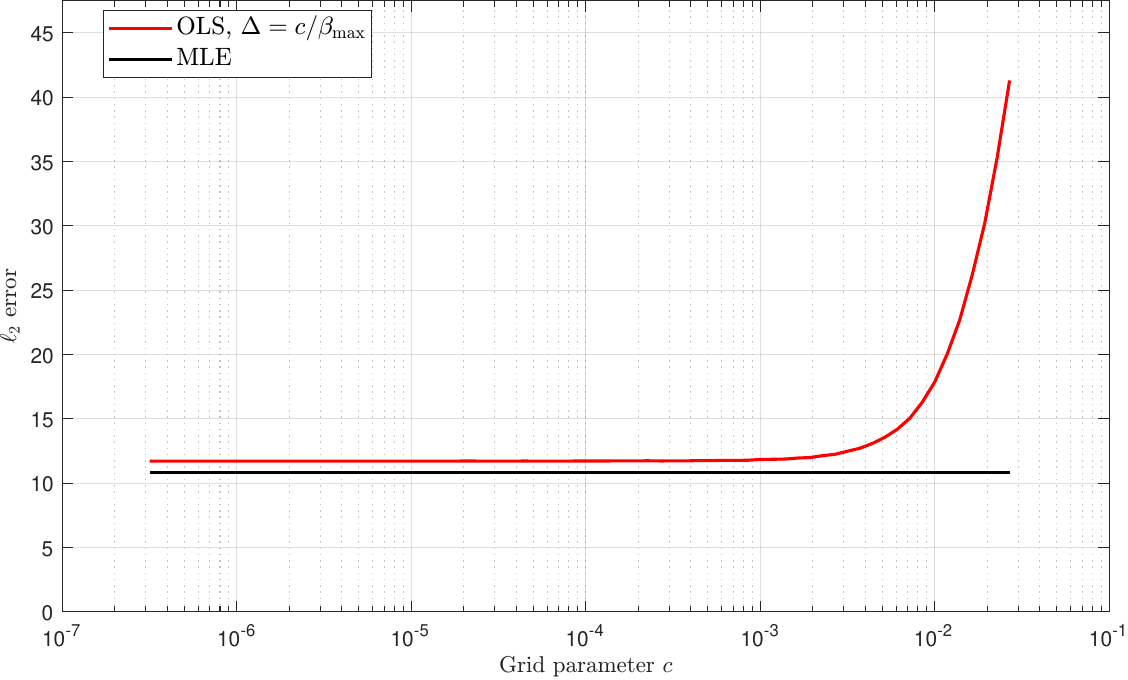}
{3}
{11,188,2000}
{70,750,5000}
{0.16,0.25,0.40}

\par\vspace{0.35em}

\simpanel
{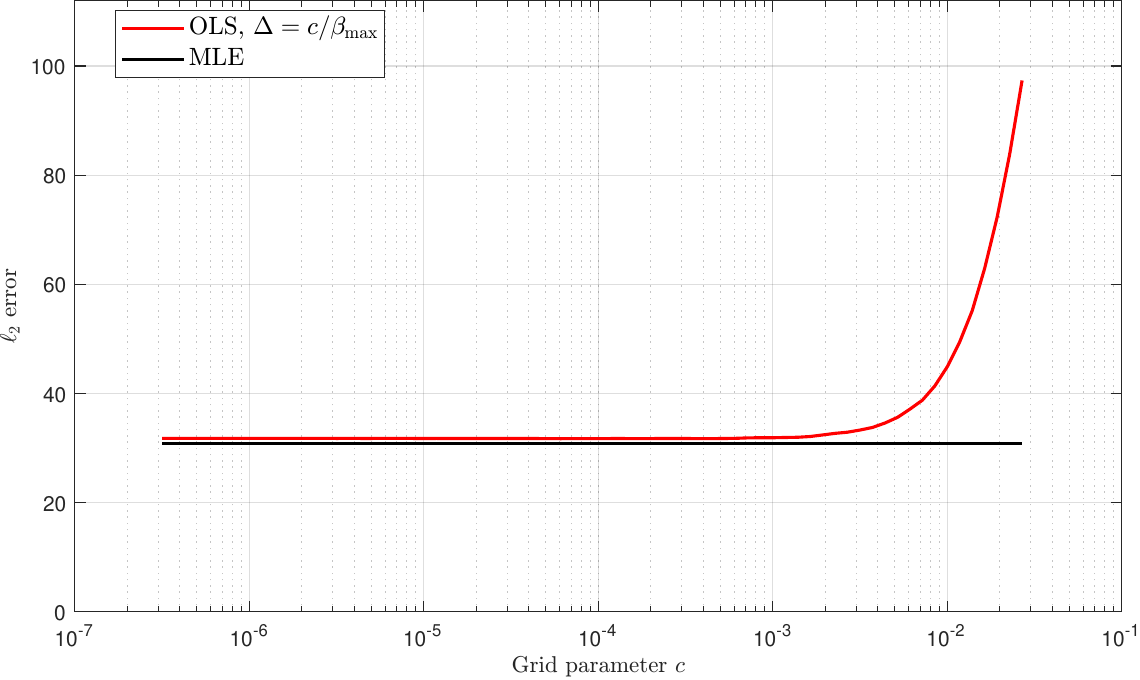}
{4}
{50,100,400,4000}
{250,500,2000,20000}
{0.20,0.20,0.20,0.20}
\hfill
\simpanel
{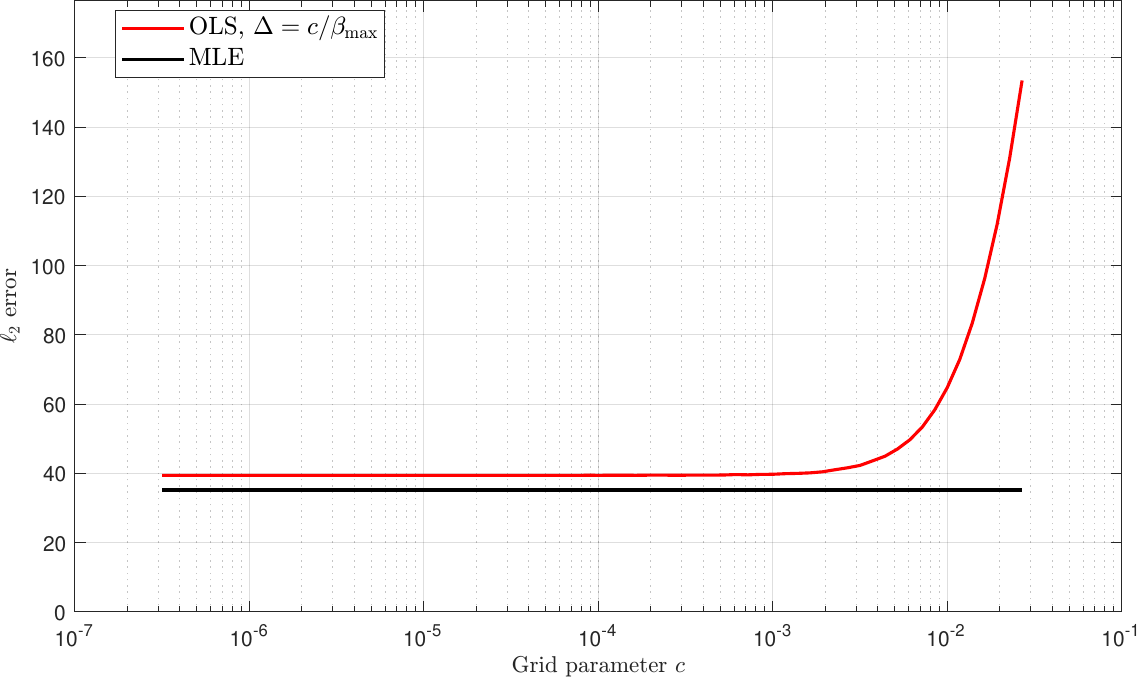}
{4}
{20,70,460,7000}
{250,500,2000,20000}
{0.08,0.14,0.23,0.35}
\hfill
\simpanel
{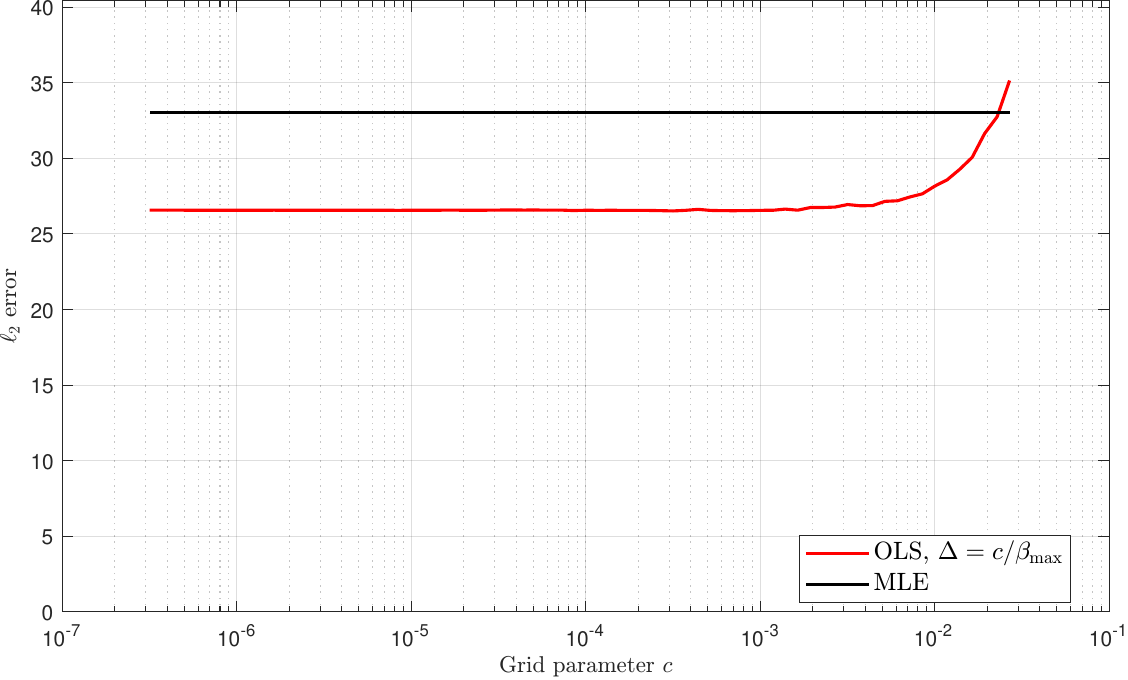}
{4}
{50,100,400,1000}
{250,500,2000,5000}
{0.20,0.20,0.20,0.20}
\hfill
\simpanel
{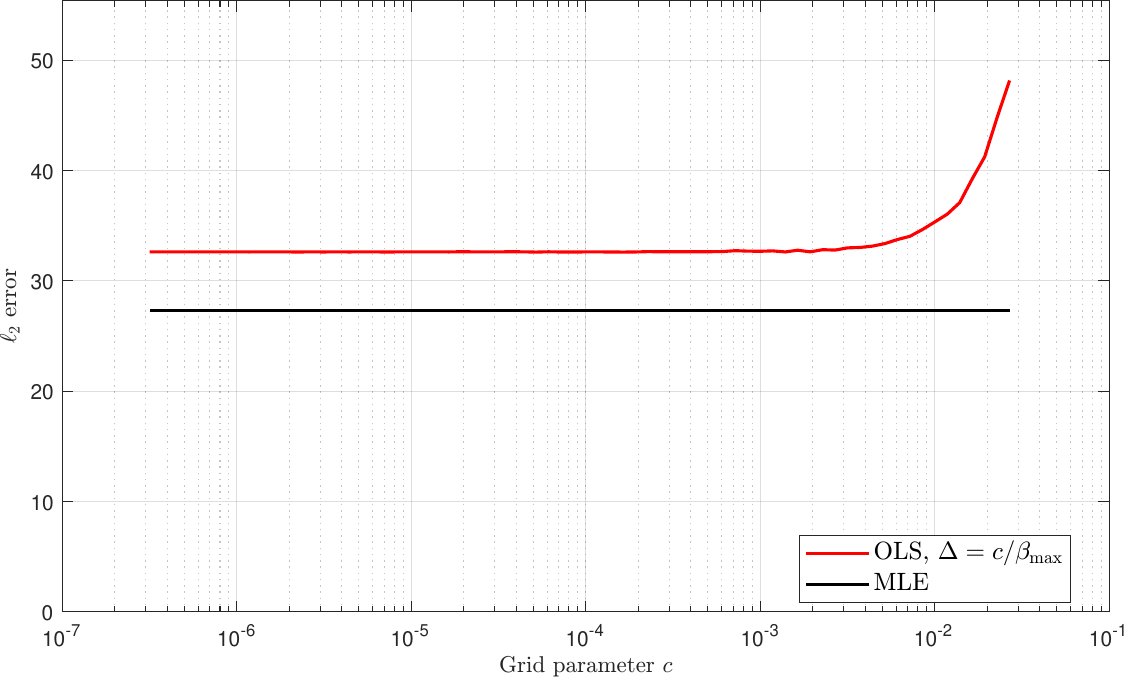}
{4}
{20,70,460,1750}
{250,500,2000,5000}
{0.08,0.14,0.23,0.35}

\par\vspace{0.35em}

\simpanel
{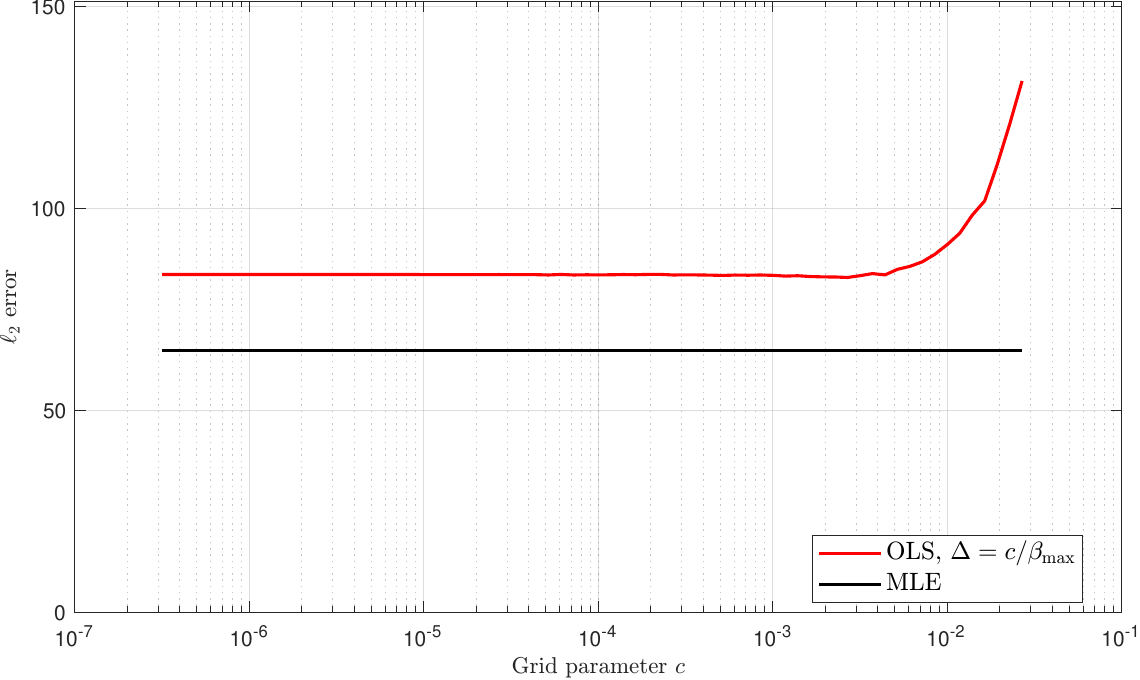}
{5}
{16,80,320,1600,4800}
{100,500,2000,10000,30000}
{0.16,0.16,0.16,0.16,0.16}
\hfill
\simpanel
{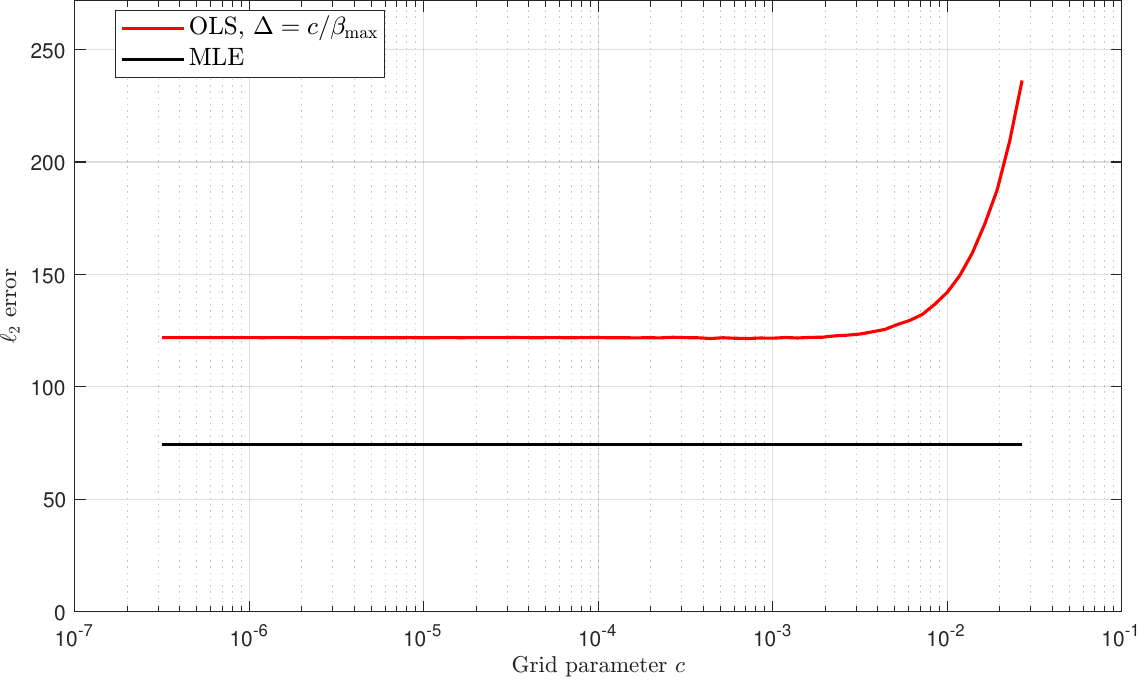}
{5}
{3,35,240,2300,10500}
{100,500,2000,10000,30000}
{0.03,0.07,0.12,0.23,0.35}
\hfill
\simpanel
{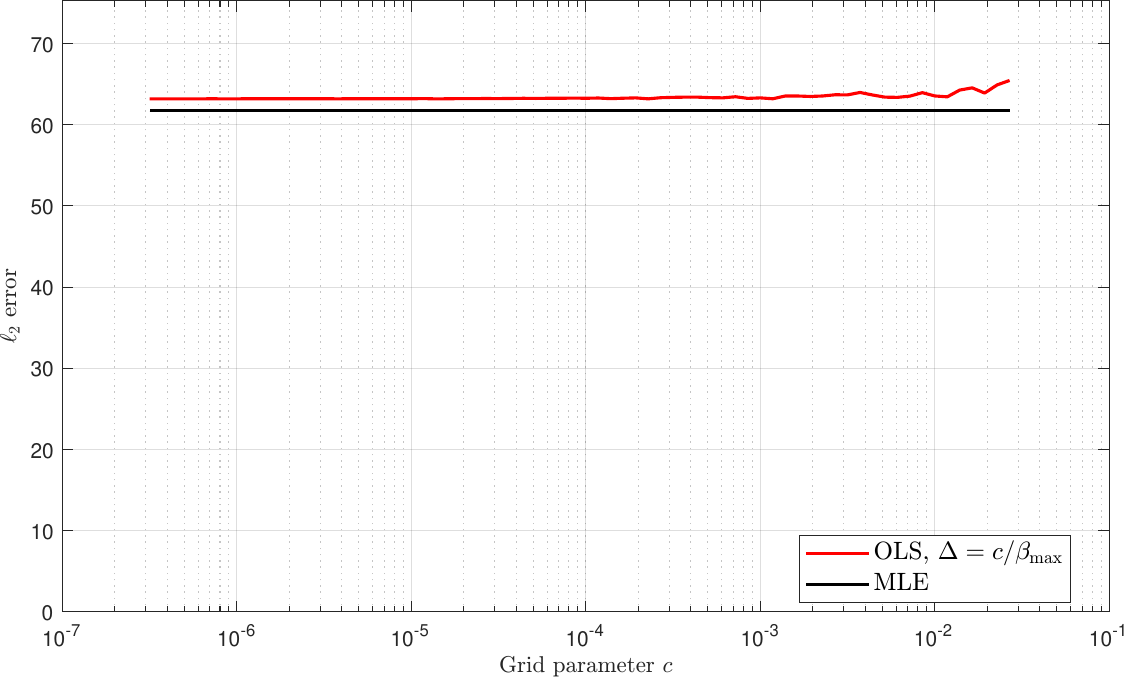}
{5}
{16,48,160,480,800}
{100,300,1000,3000,5000}
{0.16,0.16,0.16,0.16,0.16}
\hfill
\simpanel
{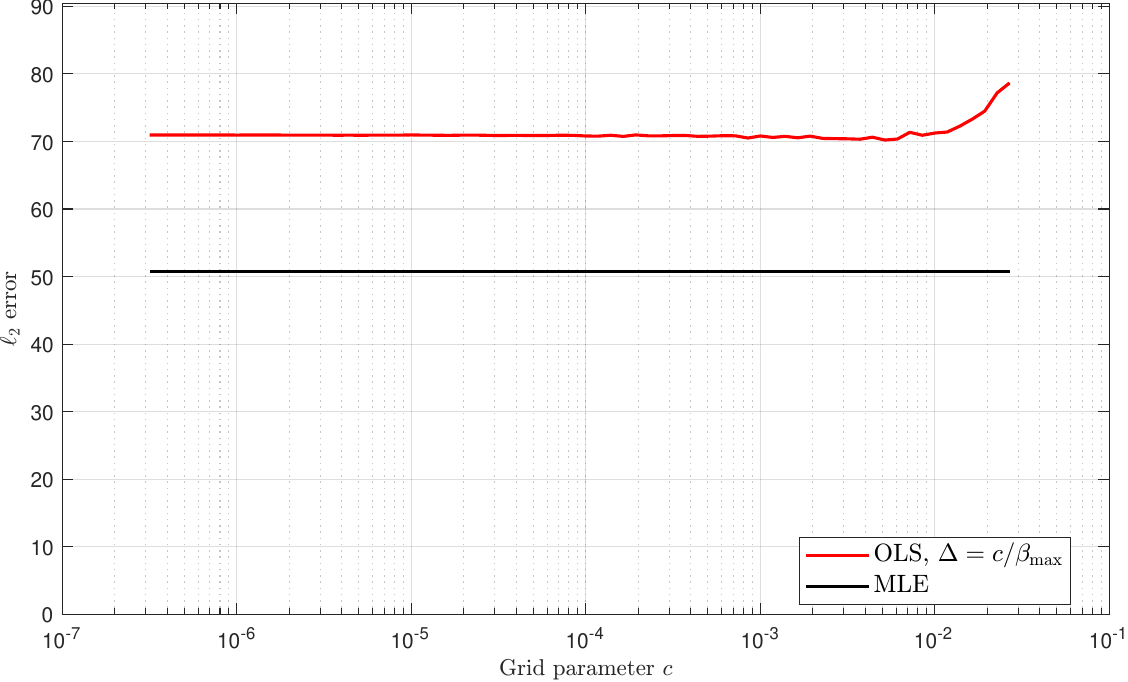}
{5}
{3,21,120,690,1750}
{100,300,1000,3000,5000}
{0.03,0.07,0.12,0.23,0.35}

\vspace{0.2em}

\caption{Average $\ell_2$-errors based on 200 
replications, $\nu^*=5$ and $\text{BR}_{n_h}\approx0.8$.
Each row corresponds to a number of kernel functions $n_h$. The first two columns correspond to large $b^*_{\max}$ design; the last two columns correspond to capped $b^*_{\max}$ design. For each design, homogeneous and heterogeneous branching-ratios are reported separately. The horizontal black line denotes the average estimation error of the MLE.}
\label{fig:mu5-br08}
\end{figure}

\begin{figure}[p]
\centering

\simheading{Large $b^*_{\max}$}{Homogeneous}
\hfill
\simheading{Large $b^*_{\max}$}{Heterogeneous}
\hfill
\simheading{Capped $b^*_{\max}$}{Homogeneous}
\hfill
\simheading{Capped $b^*_{\max}$}{Heterogeneous}

\vspace{0.25em}

\simpanel
{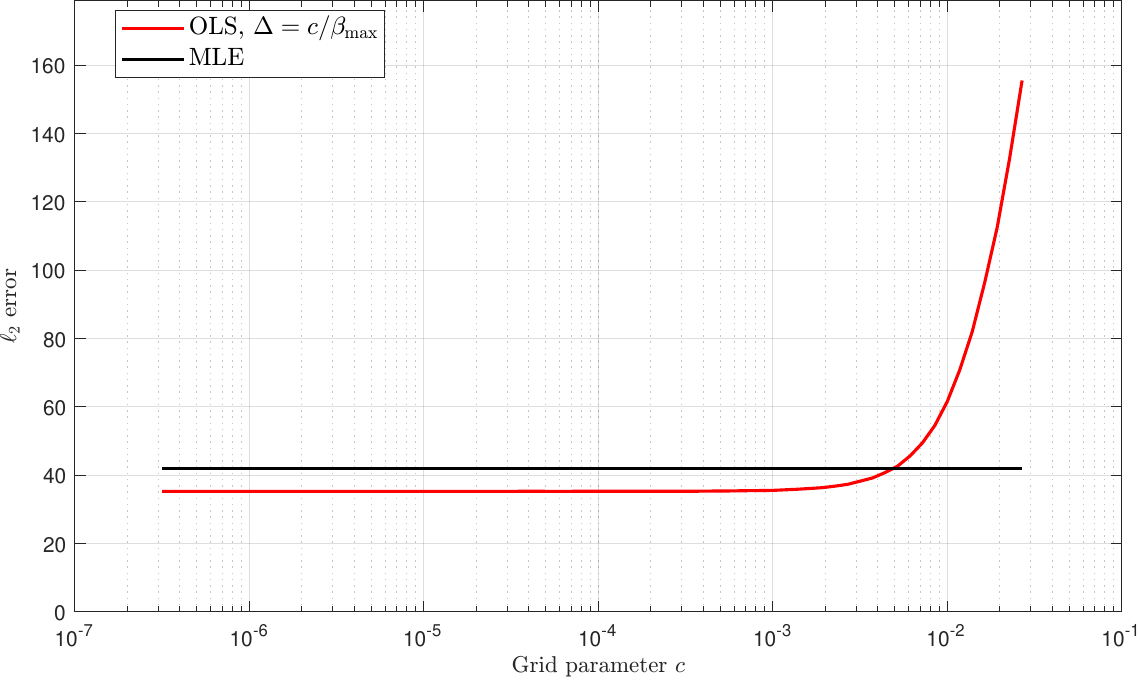}
{2}
{88,7858}
{180,16000}
{0.49,0.49}
\hfill
\simpanel
{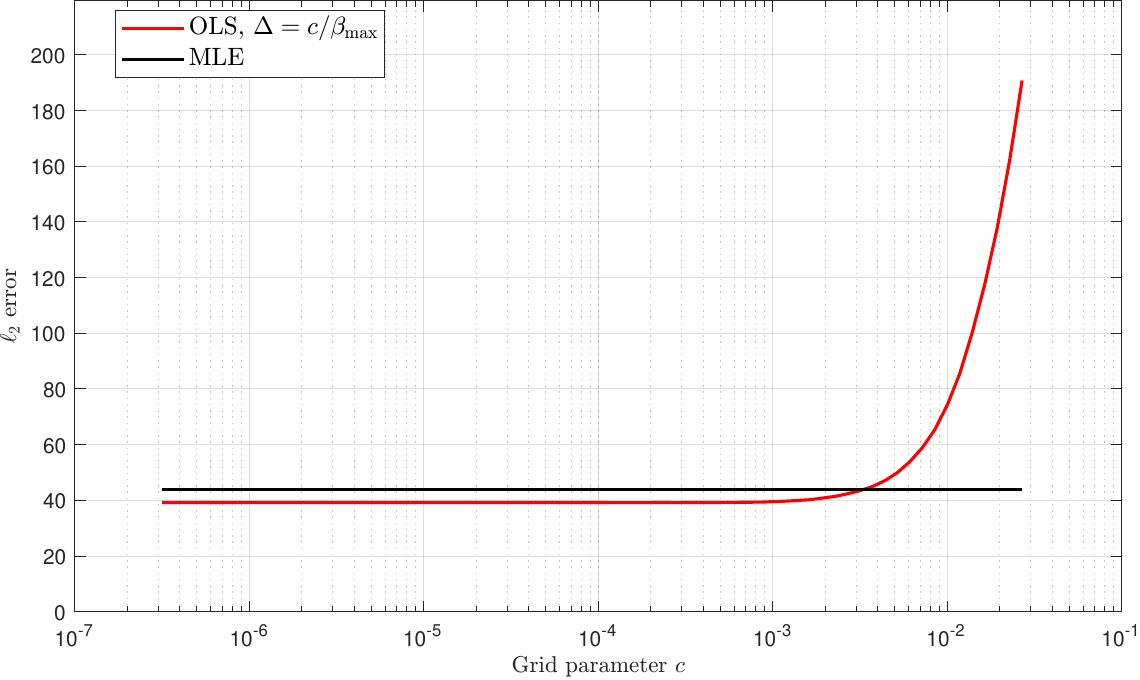}
{2}
{55,10792}
{180,16000}
{0.31,0.68}
\hfill
\simpanel
{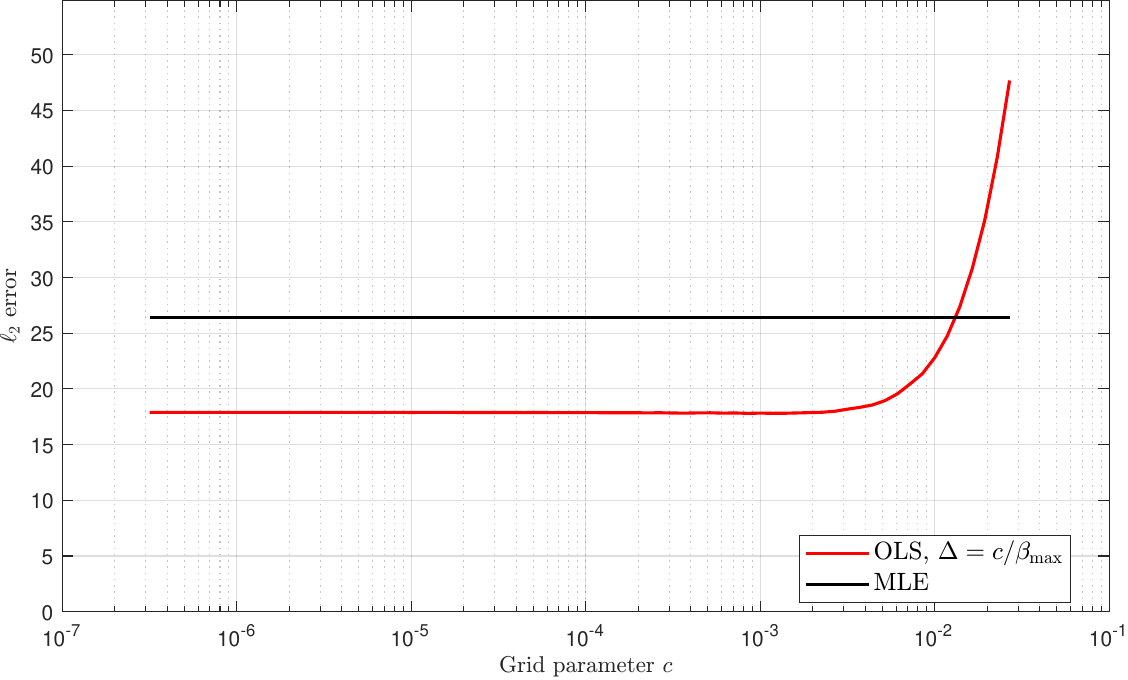}
{2}
{98,2450}
{200,5000}
{0.49,0.49}
\hfill
\simpanel
{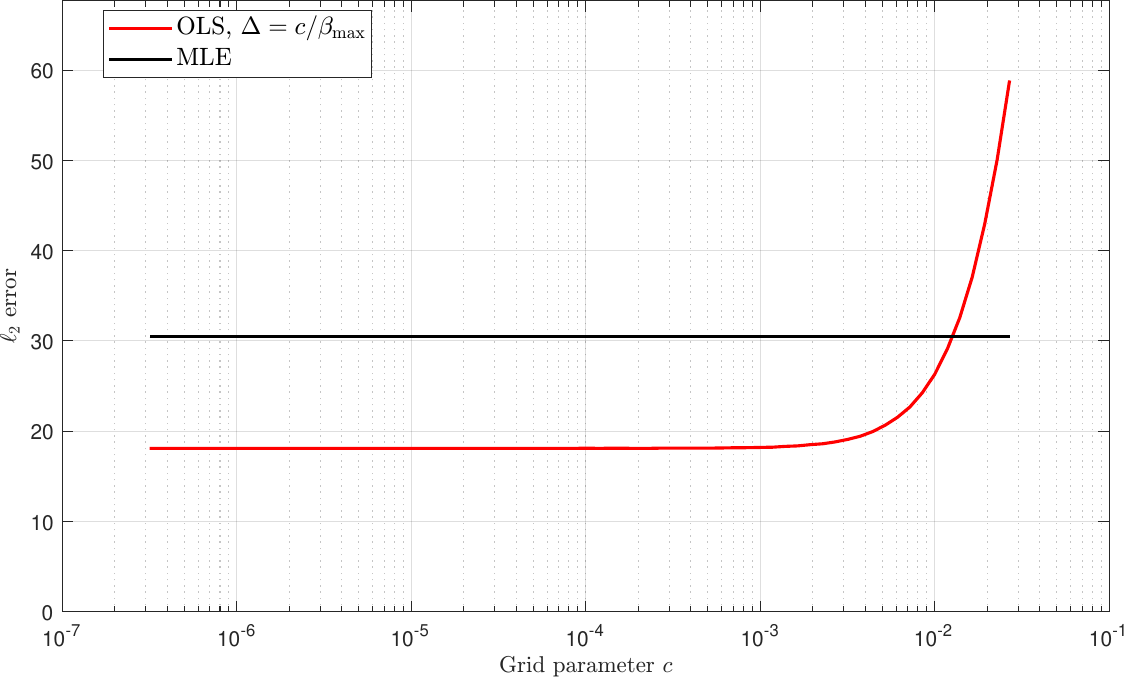}
{2}
{61,3375}
{200,5000}
{0.31,0.68}

\par\vspace{0.35em}

\simpanel
{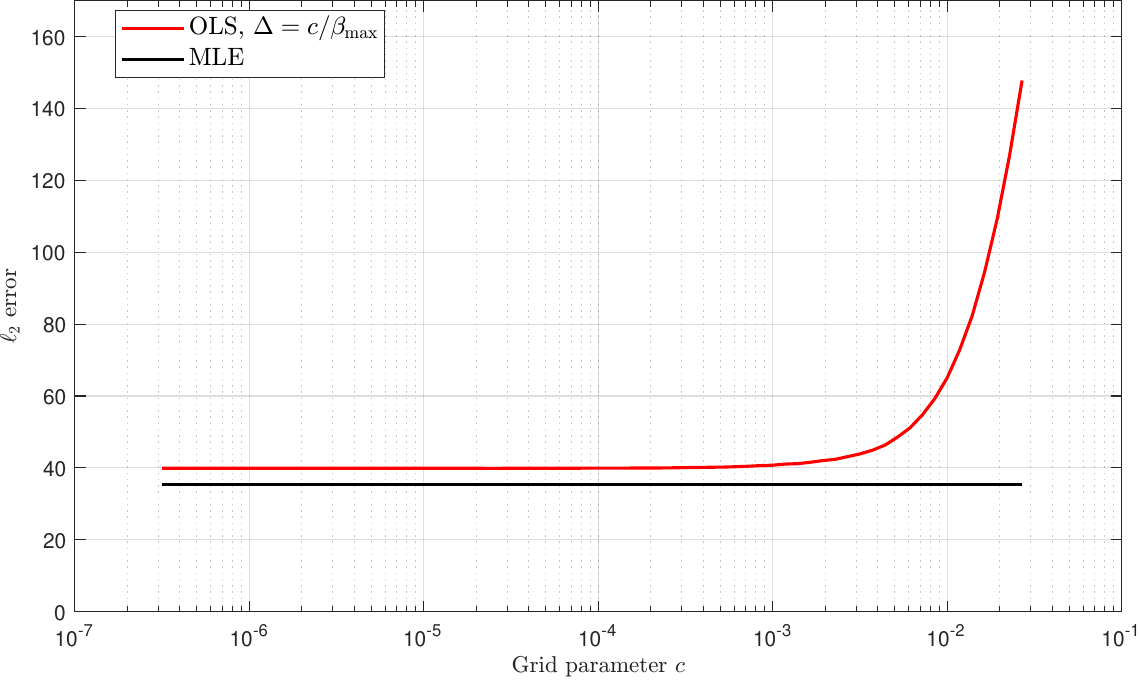}
{3}
{23,245,6532}
{70,750,20000}
{0.33,0.33,0.33}
\hfill
\simpanel
{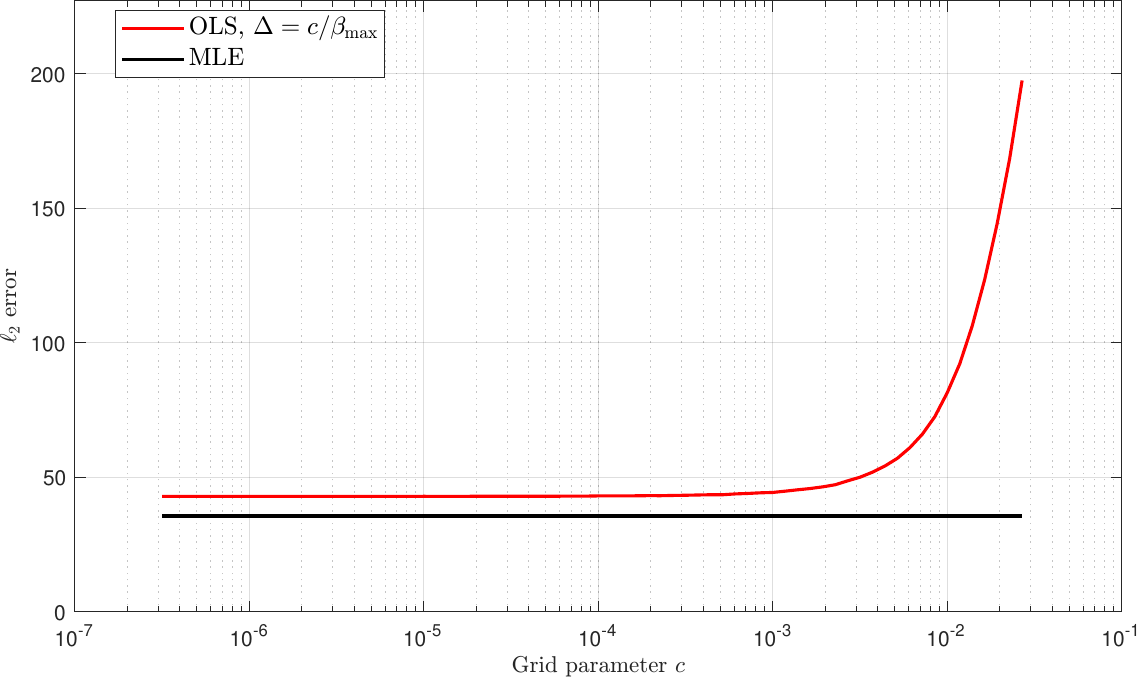}
{3}
{13,230,9800}
{70,750,20000}
{0.19,0.31,0.49}
\hfill
\simpanel
{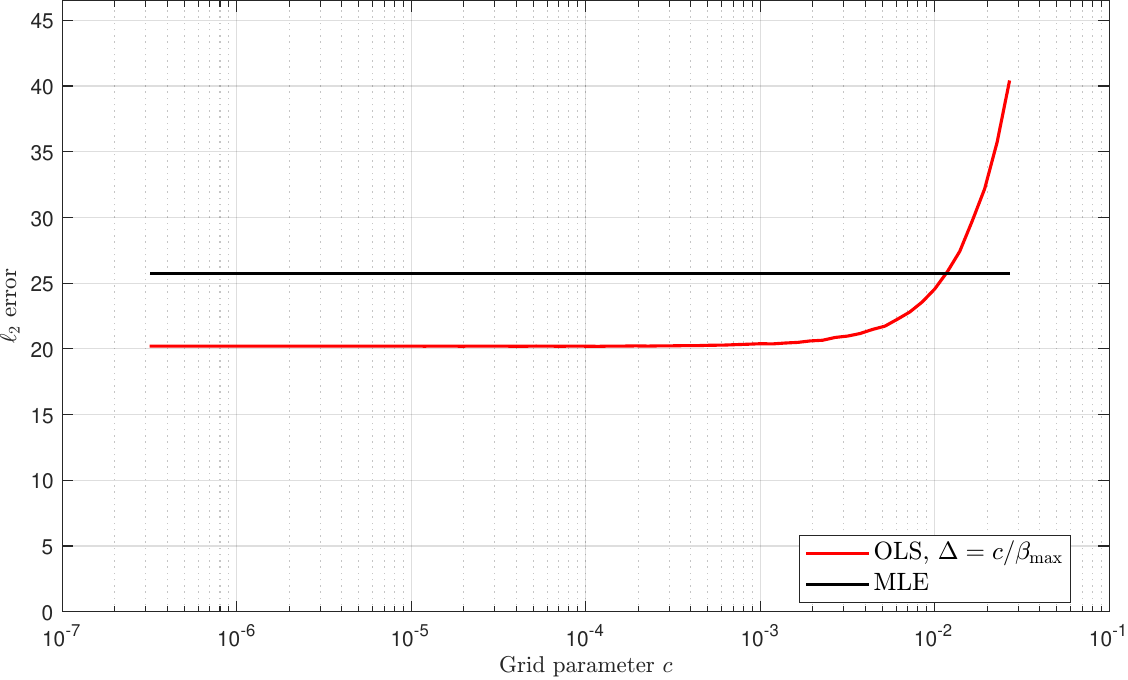}
{3}
{23,245,1633}
{70,750,5000}
{0.33,0.33,0.33}
\hfill
\simpanel
{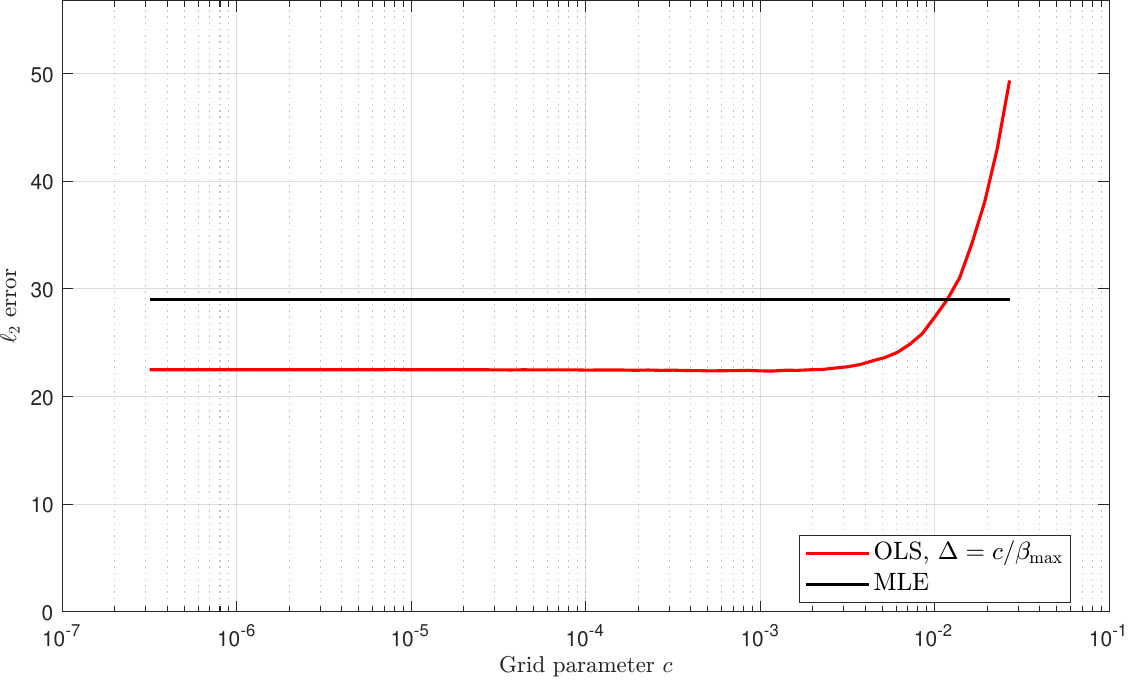}
{3}
{13,230,2450}
{70,750,5000}
{0.19,0.31,0.49}

\par\vspace{0.35em}

\simpanel
{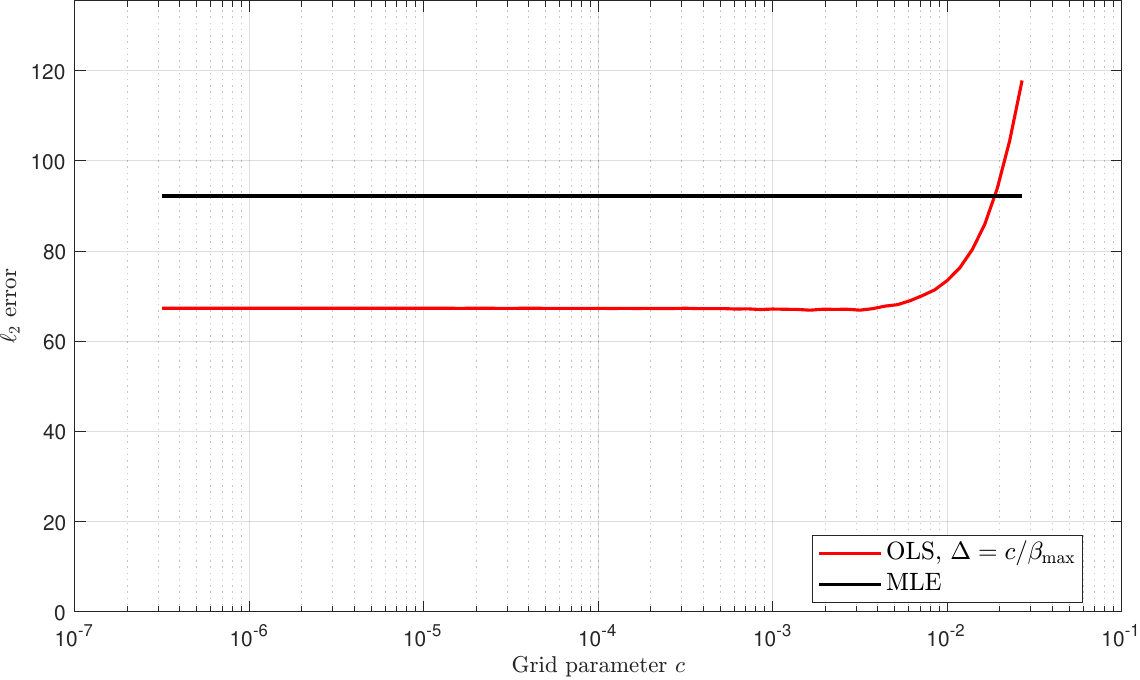}
{4}
{61,123,490,4900}
{250,500,2000,20000}
{0.24,0.25,0.25,0.25}
\hfill
\simpanel
{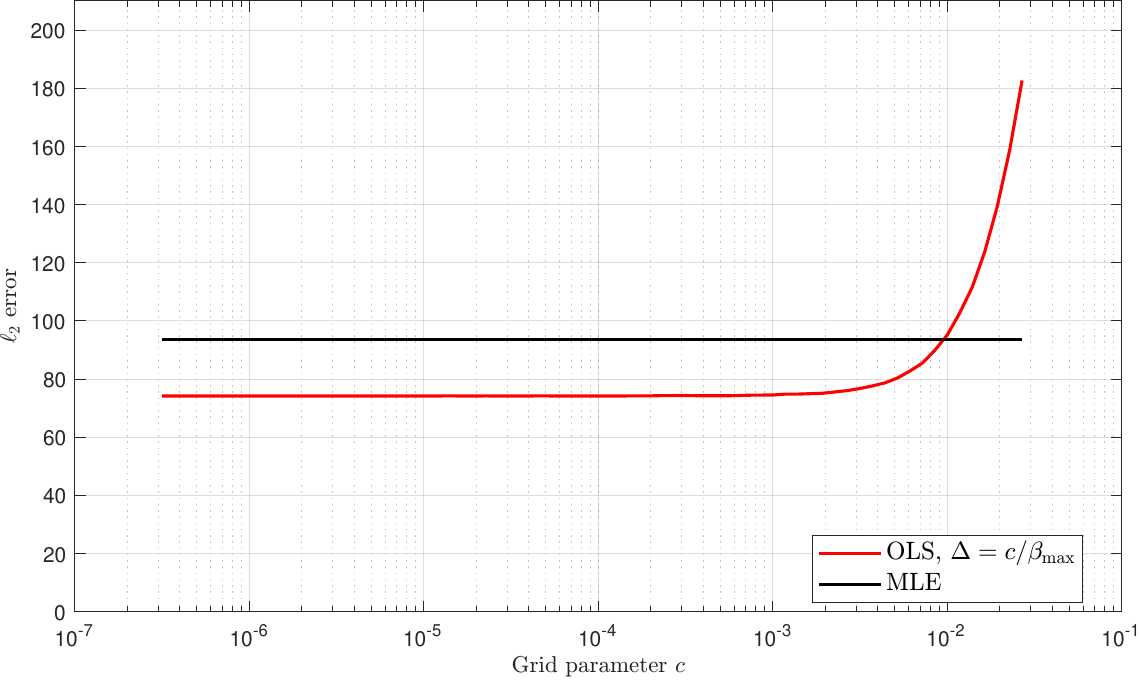}
{4}
{24,86,564,8600}
{250,500,2000,20000}
{0.10,0.17,0.28,0.43}
\hfill
\simpanel
{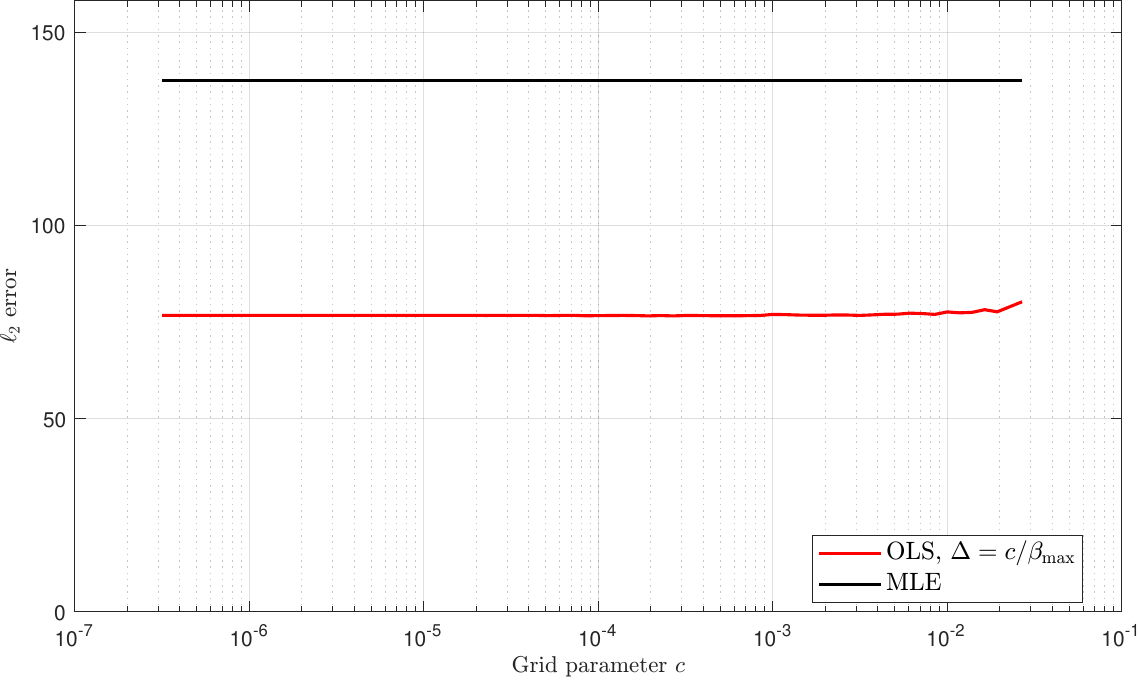}
{4}
{61,123,490,1225}
{250,500,2000,5000}
{0.24,0.25,0.25,0.25}
\hfill
\simpanel
{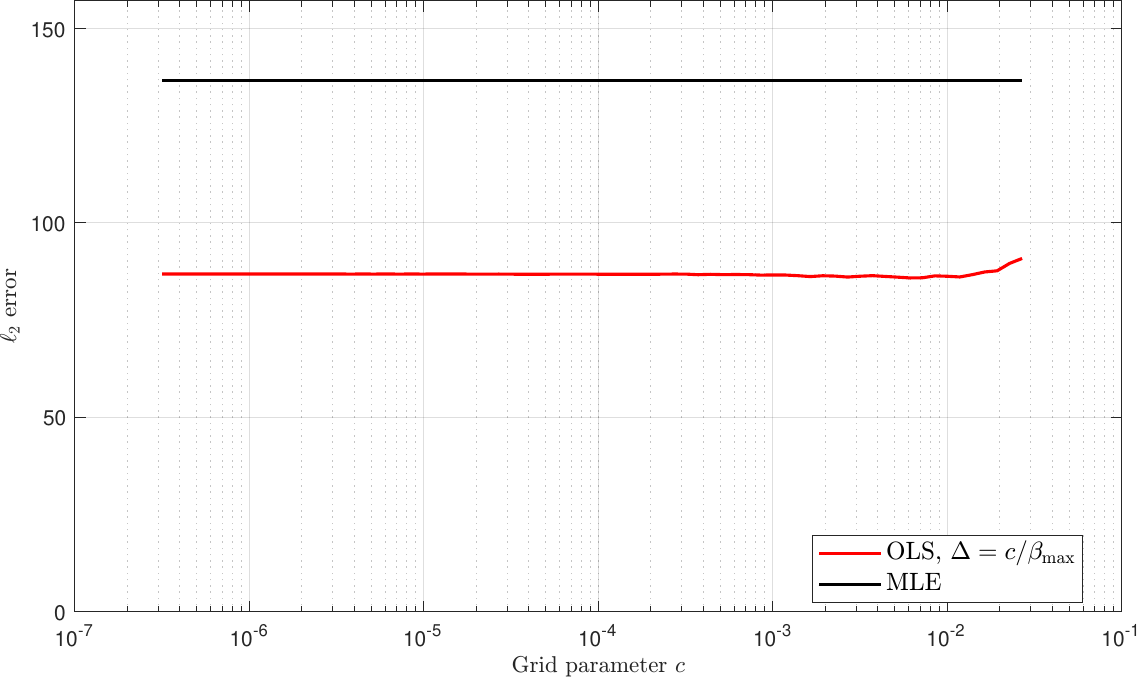}
{4}
{24,86,564,2150}
{250,500,2000,5000}
{0.1,0.17,0.28,0.43}

\par\vspace{0.35em}

\simpanel
{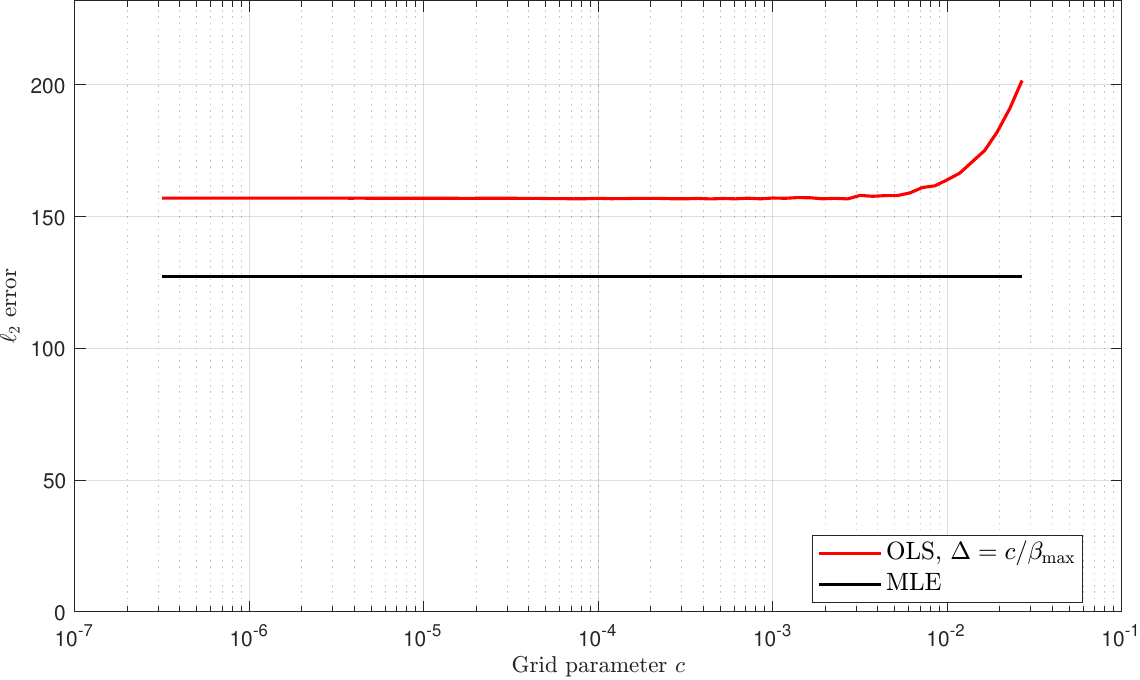}
{5}
{20,98,392,1960,5880}
{100,500,2000,10000,30000}
{0.20,0.20,0.20,0.20,0.20}
\simpanel
{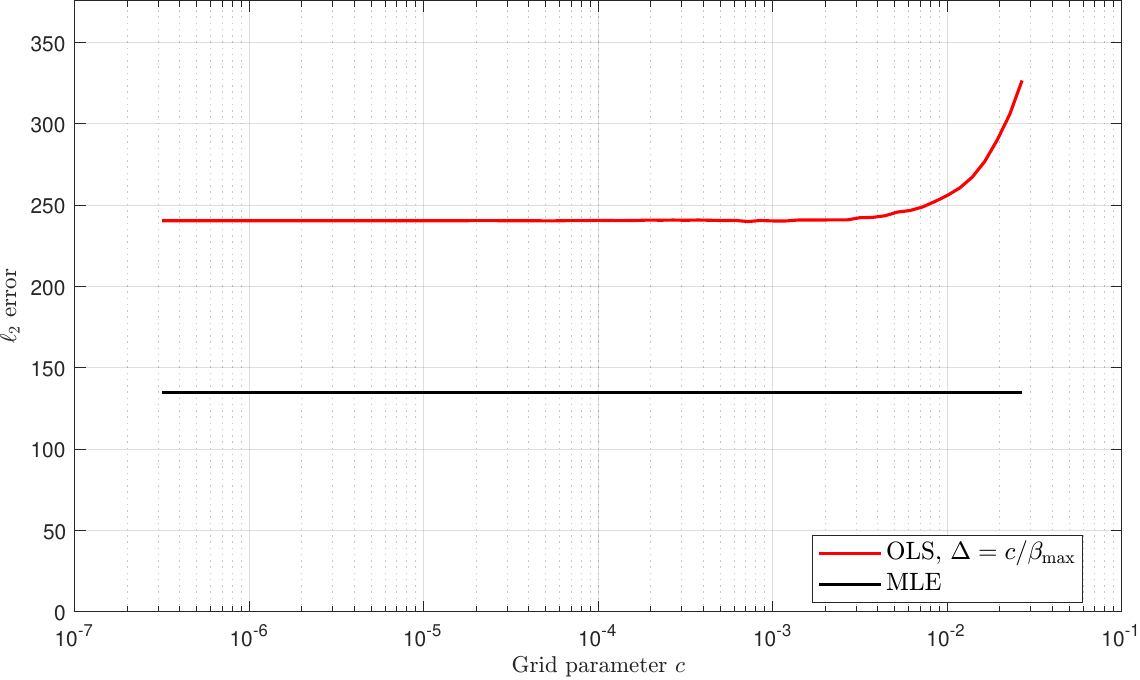}
{5}
{4,43,294,2818,12862}
{100,500,2000,10000,30000}
{0.04,0.09,0.15,0.28,0.43}
\simpanel
{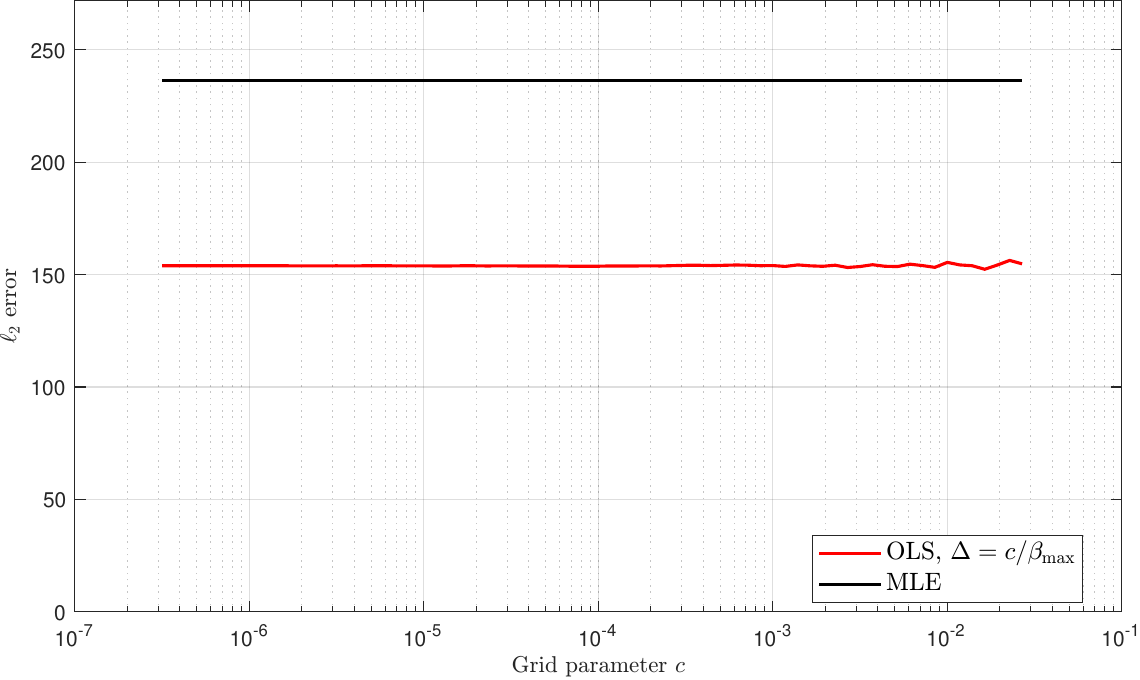}
{5}
{20,59,196,588,980}
{100,300,1000,3000,5000}
{0.20,0.20,0.20,0.20,0.20}
\simpanel
{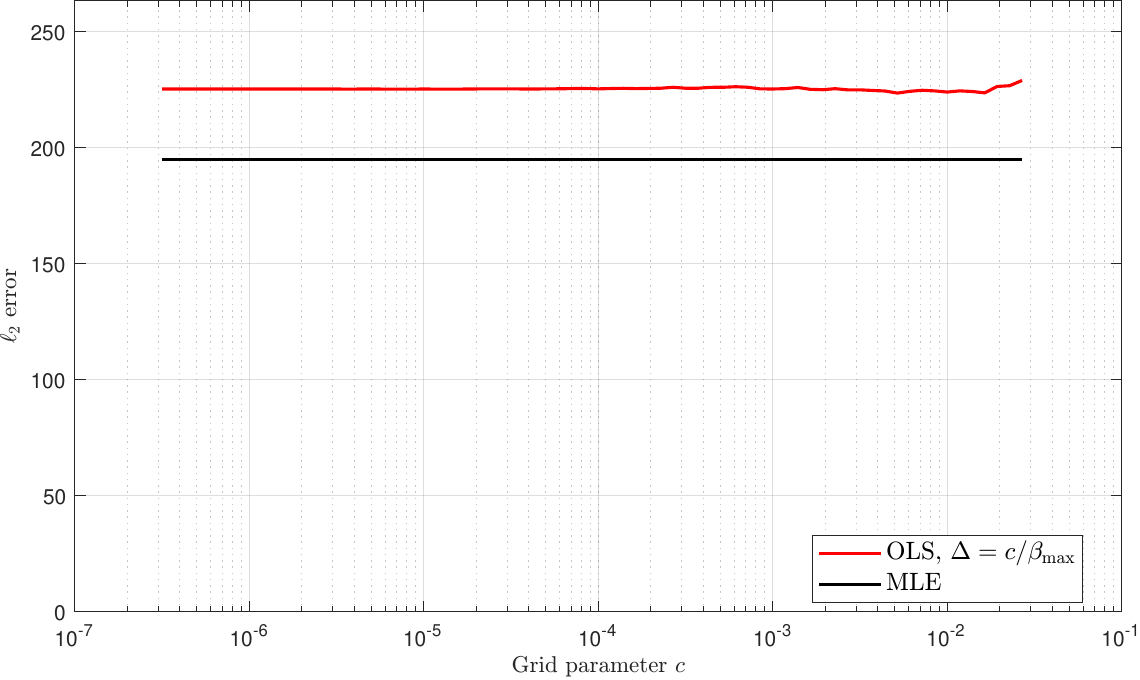}
{5}
{4,26,147,845,2144}
{100,300,1000,3000,5000}
{0.04,0.09,0.15,0.28,0.43}

\vspace{0.2em}

\caption{Average $\ell_2$-errors based on 200 
replications, $\nu^*=5$ and $\text{BR}_{n_h}\approx0.98$.
Each row corresponds to a number of kernel functions $n_h$. The first two columns correspond to large $b^*_{\max}$ design; the last two columns correspond to capped $b^*_{\max}$ design. For each design, homogeneous and heterogeneous branching-ratios are reported separately. The horizontal black line denotes the average estimation error based on maximum likelihood estimation.}
\label{fig:mu5-br098}
\end{figure}


\begin{figure}[p]
\centering

\simheading{Large $b^*_{\max}$}{Homogeneous}
\hfill
\simheading{Large $b^*_{\max}$}{Heterogeneous}
\hfill
\simheading{Capped $b^*_{\max}$}{Homogeneous}
\hfill
\simheading{Capped $b^*_{\max}$}{Heterogeneous}

\vspace{0.25em}

\simpanel
{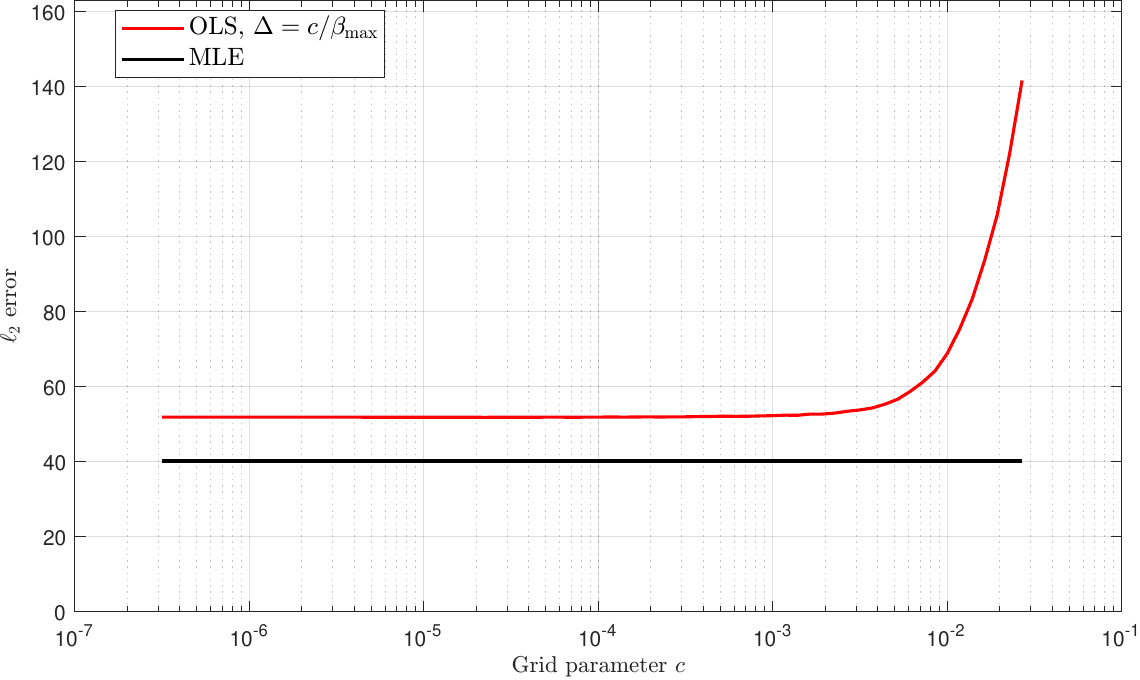}
{2}
{72,6400}
{180,16000}
{0.40,0.40}
\hfill
\simpanel
{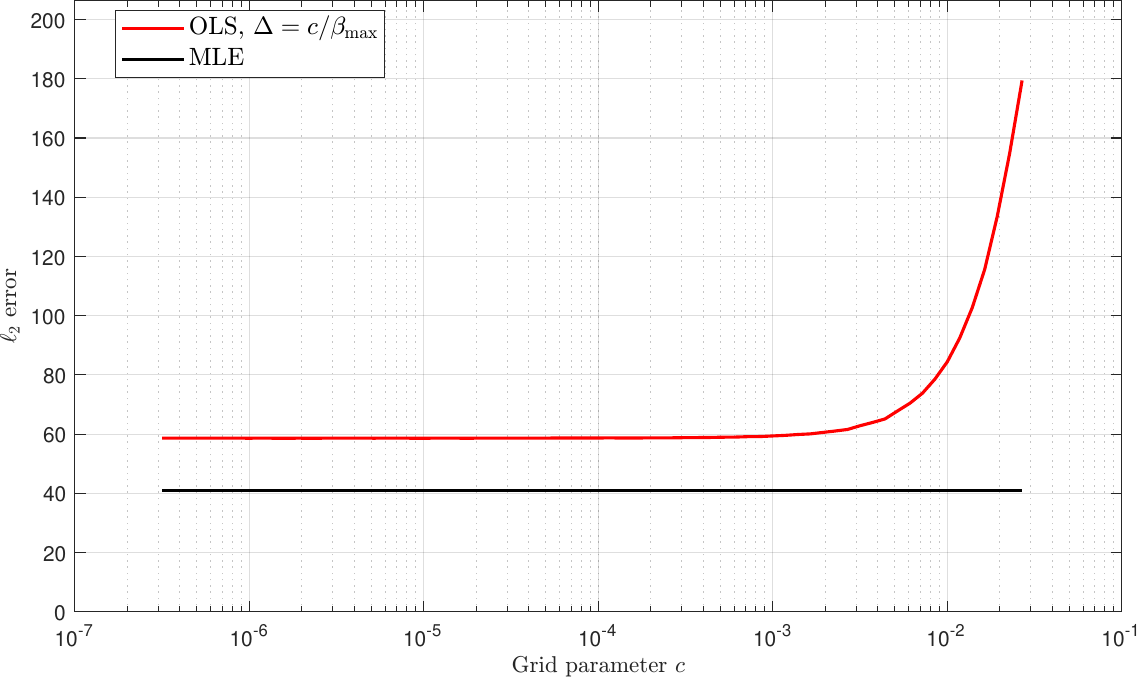}
{2}
{45,8800}
{180,16000}
{0.25,0.55}
\hfill
\simpanel
{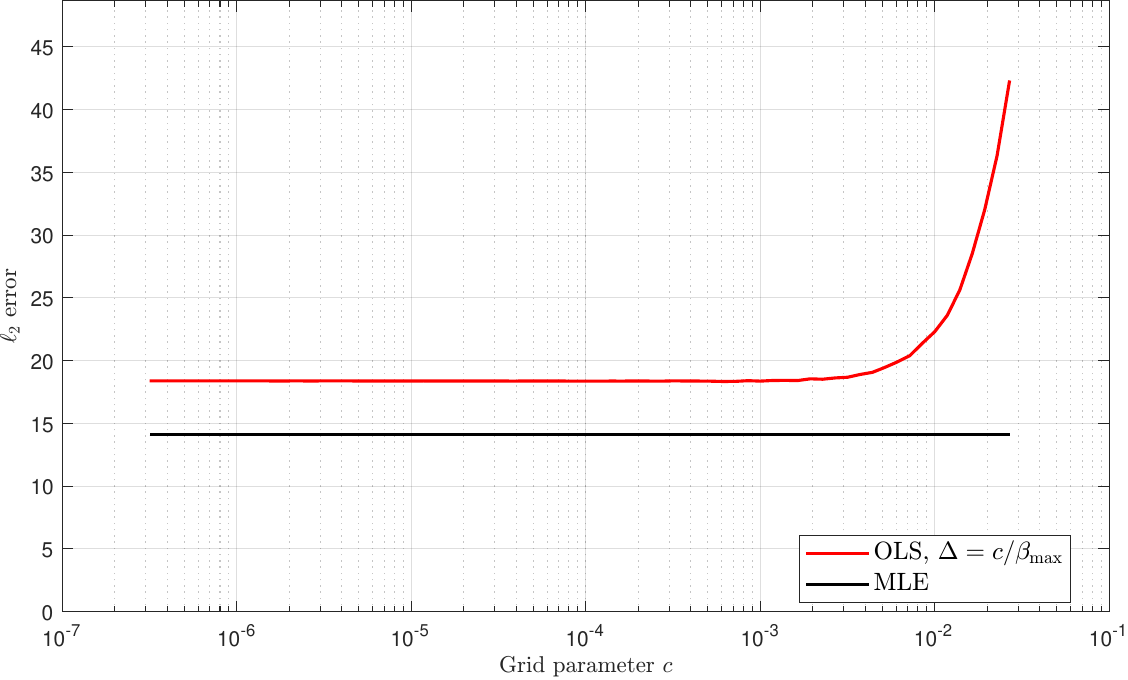}
{2}
{80,2000}
{200,5000}
{0.40,0.40}
\hfill
\simpanel
{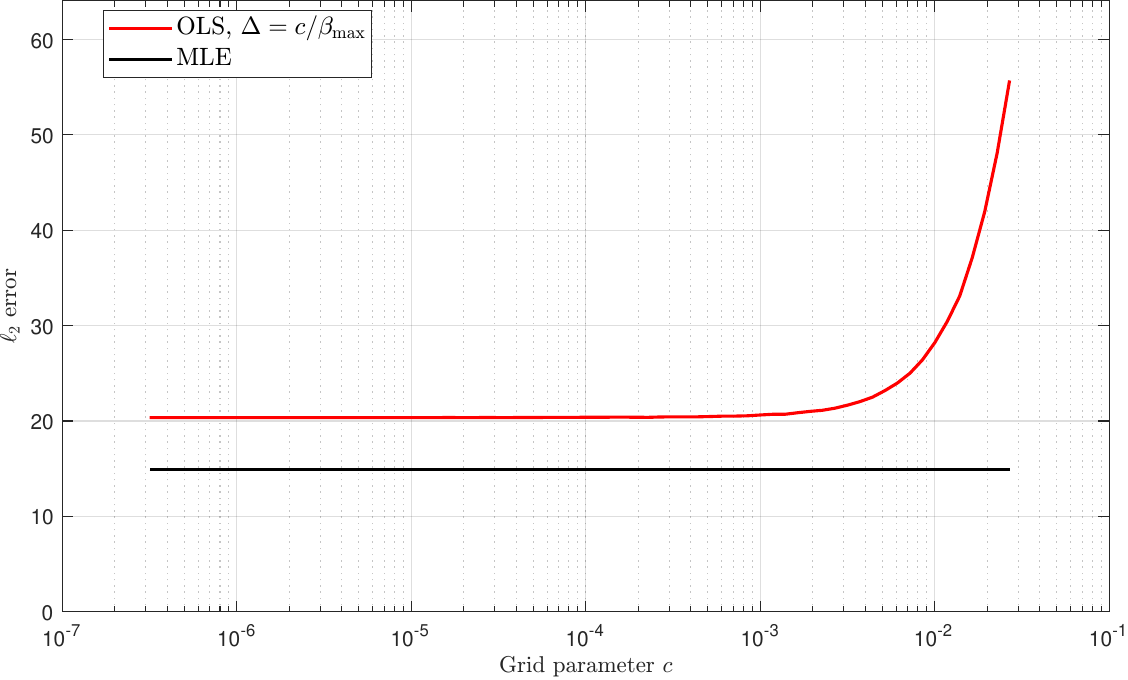}
{2}
{50,2750}
{200,5000}
{0.25,0.55}

\par\vspace{0.35em}

\simpanel
{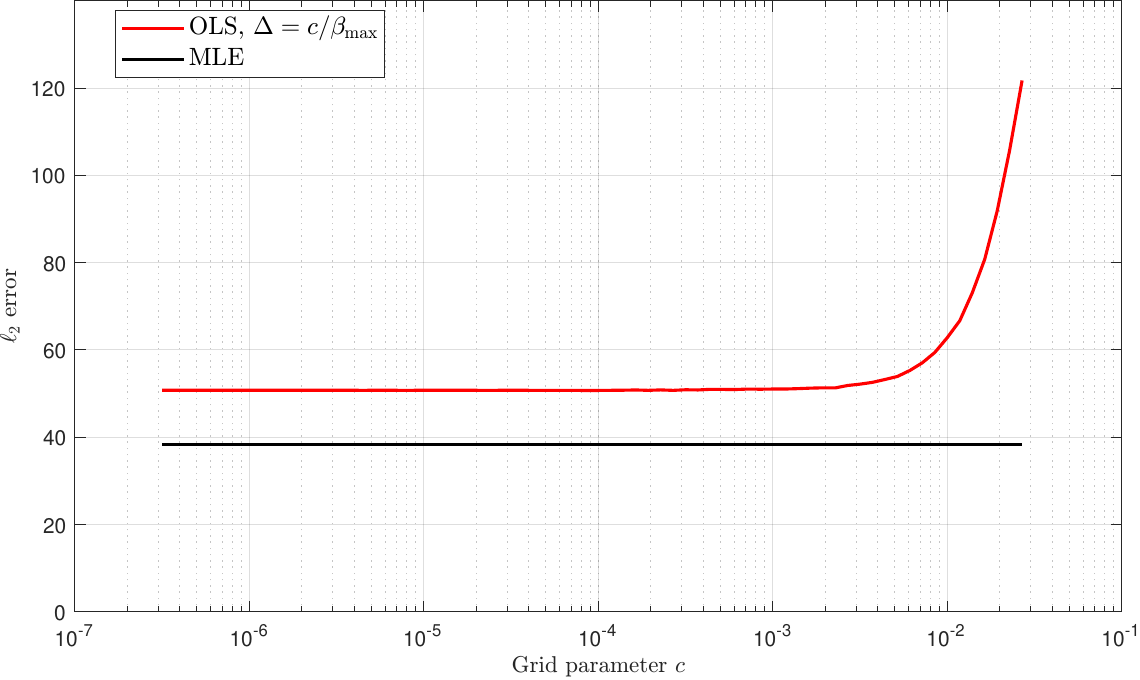}
{3}
{19,200,5332}
{70,750,20000}
{0.27,0.27,0.27}
\hfill
\simpanel
{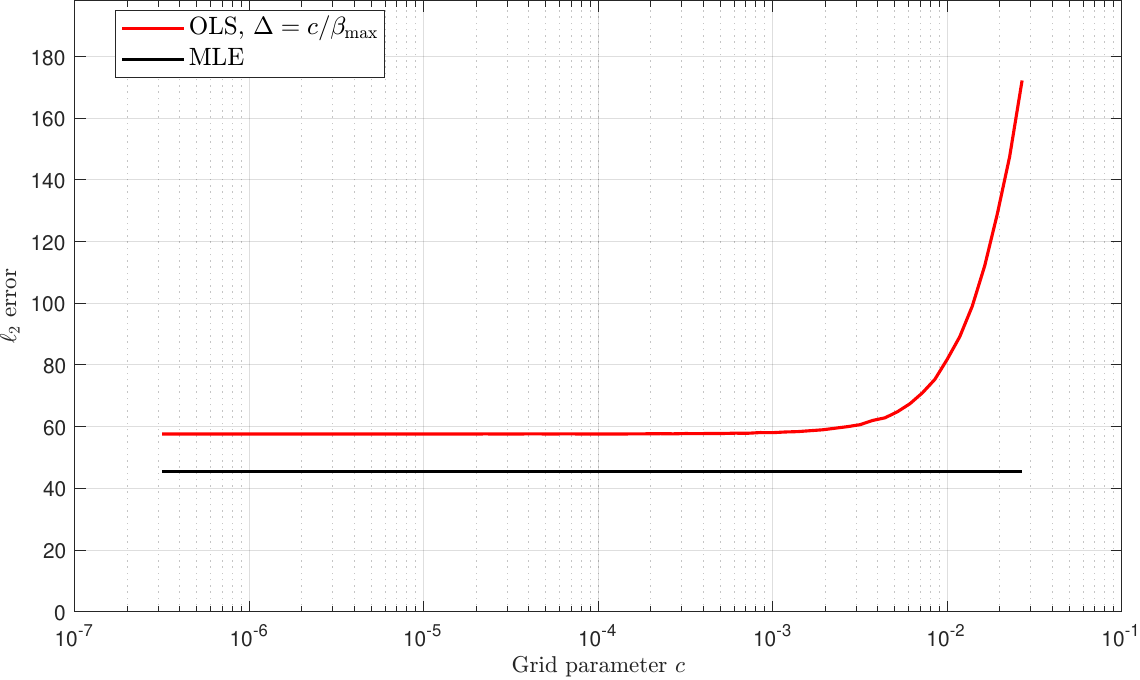}
{3}
{11,188,8000}
{70,750,20000}
{0.16,0.25,0.40}
\hfill
\simpanel
{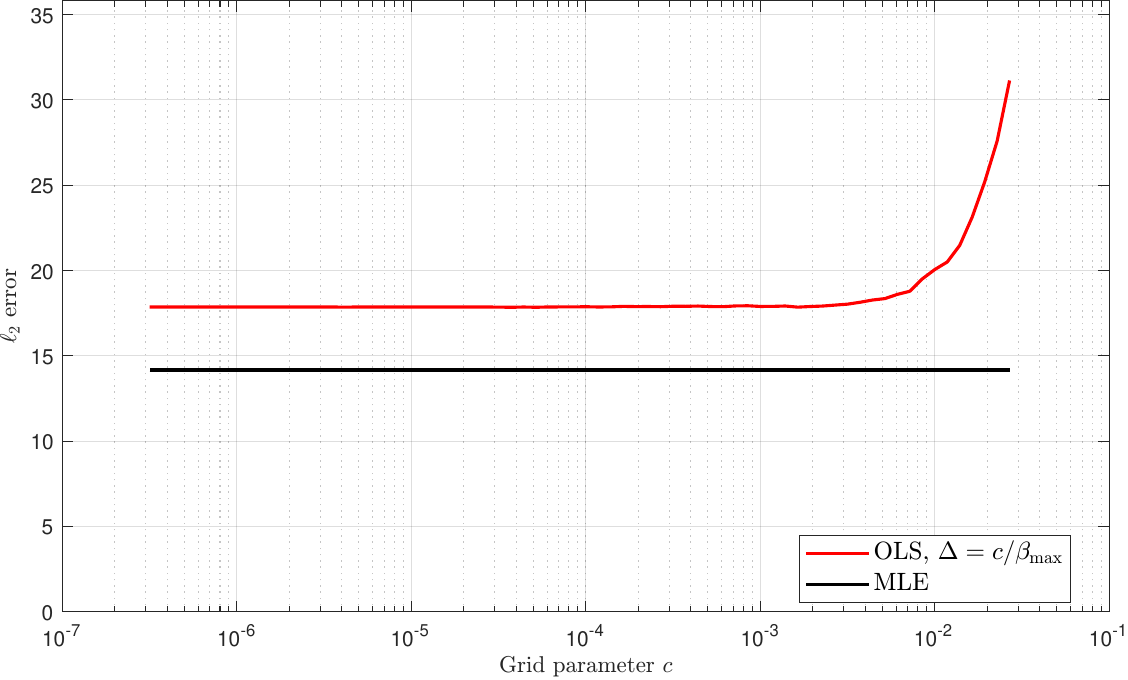}
{3}
{19,200,1333}
{70,750,5000}
{0.27,0.27,0.27}
\hfill
\simpanel
{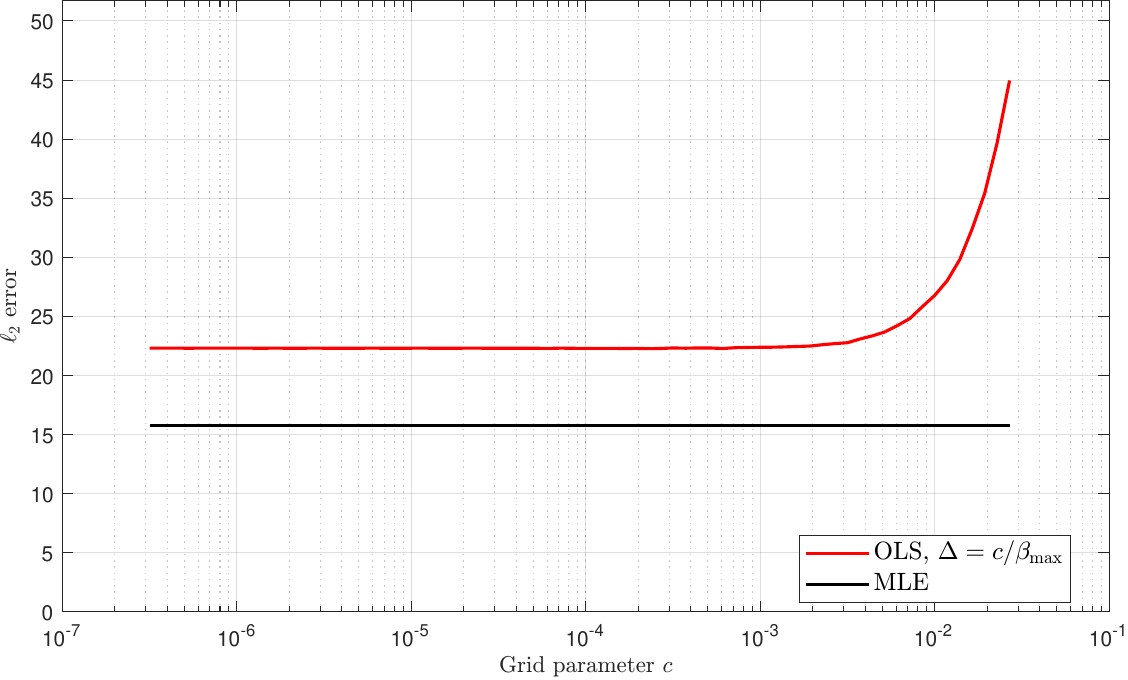}
{3}
{11,188,2000}
{70,750,5000}
{0.16,0.25,0.40}

\par\vspace{0.35em}

\simpanel
{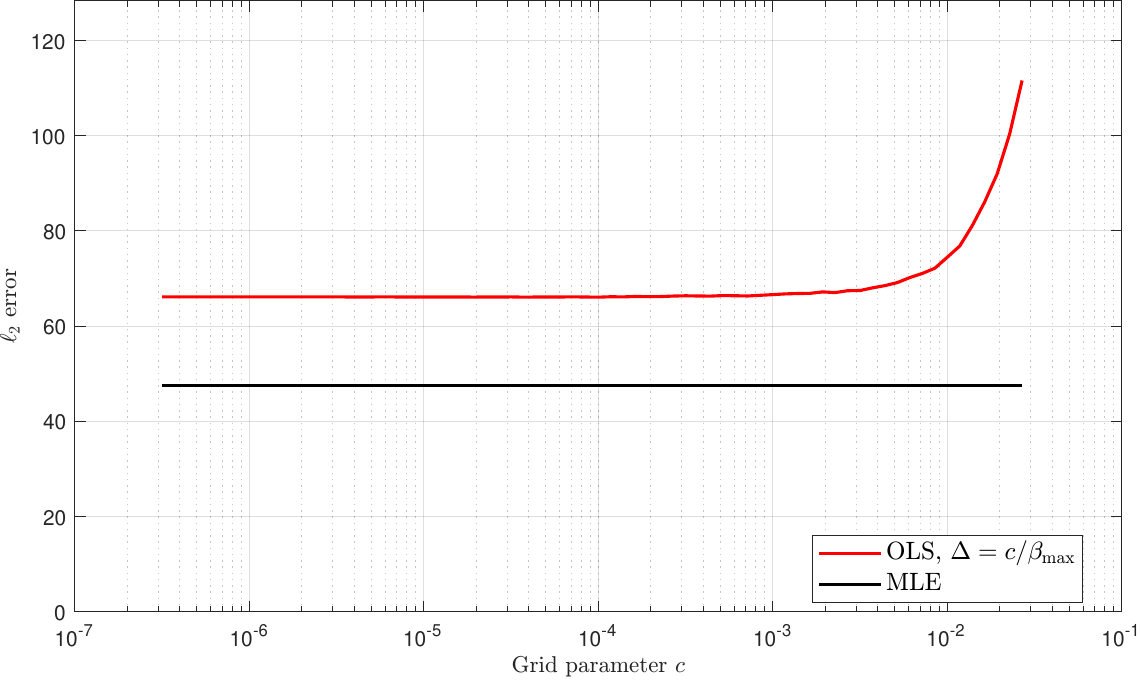}
{4}
{50,100,400,4000}
{250,500,2000,20000}
{0.20,0.20,0.20,0.20}
\hfill
\simpanel
{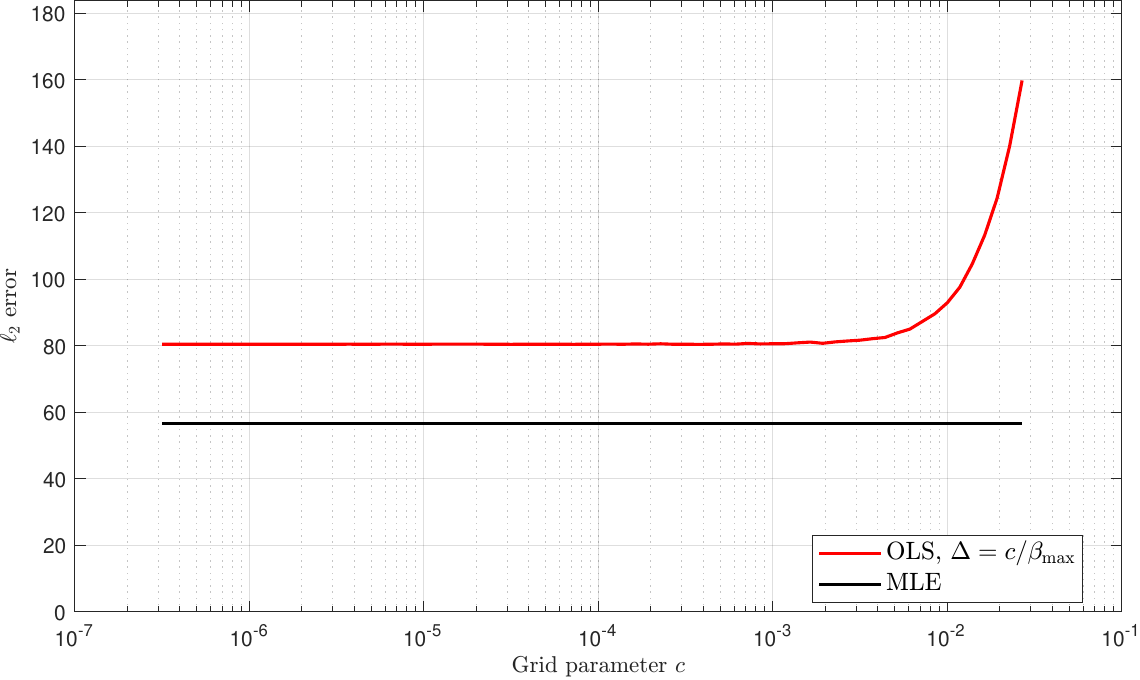}
{4}
{20,70,460,7000}
{250,500,2000,20000}
{0.08,0.14,0.23,0.35}
\hfill
\simpanel
{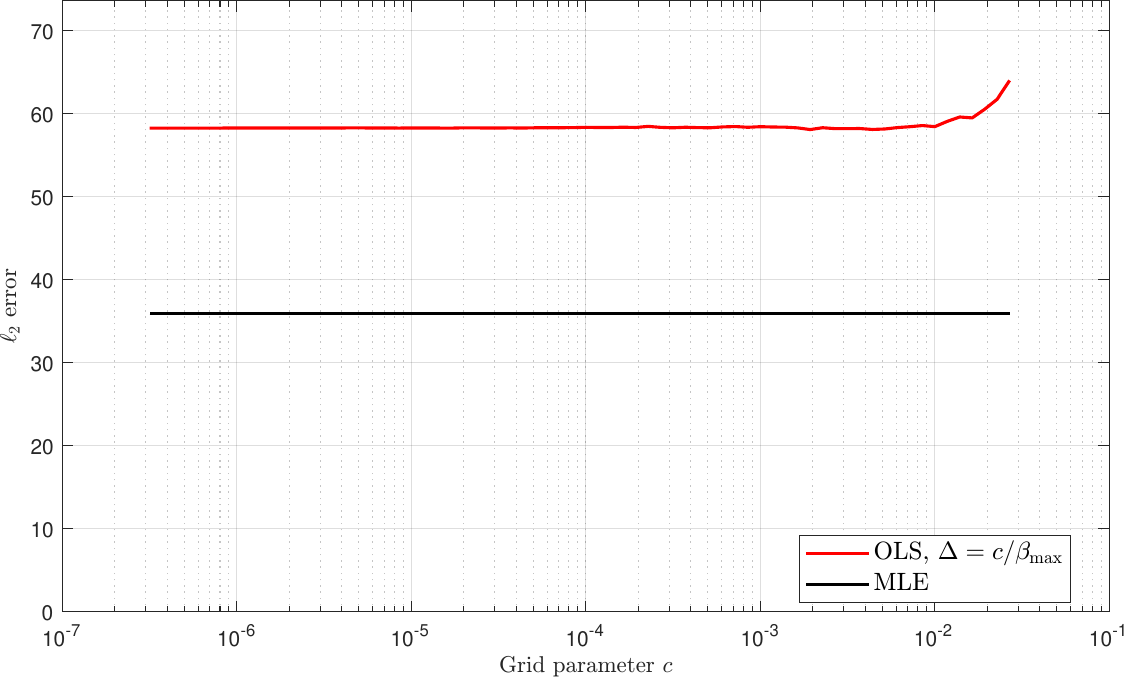}
{4}
{50,100,400,1000}
{250,500,2000,5000}
{0.20,0.20,0.20,0.20}
\hfill
\simpanel
{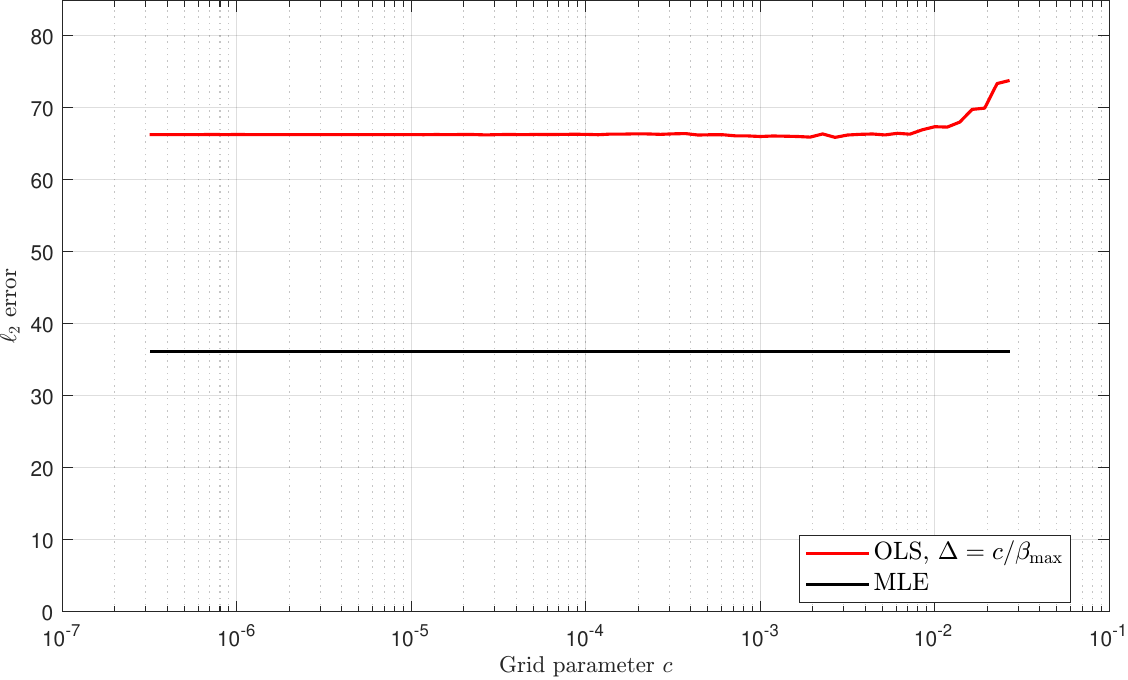}
{4}
{20,70,460,1750}
{250,500,2000,5000}
{0.08,0.14,0.23,0.35}

\par\vspace{0.35em}

\simpanel
{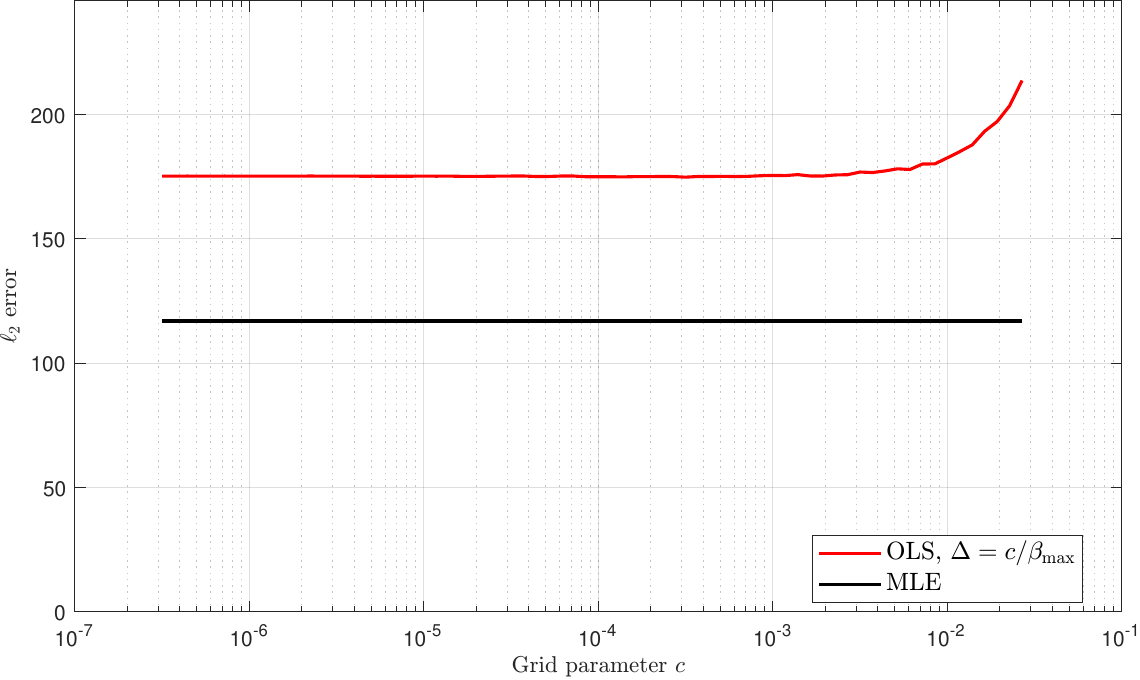}
{5}
{16,80,320,1600,4800}
{100,500,2000,10000,30000}
{0.16,0.16,0.16,0.16,0.16}
\hfill
\simpanel
{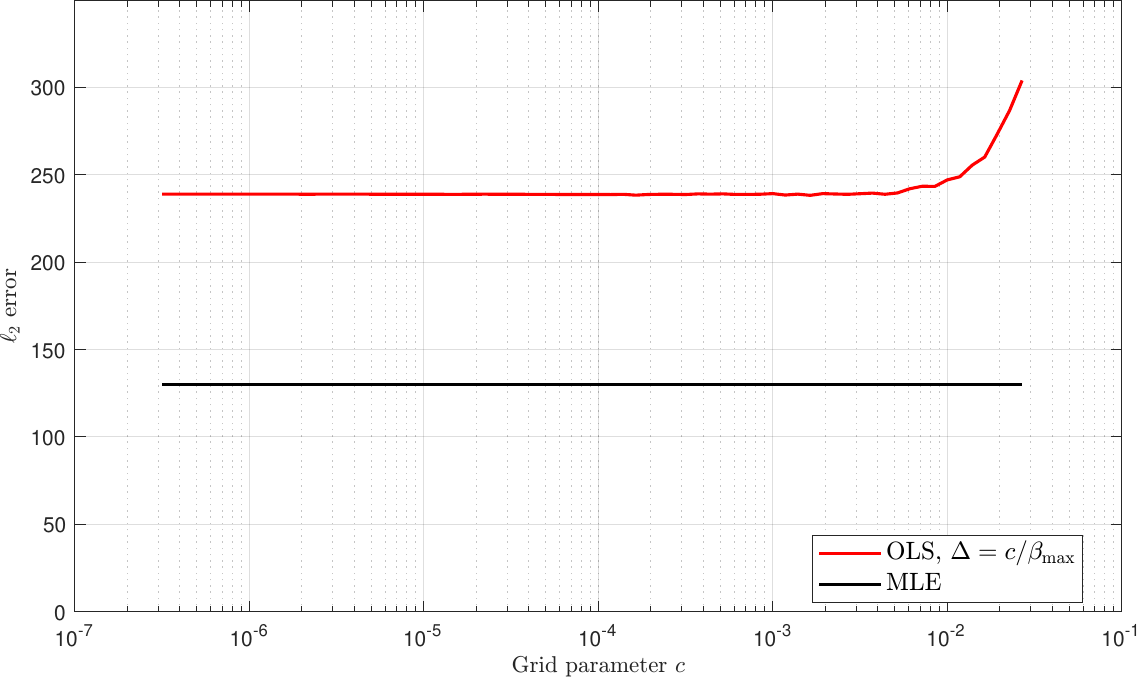}
{5}
{3,35,240,2300,10500}
{100,500,2000,10000,30000}
{0.03,0.07,0.12,0.23,0.35}
\hfill
\simpanel
{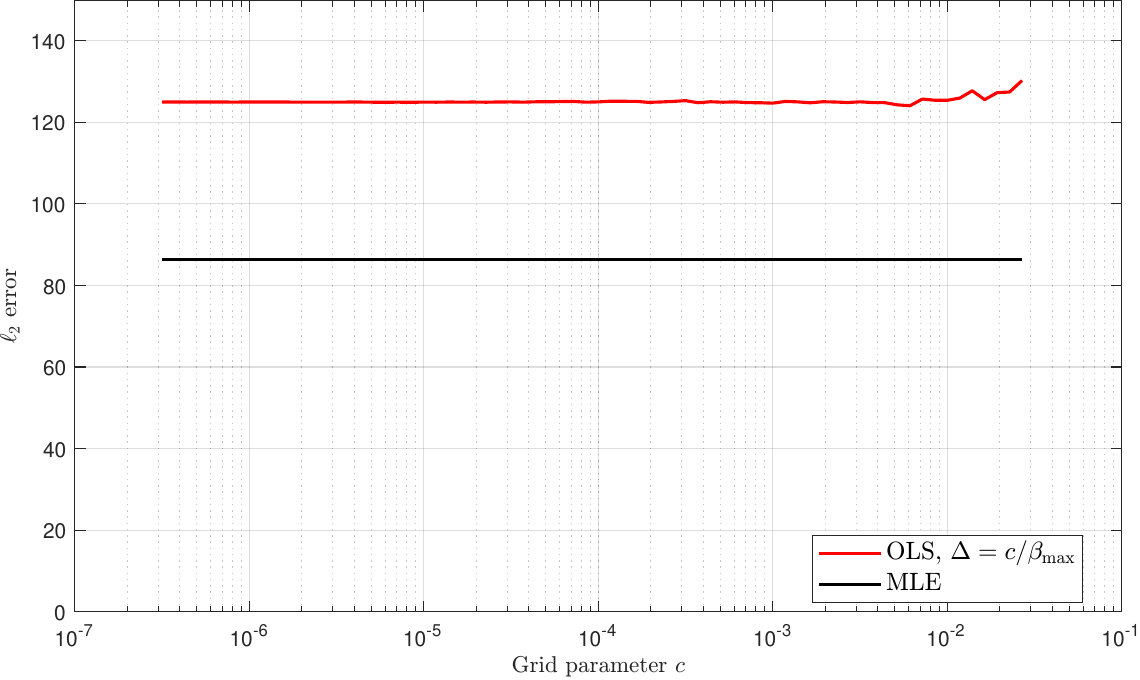}
{5}
{16,48,160,480,800}
{100,300,1000,3000,5000}
{0.16,0.16,0.16,0.16,0.16}
\hfill
\simpanel
{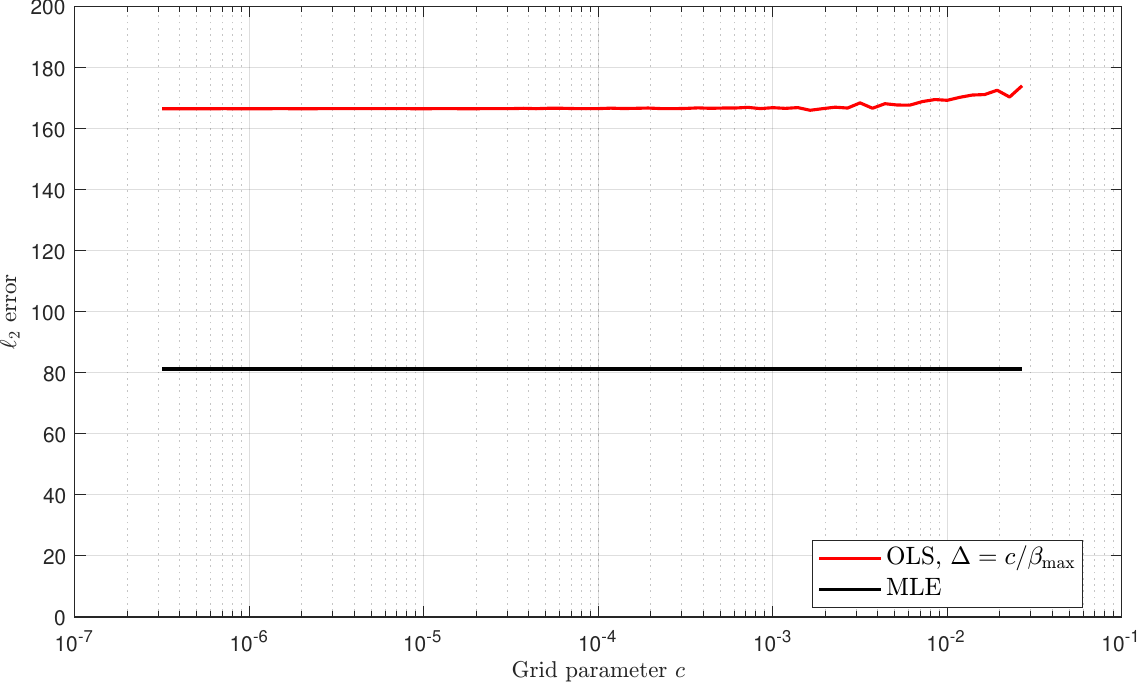}
{5}
{3,21,120,690,1750}
{100,300,1000,3000,5000}
{0.03,0.07,0.12,0.23,0.35}

\vspace{0.2em}

\caption{Average $\ell_2$-errors based on 200 
replications, $\nu^*=1$ and $\text{BR}_{n_h}\approx0.8$.
Each row corresponds to a number of kernel functions $n_h$. The first two columns correspond to large $b^*_{\max}$ design; the last two columns correspond to capped $b^*_{\max}$ design. For each design, homogeneous and heterogeneous branching-ratios are reported separately. The horizontal black line denotes the average estimation error based on maximum likelihood estimation.}
\label{fig:mu1-br08}
\end{figure}

\begin{figure}[p]
\centering

\simheading{Large $b^*_{\max}$}{Homogeneous}
\hfill
\simheading{Large $b^*_{\max}$}{Heterogeneous}
\hfill
\simheading{Capped $b^*_{\max}$}{Homogeneous}
\hfill
\simheading{Capped $b^*_{\max}$}{Heterogeneous}

\vspace{0.25em}

\simpanel
{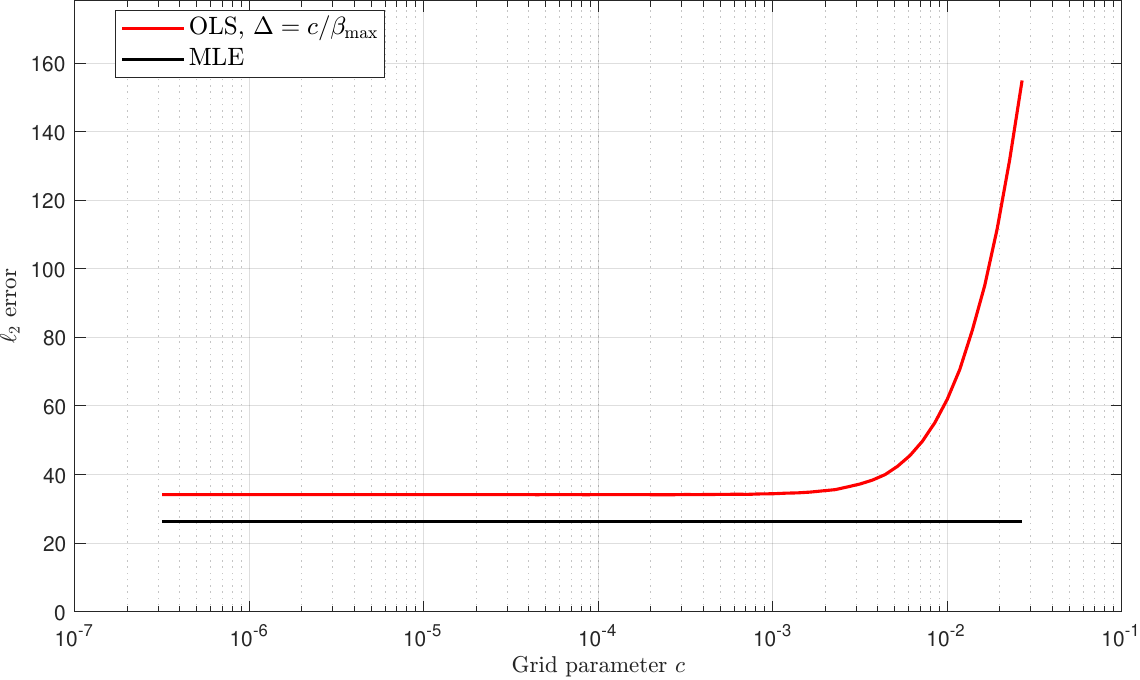}
{2}
{88,7858}
{180,16000}
{0.49,0.49}
\hfill
\simpanel
{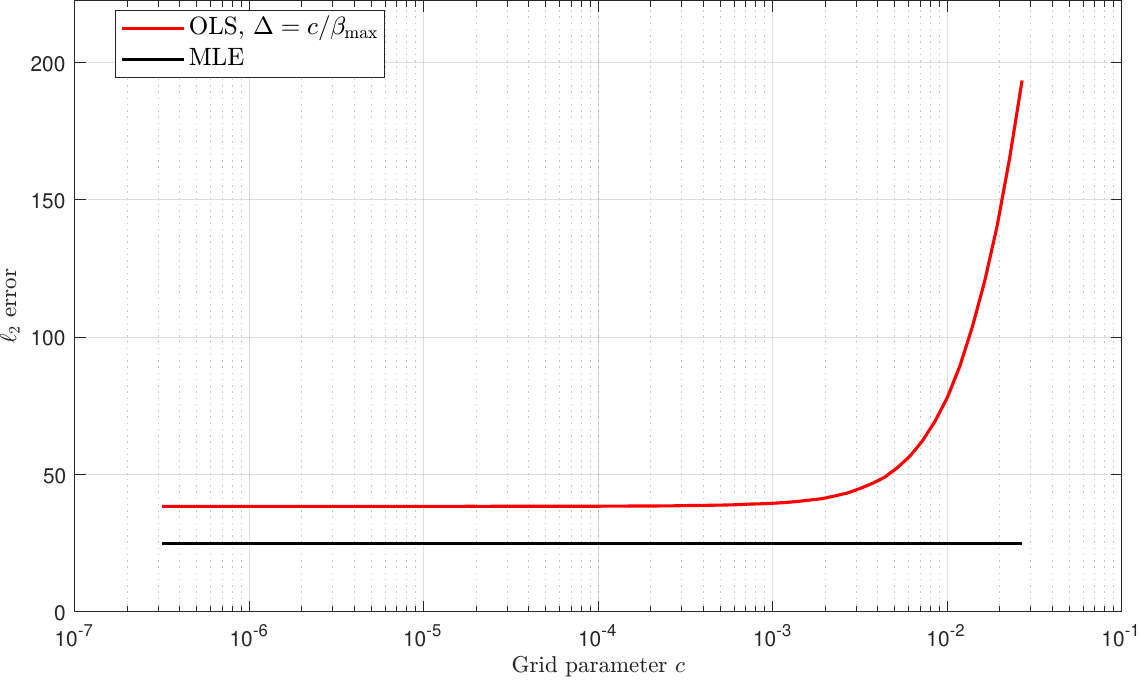}
{2}
{55,10792}
{180,16000}
{0.31,0.68}
\hfill
\simpanel
{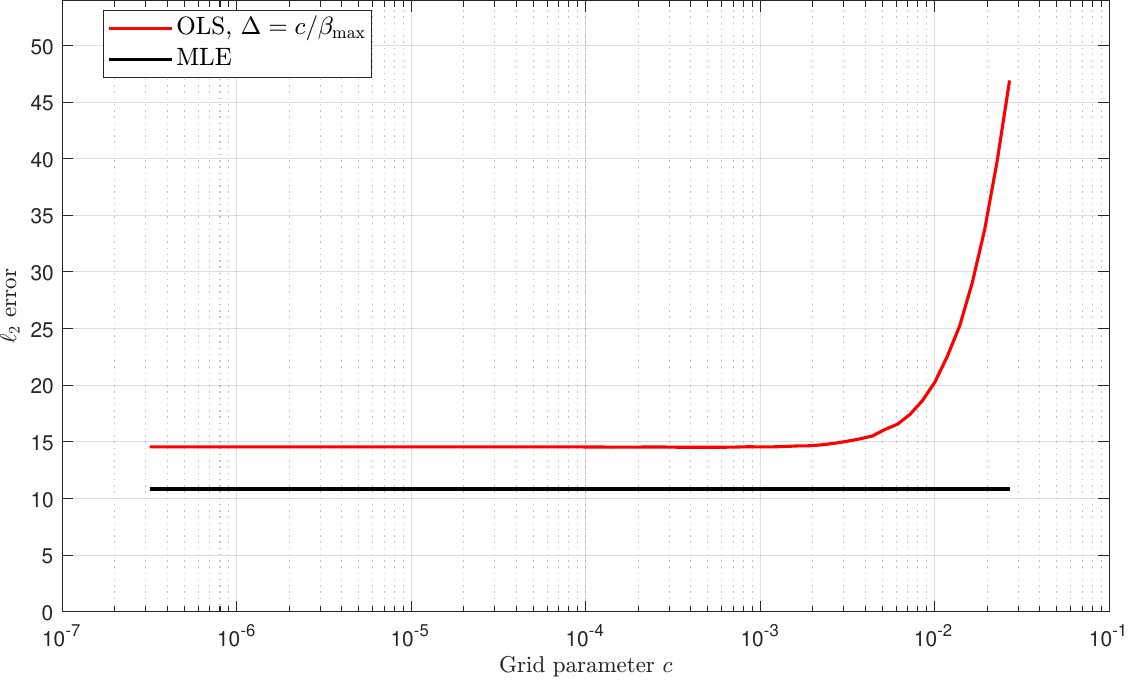}
{2}
{98,2450}
{200,5000}
{0.49,0.49}
\hfill
\simpanel
{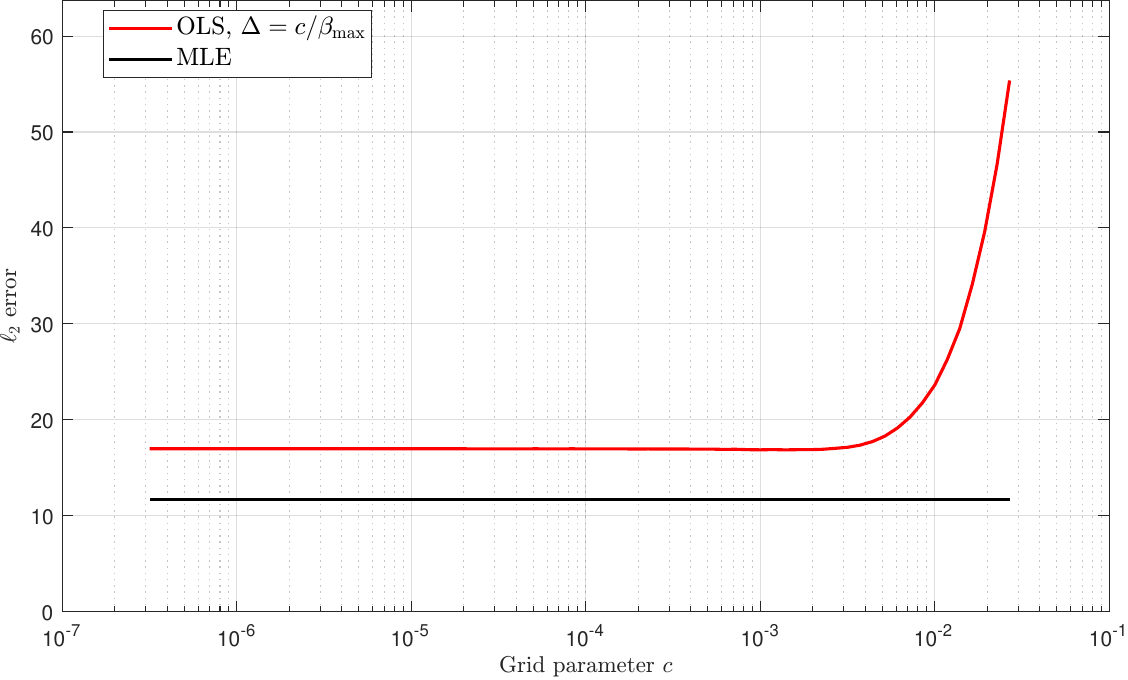}
{2}
{61,3375}
{200,5000}
{0.31,0.68}

\par\vspace{0.35em}

\simpanel
{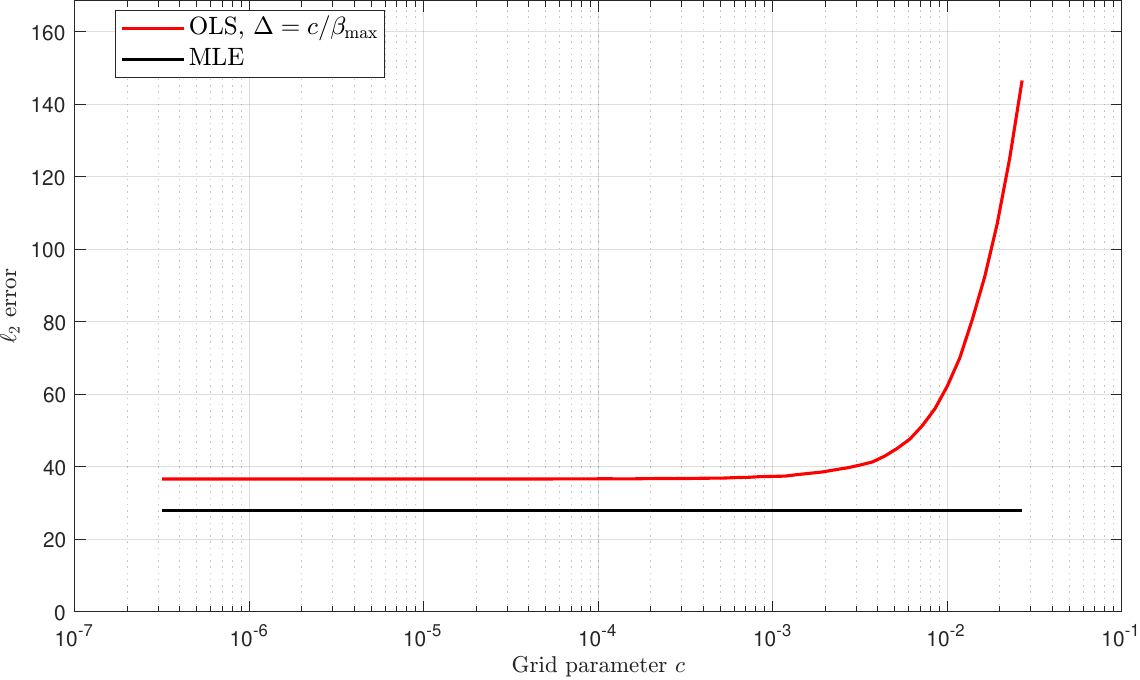}
{3}
{23,245,6532}
{70,750,20000}
{0.33,0.33,0.33}
\hfill
\simpanel
{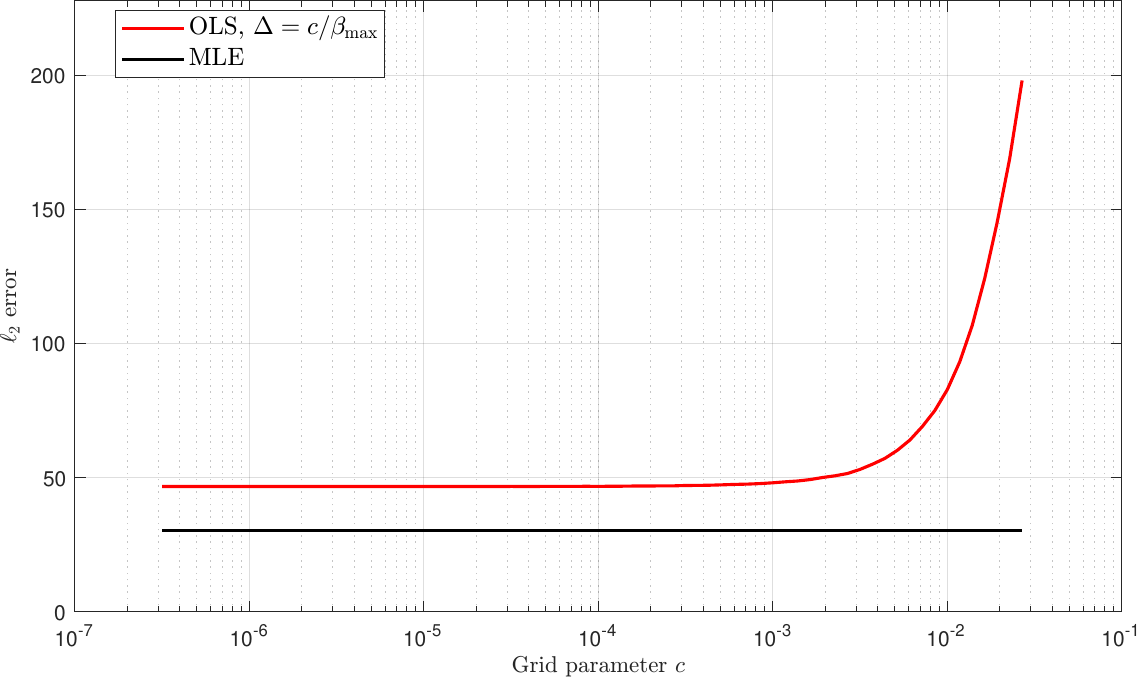}
{3}
{13,230,9800}
{70,750,20000}
{0.19,0.31,0.49}
\hfill
\simpanel
{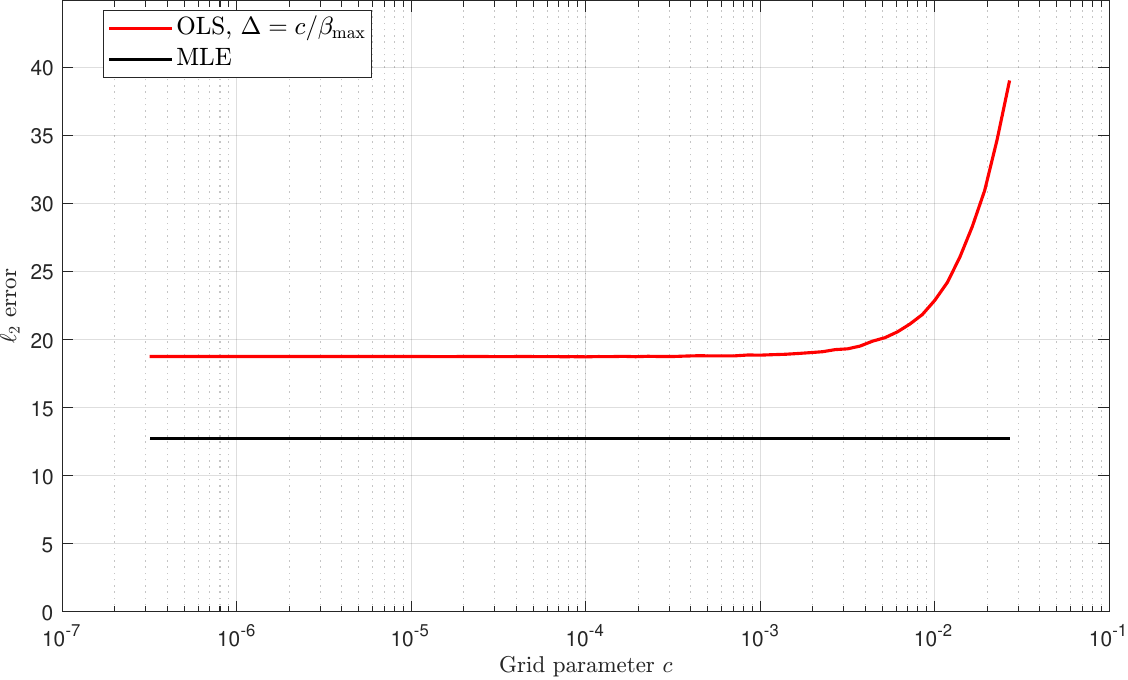}
{3}
{23,245,1633}
{70,750,5000}
{0.33,0.33,0.33}
\hfill
\simpanel
{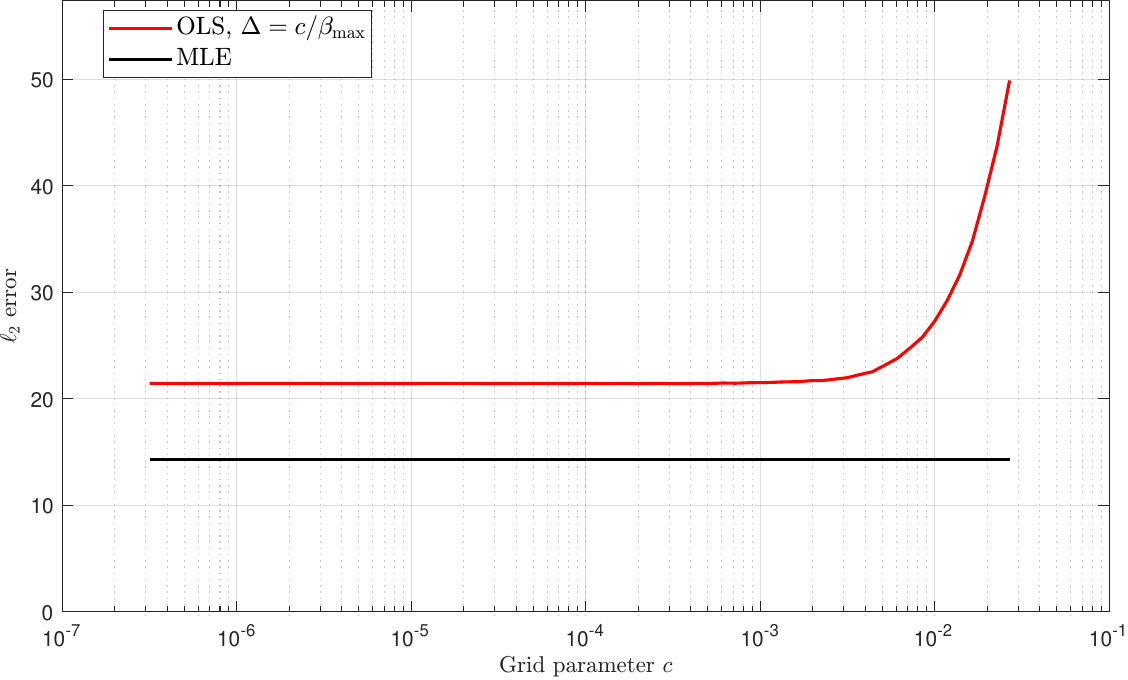}
{3}
{13,230,2450}
{70,750,5000}
{0.19,0.31,0.49}

\par\vspace{0.35em}

\simpanel
{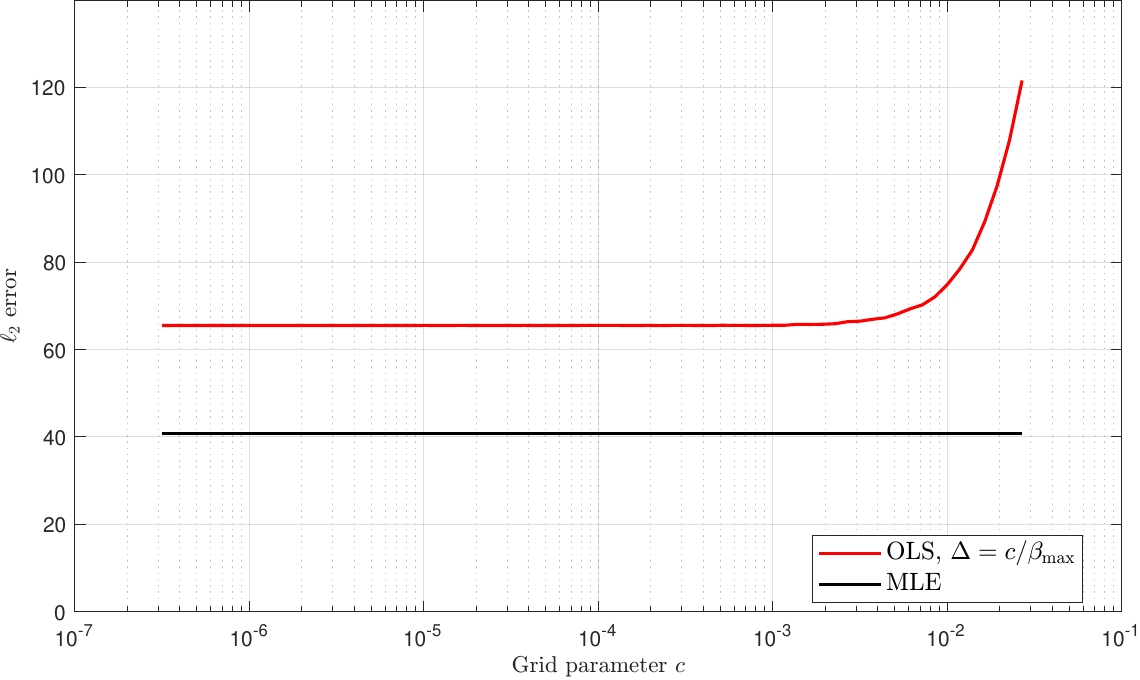}
{4}
{61,123,490,4900}
{250,500,2000,20000}
{0.24,0.25,0.25,0.25}
\hfill
\simpanel
{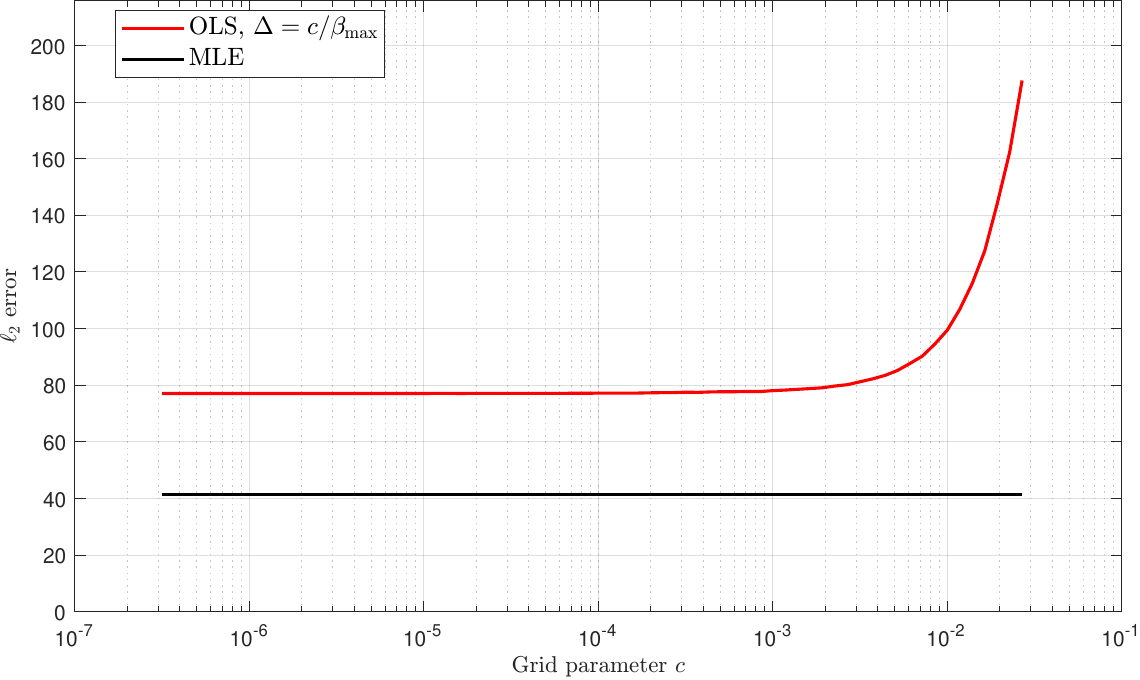}
{4}
{24,86,564,8600}
{250,500,2000,20000}
{0.10,0.17,0.28,0.43}
\hfill
\simpanel
{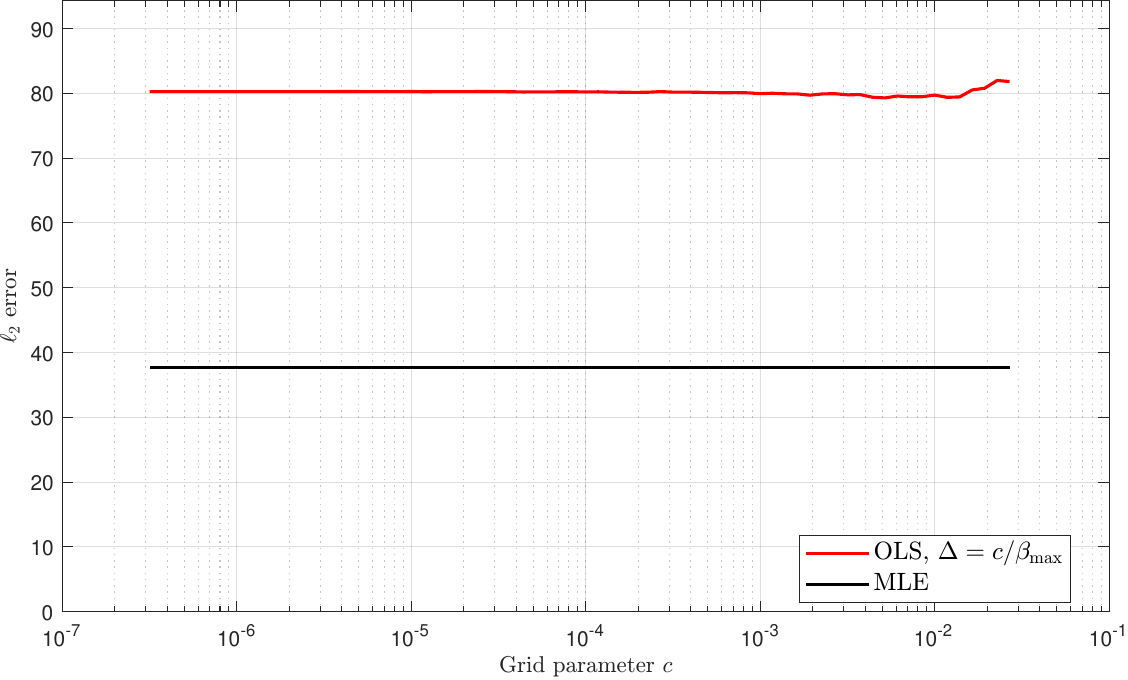}
{4}
{61,123,490,1225}
{250,500,2000,5000}
{0.24,0.25,0.25,0.25}
\hfill
\simpanel
{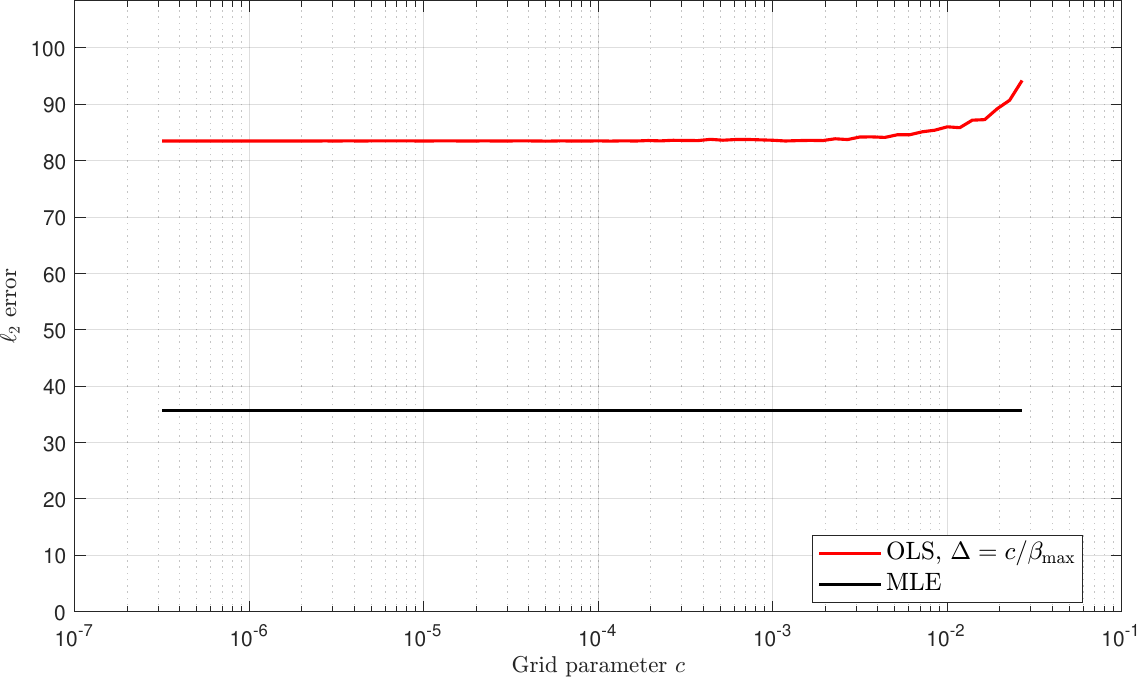}
{4}
{24,86,564,2150}
{250,500,2000,5000}
{0.1,0.17,0.28,0.43}

\par\vspace{0.35em}

\simpanel
{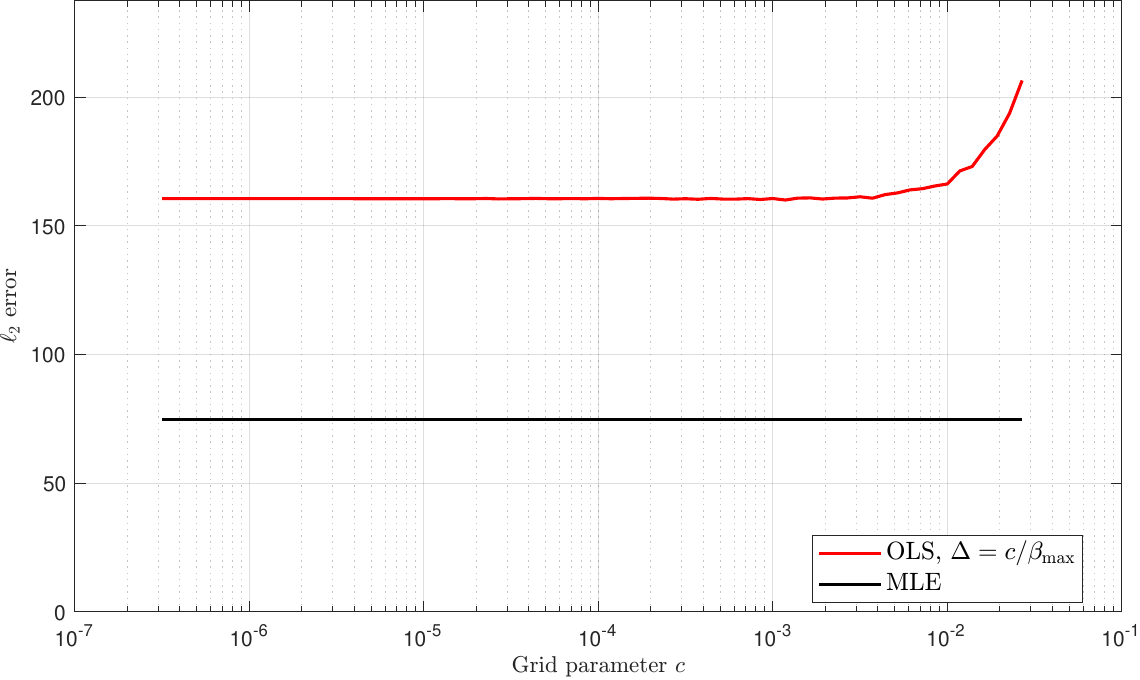}
{5}
{20,98,392,1960,5880}
{100,500,2000,10000,30000}
{0.20,0.20,0.20,0.20,0.20}
\simpanel
{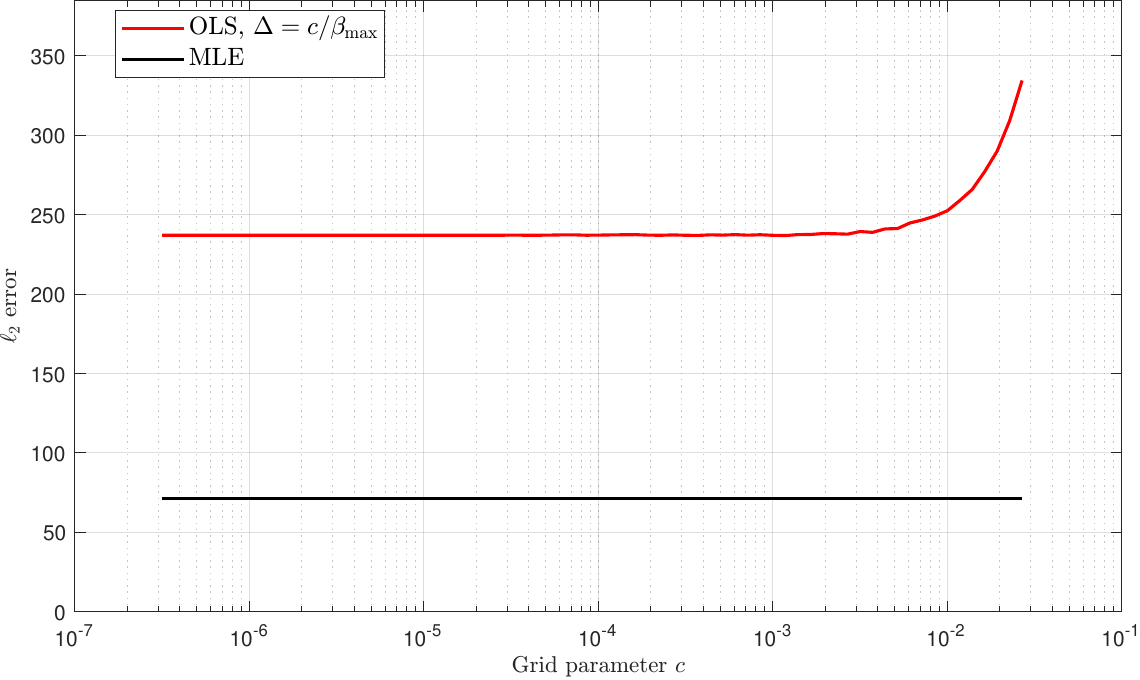}
{5}
{4,43,294,2818,12862}
{100,500,2000,10000,30000}
{0.04,0.09,0.15,0.28,0.43}
\simpanel
{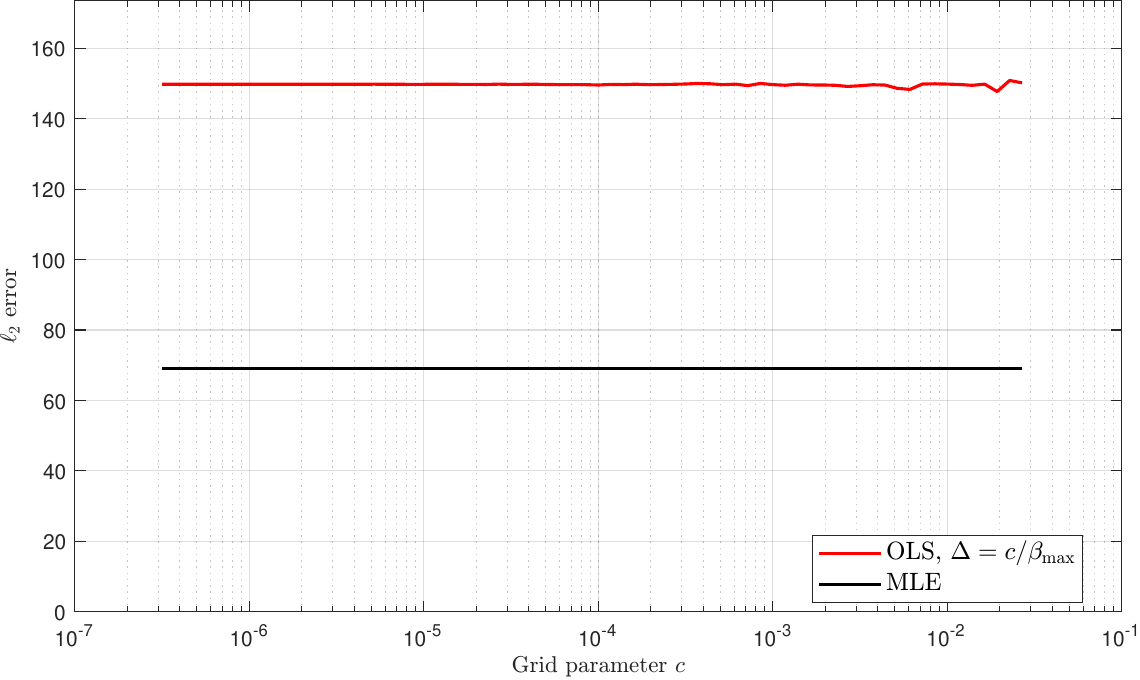}
{5}
{20,59,196,588,980}
{100,300,1000,3000,5000}
{0.20,0.20,0.20,0.20,0.20}
\simpanel
{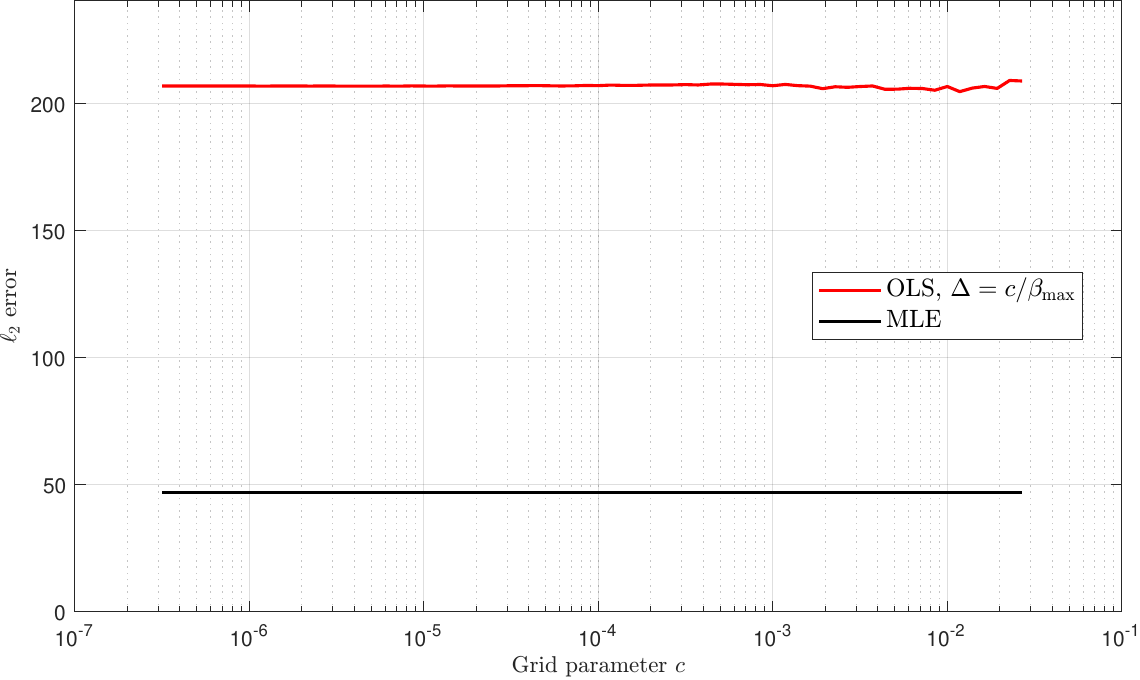}
{5}
{4,26,147,845,2144}
{100,300,1000,3000,5000}
{0.04,0.09,0.15,0.28,0.43}

\vspace{0.2em}

\caption{Average $\ell_2$-errors based on 200 
replications, $\nu^*=1$ and $\text{BR}_{n_h}\approx0.98$.
Each row corresponds to a number of kernel functions $n_h$. The first two columns correspond to large $b^*_{\max}$ design; the last two columns correspond to capped $b^*_{\max}$ design. For each design, homogeneous and heterogeneous branching-ratios are reported separately. The horizontal black line denotes the average estimation error based on maximum likelihood estimation.}
\label{fig:mu1-br098}
\end{figure}
\newpage

\subsection{Statistical test}\label{appendix_subsec:sim_wald}

We now investigate the finite-sample performance of statistical test based on the Wald test statistic for homogeneity of the excitation parameters $(\alpha_1,\ldots,\alpha_{n_h})'$. The true DGP for the intensity is given by Equation (\ref{defhawkesintensity}). We generate 2000 independent batches of $\{N_t\}_{0\leq t \leq T}$ in which $T=21600$. This corresponds to the number of trading seconds in one trading day excluding the first and last 30 minutes. Throughout the simulation study, we fix $n_h=3$, $\nu^*=5$, and $b^*=(300,4000,15000)'$. The latter vector of decay numbers is assumed known. The procedure based on ordinary least squares is conducted with $\Delta = c/b^*_{\max}$, $b^*_{\max}=\max_{k}b^*_{k}$, $c=0.001$, $c=0.005$ and $c=0.027$.
For the true excitation parameter $\alpha^*$, we consider the following two data generating processes: \textbf{DGP 1} with $\alpha^*=(100,1500,1500)'$ and \textbf{DGP 2} with $\alpha^*=(100,1500,2000)'$.
For each simulated realization, we test the three null hypotheses
\begin{equation*}
H_{0,1}:\alpha_1^*=\alpha_2^*,\qquad
H_{0,2}:\alpha_1^*=\alpha_3^*,\qquad
H_{0,3}:\alpha_2^*=\alpha_3^*.
\end{equation*}
Therefore, only $H_{0,3}$ is true under \textbf{DGP 1} whereas all three null hypotheses are false under \textbf{DGP 2}.
The Wald statistic $S_{T\Delta}$ is computed according to Equation (\ref{defW}). Writing each null hypothesis in the form $R\theta^*=r$ in which
$\theta^*=(\nu^*,\alpha_1^*,\alpha_2^*,\alpha_3^*)'$,
we set $r=0$ and use the restriction matrices $R_1=(0,1,-1,0)$, $R_2=(0,1,0,-1)$, and $R_3=(0,0,1,-1)$. These three matrices correspond respectively to $H_{0,1}$, $H_{0,2}$ and $H_{0,3}$. For each value of the parameter $c$, Table \ref{tab:wald_sim_rejection} reports the empirical rejection frequency at the nominal 5\% significance level.

\begin{table}[H]
\centering
\caption{Empirical rejection frequencies of the pairwise Wald tests at the
5\% significance level based on 2,000 replications. Entries are percentages.
Under DGP 1, the rejection frequency for
$H_{0,3}:\alpha_2^*=\alpha_3^*$ measures empirical size, whereas the other entries
measure empirical power. Under DGP 2, all entries measure empirical power.}
\label{tab:wald_sim_rejection}
\begin{tabular}{l|ccc|ccc}
\hline\hline
&
\multicolumn{3}{c|}{\textbf{DGP 1}}
&
\multicolumn{3}{c}{\textbf{DGP 2}}
\\
&
\multicolumn{3}{c|}{$\alpha^*=(100,1500,1500)'$}
&
\multicolumn{3}{c}{$\alpha^*=(100,1500,2000)'$}
\\
\cline{2-7}
Null hypothesis
& $c=0.001$ & $c=0.005$ & $c=0.027$
& $c=0.001$ & $c=0.005$ & $c=0.027$
\\
\hline

$H_{0,1}:\alpha_1^*=\alpha_2^*$
& 100.00 & 100.00 & 100.00
& 100.00 & 100.00 & 100.00
\\

$H_{0,2}:\alpha_1^*=\alpha_3^*$
& 100.00 & 100.00 & 100.00
& 100.00 & 100.00 & 100.00
\\

$H_{0,3}:\alpha_2^*=\alpha_3^*$
& 4.85 & 5.40 & 13.35
& 100.00 & 100.00 & 100.00
\\

\hline\hline
\end{tabular}
\end{table}

Table \ref{tab:wald_sim_rejection} suggests that the proposed Wald test results in good finite-sample power in the considered settings. Under \textbf{DGP 2} in which all three null hypotheses are false, the rejection frequency is $100\%$ for all hypothesis and every time increment $\Delta$. Likewise, the false null hypotheses $H_{0,1}$ and $H_{0,2}$ are always rejected under \textbf{DGP 1} which indicates adequate empirical power.

The only true null hypothesis is $H_{0,3}$ under \textbf{DGP 1}. Its empirical rejection frequency is $4.85\%$ and $5.40\%$ for $c=0.001$ and $c=0.005$, respectively, which are close to the $5\%$ significance level. This suggests that the asymptotic approximation underlying the Wald statistic provides an accurate description of its finite-sample distribution for sufficiently small values of the discretization parameter.

The rejection frequency is altered when $c=0.027$ in which the empirical rejection frequency becomes $13.35\%$. This suggests a size distortion when the time increment $\Delta$ is large. Larger values of the parameter $c$ correspond to coarser discretizations of the Hawkes process, reducing the accuracy of the discrete-time approximation. Overall, these results suggest that the proposed Wald test exhibits good finite-sample performances provided that the parameter $c$ is chosen sufficiently small.

\subsection{Asymptotic distribution}\label{appendix_subsec:asymptotic_dist}

In this subsection, we conduct a finite-sample simulation experiment to illustrate the result derived in Theorem \ref{thcltest} of the main text. More precisely, we generate $2000$ independent batches of the Hawkes process defined by Equation (\ref{defhawkesintensityfamily}) of the main text and set $T=10000$ as in Subsection \ref{appendix_subsec:sim_l2error} of the Appendices. 
Throughout the simulation study, we set $c=0.005$, $\nu^*=1$ and $b^*=(500,2000,10000)'$. The number of kernel functions is $n_h=3$. Let $\text{BR}_{n_h}=\| h \|_1=\sum^{n_h}_{k=1}\alpha^*_k/b^*_k$. We consider four settings for the excitation parameter $\alpha^*$:
\begin{enumerate}
\item[(i)] $\alpha^*=(130,540,2700)'$, yielding $\text{BR}_3=0.80$ with homogeneous ratios $\alpha^*_k/b^*_k$, $1 \leq k \leq n_h$
\item[(ii)] $\alpha^*=(70,520,4000)'$, yielding $\text{BR}_3=0.80$ with heterogeneous ratios $\alpha^*_k/b^*_k$, $1 \leq k \leq n_h$
\item[(iii)] $\alpha^*=(160,660,3300)'$, yielding $\text{BR}_3=0.98$ with homogeneous ratios $\alpha^*_k/b^*_k$, $1 \leq k \leq n_h$
\item[(iv)] $\alpha^*=(90,620,4900)'$, yielding $\text{BR}_3=0.98$ with heterogeneous ratios $\alpha^*_k/b^*_k$, $1 \leq k \leq n_h$.
\end{enumerate}
The vector $b^*$ is fixed and known. For each setting and each batch, we compute the vector $\widehat{z}_{T\Delta}=\widehat{D}^{-1/2}_{T\Delta}\sqrt{T\Delta}(\widehat{\theta}_{T\Delta}-\theta^*) \in \mathbb{R}^4$. We repeat the procedure over the 2000 batches and compare the empirical distribution of each component with the standard normal distribution through Q-Q plots.

Figures~\ref{qq_plot_0.8} and \ref{qq_plot_98} suggest that the empirical distributions of the standardized statistics are very close to the standard normal distribution across all four components. In both branching-ratio settings, the points lie almost perfectly on the reference line over the central part of the distribution, indicating that the Gaussian approximation derived in Theorem~\ref{thcltest} of the main text is already accurate for $T=10000$.

Small departures from the reference line are visible only in the extreme upper and lower tails, particularly for the heterogeneous specifications and for the third and fourth components. Such deviations are relatively minor even when the branching ratio is $\text{BR}_3=0.98$. Overall, these simulation results provide strong empirical support for the asymptotic normal approximation underlying the proposed inference procedure.

\begin{figure}[htbp]
\centering
\includegraphics[width=0.8\linewidth]{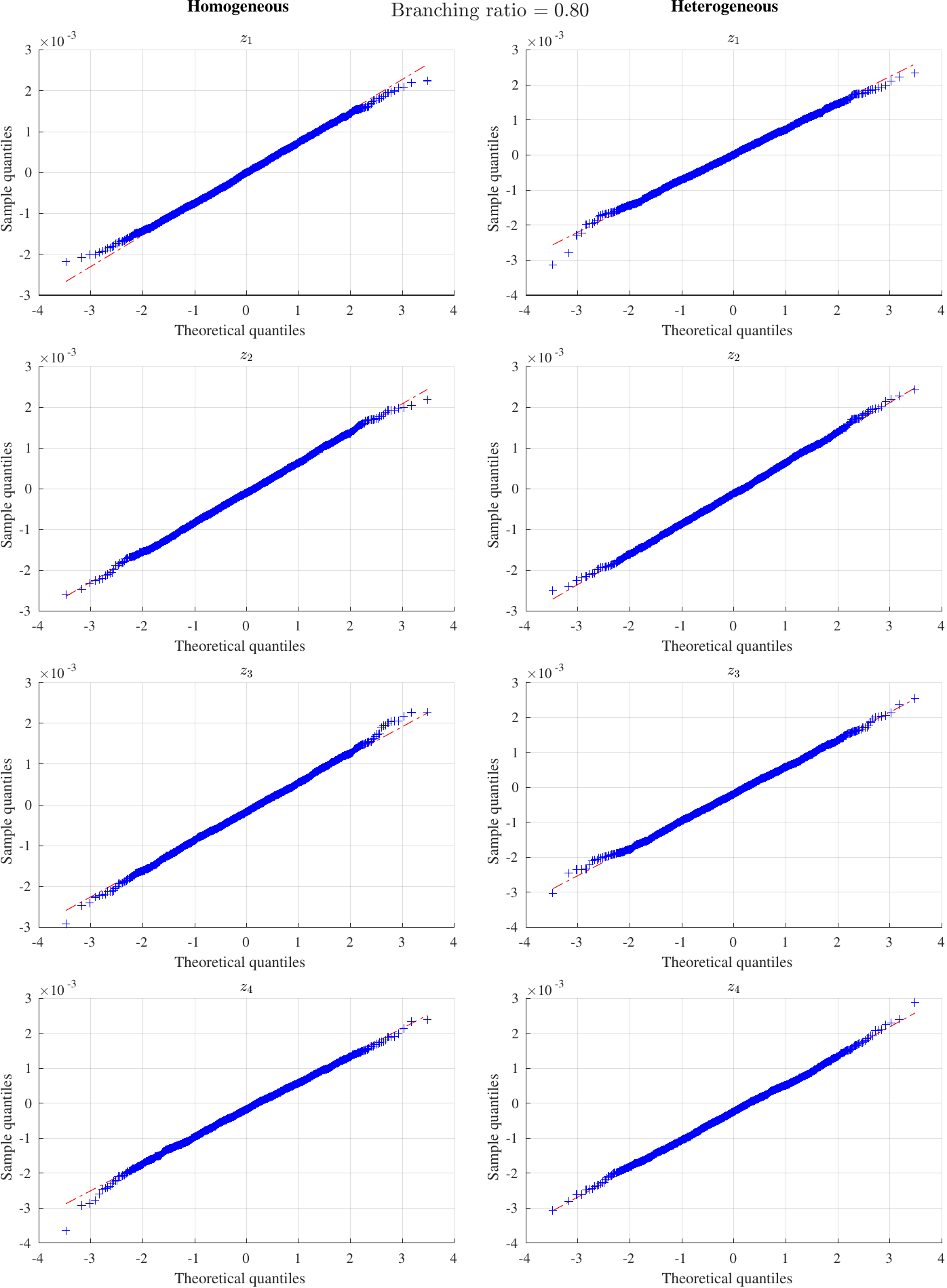}
\caption{Q-Q plots of the standardized statistics $(z_1,z_2,z_3,z_4)$ based on $2000$ replications under Case 1  with branching ratio $\text{BR}_3 = 0.8$, $c=0.005$. The left panels correspond to the homogeneous case with $\nu^*=1$, $\alpha^*=(130,540,2700)'$ and $b^*=(500,2000,10000)'$. The right panels correspond to the heterogeneous case with $\nu^*=1$, $\alpha^*=(70,520,4000)'$ and $b^*=(500,2000,10000)'$.}
\label{qq_plot_0.8}
\end{figure}

\begin{figure}[htbp]
\centering
\includegraphics[width=0.8\linewidth]{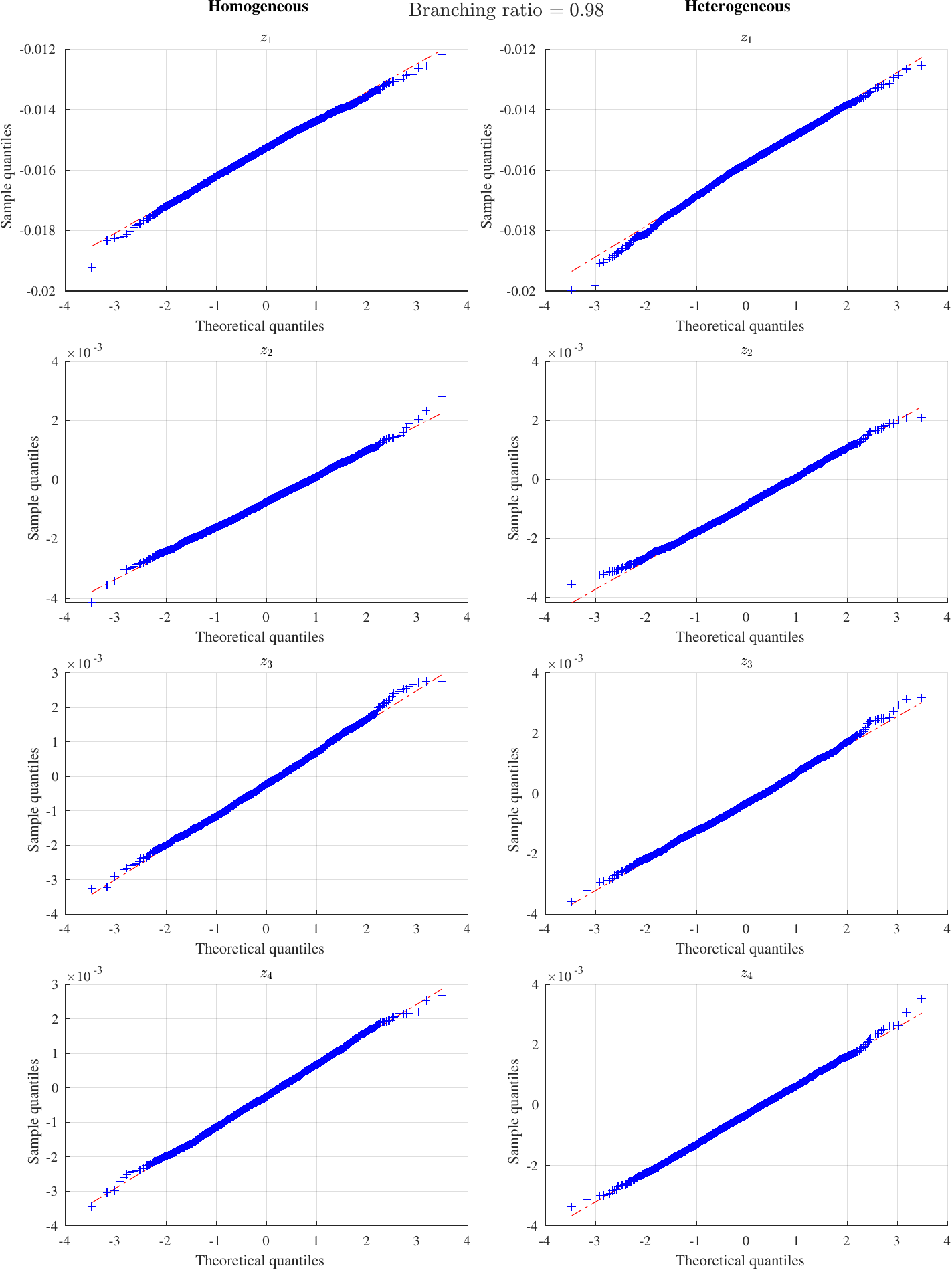}
\caption{Q-Q plots of the standardized statistics $(z_1,z_2,z_3,z_4)$ based on $2000$ replications under Case 1  with branching ratio $\text{BR}_3=0.98$, $c=0.005$. The left panels correspond to the homogeneous case with $\nu^*=1$, $\alpha^*=(160,660,3300)'$ and $b^*=(500,2000,10000)'$. The right panels correspond to the heterogeneous case with $\nu^*=1$, $\alpha^*=(90,620,4900)'$ and $b^*=(500,2000,10000)'$.}
\label{qq_plot_98}
\end{figure}

\newpage

\section{Real data analysis: in sample evaluation}
\label{secempiricalstudyinsample}

In this section, we evaluate the quality of the in-sample goodness of fit through the determination coefficient $R^2$ and the Akaike Information Criterion (AIC) metrics using the data described in Section \ref{secempirical} of the main text. Overall, the results indicate that multiple functions improve the fit, but the marginal benefits decrease beyond three or four functions (i.e., $n_h \geq 3$) and that the choice of the time increment $\Delta$ remains asset-specific. This can be interpreted as evidence for the existence of three different types of traders in each asset, a heterogeneity property investigated in Section \ref{secempiricalstudytest} of the main text.

For each trading day and for a fixed time increment $\Delta$, we estimate the Hawkes process by ordinary least squares for each function number $n_h=2,3,4,5$. The same time increment $\Delta$ is used for every value of the number $n_h$. This allows us to compute the determination coefficient $R^2$ based on the same time increment for any number $n_h$. The unknown numbers $b_k^*$ are estimated with maximum likelihood as there are unknown. The daily in-sample Akaike Information Criterion metric is computed by $\text{AIC}\big(\widehat{\theta}_{T\Delta},\widehat{b}^{\text{mle}}_{T\Delta}\big)$ from Definition (\ref{defaic}). Here, the starting time is the beginning of the trading session (9:30 am) and time is measured in seconds after 9:30 am. 
For every value of the number $n_h$, the time increment is estimated as $\Delta=c/\widehat{b}^{\text{mle}}_{\text{av},T\Delta}$. Here, $\widehat{b}^{\text{mle}}_{\text{av},T\Delta}=\frac{1}{249}\sum_{j=1}^{249}\max_{1\leq k\leq 3}\widehat{b}^{\text{mle}}_{j,k,T\Delta}$,
and $\widehat{b}^{\text{mle}}_{j,k,T\Delta}$ denotes the $k$-th component of $b^*$ estimated by maximum likelihood on day $j$ under the Hawkes model with $n_h=3$ functions in the kernel. The choice $n_h=3$ provides an intermediate level of model complexity between  parsimonious ($n_h=2$) and more flexible ($n_h=4,5$) specifications, and is therefore used as the reference model to define a common time increment across all values of $n_h$. To assess the sensitivity with respect to the choice of the time increment $\Delta$ and to clarify the presentation of the results, we report the determination coefficient $R^2$ and the Akaike Information Criterion when $c=0.316$, $c=0.720$, $c=1.638$ and $c=3.728$. These parameter values of $c$ are included in the grid $\mathcal{C}=\{10^{-6.5},\ldots,10^2\}$.

Table \ref{tab:insample_R2_AIC} reports the average daily in-sample values of the determination coefficient $R^2$ and the Akaike Information Criterion for different
time increment parameters $c$ and numbers of kernels $n_h$. This suggests that increasing the number of kernels $n_h$ generally improves the in-sample fit. The determination coefficient $R^2$ is particularly improved when moving from $n_h=2$ to $n_h=3$, whereas the gains resulting from $n_h\geq 4$ are comparatively small. The results are nevertheless asset-dependent. For META and AAPL, the determination coefficient $R^2$ increases with the time increment parameter $c$ while it reaches a maximum at intermediate values of $c$ for EOG. The Akaike Information Criterion results favor intermediate time increment parameters such as $c=0.316$ or $c=0.720$. In general, the Akaike Information Criterion results select the most flexible specification with $n_h=5$. However, the Akaike Information Criterion values obtained with $n_h=4,5$ are nearly identical for META, suggesting that the four functions may already capture the underlying information.  

\begin{table}[H]
\centering
\caption{Average daily in-sample values of the determination coefficient $R^2$ and the Akaike Information Criterion for different
time increment parameters $c$ and numbers of kernels $n_h$. Larger values of the determination coefficient $R^2$
indicate a better fit. Since the Akaike Information Criterion is computed from the Hawkes
log-likelihood, its value should be compared across
different values of $n_h$ and $c$ for the same asset.}
\label{tab:insample_R2_AIC}
\resizebox{\textwidth}{!}{
\begin{tabular}{l|cccccc|cccccc}
\hline\hline
&
\multicolumn{6}{c|}{Determination coefficient $R^2$}
&
\multicolumn{6}{c}{Akaike Information Criterion}
\\
&
$c=0.005$
&
$c=0.027$
&
$c=0.316$
&
$c=0.720$
&
$c=1.638$
&
$c=3.728$
&
$c=0.005$
&
$c=0.027$
&
$c=0.316$
&
$c=0.720$
&
$c=1.638$
&
$c=3.728$
\\
\hline

$n_h=2$ &&&&&&&&&&&&\\

AAPL
&0.000&0.002&0.028&0.047&0.053&0.053
&-1021338&-1021560&-1024405&-1025964&-1015842&-988459\\

META
&0.000&0.002&0.021&0.034&0.036&0.039
&-620290&-620392&-621888&-622460&-614700&-598450\\

EOG
&0.002&0.008&0.078&0.100&0.102&0.076
&-129569&-129309&-128580&-129344&-125498&-114816\\
\hline

$n_h=3$ &&&&&&&&&&&&\\

AAPL
&0.000&0.002&0.029&0.048&0.055&0.057
&-1032474&-1033686&-1040583&-1041940&-1031255&-1009110\\

META
&0.000&0.002&0.023&0.036&0.038&0.043
&-628033&-629184&-634161&-634423&-625127&-612882\\

EOG
&0.002&0.008&0.081&0.102&0.103&0.077
&-132128&-132055&-131446&-131239&-126736&-116305\\
\hline

$n_h=4$ &&&&&&&&&&&&\\

AAPL
&0.000&0.002&0.029&0.049&0.055&0.057
&-1036485&-1038367&-1043655&-1046289&-1036476&-1014590\\

META
&0.000&0.002&0.022&0.036&0.038&0.044
&-632951&-633776&-636429&-637788&-629378&-617592\\

EOG
&0.002&0.008&0.083&0.103&0.103&0.077
&-132794&-132801&-131922&-131608&-126907&-116839\\
\hline

$n_h=5$ &&&&&&&&&&&&\\

AAPL
&0.000&0.002&0.029&0.048&0.054&0.058
&-1039979&-1041067&-1044497&-1046915&-1037306&-1018961\\

META
&0.000&0.002&0.022&0.035&0.038&0.044
&-633238&-634072&-637815&-637254&-631380&-620007\\

EOG
&0.002&0.008&0.083&0.103&0.103&0.077
&-132833&-132969&-132083&-131724&-126982&-117065\\

\hline\hline
\end{tabular}
}
\end{table}

\newpage

\section{Proofs}
\label{secproofs}

The main idea of the proofs is to apply the technology of \cite{white2001asymptotic} (Theorem 5.17, p. 126). More specifically, we consider central limit theorem for time series with observations recorded at successive equally spaced points in time. We introduce the time series $z=\{z_{t\Delta} \}_{t \in \naturels_*}$ of dimension $n+1$ in which $z_{t\Delta}=(x_{t\Delta},u_{t\Delta})$ for any time index $t \in \naturels_*$. Since Condition (ii) from Theorem 5.17 in \cite{white2001asymptotic} requires that the time series $z$ is stationary, we construct a stationary time series $\overline{z}$ in this first step from the proofs.

We first introduce the following definition of stationarity for stochastic processes $P=\{P_t\}_{t \in \reels^+}$ defined on the continuous time interval $\reels^+$.

\begin{definition}[stationary continuous time stochastic process]
\label{defstationarycontinuous}

We introduce the stochastic process $P=\{P_t\}_{t \in \reels^+}$ defined on the continuous time interval $\reels^+$ and the function ${\displaystyle F_{P}(p_{t_{1}+\tau},\ldots,p_{t_{m}+\tau})}$ which represents the cumulative distribution function of the unconditional joint distribution of the stochastic process $P$ at times $t_{1}+\tau,\ldots,t_{m}+\tau$. Then, the stochastic process $P$ is said to be stationary if
$${\displaystyle
F_{P}(p_{t_{1}+\tau},\ldots,p_{t_{m}+\tau})
=
F_{P}(p_{t_{1}},\ldots,p_{t_{m}})
\quad
{\text{for any time }} \tau,t_{1},\ldots,t_{m}\in\reels^+
{\text{ and any index }} m\in\mathbb{N}_{*}.
}$$
\end{definition}


Since the time $\tau$ does not affect the cumulative distribution function $F_P$, we have that the cumulative distribution function $F_P$ is independent of time in Definition \ref{defstationarycontinuous}. If $Y$ is a random variable, we define its $L^p$ norm as $||Y||_p = \esp [|Y|^p]^{1/p}$. The next lemma proves the existence of a simple point process $\overline{N}=\{\overline{N}_t\}_{t \in \reels^+}$ which is a stationary stochastic process in the sense of Definition \ref{defstationarycontinuous} with its intensity denoted $\overline{\lambda}$ satisfying Equation (\ref{defhawkesintensityfamily}). The proof slightly extends the arguments from Proposition 4.4 (pp. 1819-1820) in \cite{clinet2017statistical} along with the stationary condition $\| h \|_1 < 1$ and the condition that the functions $h_k$ are exponential.

\begin{lemma}[existence]
\label{lemstatexistence}
We assume that Conditions \ref{condh}  \ref{condhparameterspace}, \ref{condhhexp}, \ref{condhidentifiability} and \ref{condhstationary} hold. Then: 

(i) There exists an extension of the stochastic basis $\mathbf{B} = \big(\Omega, \sigalg, \filt=\{\filt_t\}_{t \in \reels^+}, \proba\big)$ denoted by $\overline{\mathbf{B}} = \big(\overline{\Omega}, \overline{\sigalg}, \overline{\filt}=\{\overline{\filt}_t\}_{t \in \reels^+}, \overline{\proba}\big)$ and a simple point process $\overline{N}=\{\overline{N}_t\}_{t \in \reels^+}$ defined on the extended stochastic basis $\overline{\mathbf{B}}$ which is a stationary stochastic process in the sense of Definition \ref{defstationarycontinuous} with its intensity denoted $\overline{\lambda}$ satisfying Equation (\ref{defhawkesintensityfamily}).

(ii) There exists a positive real number $q \in \reels_*^+$ and a positive real number $K \in \reels_*^+$ such that for any time $t \in \reels^+$ we have
$$|| \lambda_t - \overline{\lambda}_t ||_1 \leq K \exp{(-qt)}. $$

(iii) We have that the compensated simple point process $\overline{A}=\{\overline{A}_t\}_{t \in \reels^+}$ defined for any time $t \in \reels^+$ as
$\overline{A}_{t} = \overline{N}_{t} - \int_0^t \overline{\lambda}_s ds$
is a martingale with respect to the filtration $\overline{\filt}$ almost surely.
\end{lemma}

\begin{proof}[Proof of Lemma \ref{lemstatexistence}]
First, Condition \ref{condh} yields the stationary condition $\| h \|_1 < 1$. Then, we can conclude by an application of Proposition 4.4 (i) (pp. 1819-1820) in \cite{clinet2017statistical} along with the stationary condition $\| h \|_1 < 1$ to prove (i). To prove (ii), we can extend the arguments from the proof of Proposition 4.4 (iii) (pp. 1819-1820) in \cite{clinet2017statistical} as the kernels $h_k$ are exponential functions by Condition \ref{condh} \ref{condhhexp}. Finally, (iii) is obtained by the Doob–Meyer decomposition theorem (see Theorem I.3.17 (p. 32) in \cite{jacod2003limit}) of the simple point process $\overline{N}$ established in (i).
\end{proof}

We now define the observed time series $\overline{y}=\{\overline{y}_{t\Delta}\}_{ t \in \naturels_* }$ of the stationary Hawkes process $\overline{N}$ for any time index $t \in \naturels_*$ as 
$$\overline{y}_{t\Delta}= \overline{N}_{(t\Delta)^-} -\overline{N}_{(t-1)\Delta}.$$ 
We also denote the explaining vector of the stationary Hawkes process $\overline{N}$ for any time index $t \in \naturels_*$ by
\be
\label{defxoverline}
\overline{x}_{t\Delta}=\Big(1, \int_{(t-1)\Delta}^{t\Delta} h_1(t-s) d\overline{N}_s, \ldots, \int_{(t-1)\Delta}^{t\Delta} h_{n-1}(t-s) d\overline{N}_s \Big)'.
\ee
Then, we define the error term of the stationary Hawkes process $\overline{N}$ for any time index $t \in \naturels_*$ as
\be
\label{defuoverline}
\overline{u}_{t\Delta} = \overline{y}_{t\Delta} - \overline{x}_{t\Delta}'\theta^* .
\ee
In matrix notation this relationship
is written as 
\be
\label{deflrmoverline}
\overline{Y}_{T\Delta} = \overline{X}'_{T\Delta} \theta^{*} + \overline{U}_{T\Delta}.
\ee
Here, we introduce the vector $\overline{Y}_{T\Delta}=(\overline{y}_{\Delta}, \ldots, \overline{y}_{T\Delta})'$ of dimension $T$. We also introduce the matrix $\overline{X}_{T\Delta}=(\overline{x}_{\Delta}, \ldots, \overline{x}_{T\Delta})$ of dimension $n \times T$ and the vector $\overline{U}_{T\Delta}=(\overline{u}_{\Delta}, \ldots, \overline{u}_{T\Delta})'$ of dimension $T$.

We estimate the unknown parameter $\theta^*$ by minimizing the ordinary least squares objective function
\be
\label{defSSRoverline}
\overline{SSR}_{T\Delta}(\theta)= \sum_{t=1}^T \big(\overline{y}_{t\Delta} - \overline{x}_{t\Delta} ' \theta \big)^2 = \big( \overline{Y}_{T\Delta} - \overline{X}'_{T\Delta} \theta \big) ' \big( \overline{Y}_{T\Delta} - \overline{X}'_{T\Delta} \theta \big).
\ee
 Then, the least squares estimator is defined as the solution to the minimization problem (\ref{defSSRoverline}) between the starting time $0$ and the final time $T\Delta$, namely
$$\widehat{\overline{\theta}}_{T\Delta} = \operatorname{argmin}_{\theta \in \Theta} \overline{SSR}_{T\Delta} (\theta).$$
The first-order conditions for a minimum are
$$\frac{\partial \overline{SSR}_{T\Delta}(\theta)}{
\partial \theta} = -2 \sum_{t=1}^T \overline{x}_{t\Delta} \big(\overline{y}_{t\Delta} - \overline{x}_{t\Delta} ' \theta \big) = -2 \overline{X}_{T\Delta} \big( \overline{Y}_{T\Delta} - \overline{X}_{T\Delta}' \theta \big).$$
If the matrix $\overline{X}_{T\Delta} \overline{X}_{T\Delta}' = \sum_{t=1}^T \overline{x}_{t\Delta} \overline{x}_{t\Delta} '$ is non singular, this system of $n$ equations in $n$
unknowns can be uniquely solved for the ordinary least squares 
estimator
$$\widehat{\overline{\theta}}_{T\Delta} = \Big( \sum_{t=1}^T \overline{x}_{t\Delta} \overline{x}_{t\Delta}' \Big)^{-1} \sum_{t=1}^T \overline{x}_{t\Delta} \overline{y}_{t\Delta} = \big( \overline{X}_{T\Delta} \overline{X}_{T\Delta}' \big)^{-1} \overline{X}_{T\Delta} \overline{Y}_{T\Delta}.$$
This is exactly the objective function and the estimator discussed in \cite{white2001asymptotic} (p. 2).

We introduce the quantities required for the asymptotic variance-covariance matrix. First, we define the variance matrix $\overline{V}_{T\Delta}$ at the time $T\Delta$ of  dimension $n \times n$ for any time index $T\in \naturels_*$ as
\be
\label{defVToverline}
\overline{V}_{T\Delta} = \var \big(T^{-1/2} \overline{X}_{T\Delta}\overline{U}_{T\Delta}\big).
\ee
Then, we define its limiting variance matrix when the final time $T \rightarrow \infty$ as
\be
\label{defVoverline}
\overline{V} = \lim_{T \rightarrow \infty} \overline{V}_{T\Delta}.
\ee
We also define the matrix $\overline{M}$ as
\be
\label{defMoverline}
\overline{M} = \esp [\overline{x}_{T\Delta} \overline{x}_{T\Delta}'].
\ee
By definition, the matrix $\overline{M}$ does not depend on the final time $T$ as the simple point process $\overline{N}$ is stationary in the sense of Definition \ref{defstationarycontinuous}. Finally, we define the asymptotic variance matrix $\overline{D}$ of dimension $n \times n$ as
\be
\label{defDoverline}
\overline{D}=\overline{M}^{-1} \overline{V} \,\overline{M}^{-1}.
\ee

We now give the following lemma which states that if the central limit theorem based on the stationary point process $\overline{N}$ holds, then the central limit theorem based on the nonstationary point process $N$ holds. The proof is based on an application of Lemma \ref{lemstatexistence}.

\begin{lemma}
\label{lemstat}
We assume that Condition \ref{condh} holds. We also assume the central limit theorem of the parametric estimation procedure based on ordinary least squares with the stationary point process $\overline{N}$ as the final time $T \rightarrow + \infty$, namely
\begin{eqnarray}
\label{clteqstat}
\overline{D}^{-1/2} \sqrt{T\Delta}(\widehat{\overline{\theta}}_{T\Delta} - \theta^*) \cvdistrib  \xi.
\end{eqnarray}
Then, we have the central limit theorem of the parametric estimation procedure based on ordinary least squares with the nonstationary point process $N$ as the final time $T \rightarrow + \infty$, namely
\begin{eqnarray}
\label{lemstat1}
D^{-1/2} \sqrt{T\Delta}(\widehat{\theta}_{T\Delta} - \theta^*) \cvdistrib  \xi.
\end{eqnarray}
\end{lemma}

\begin{proof}[Proof of Lemma \ref{lemstat}] 
We first have the decomposition 
\begin{eqnarray}
\label{lemstatp1}
D^{-1/2} \sqrt{T\Delta}(\widehat{\theta}_{T\Delta} - \theta^*) & = & \overline{D}^{-1/2} \sqrt{T\Delta}(\widehat{\overline{\theta}}_{T\Delta} - \theta^*) + (D^{-1/2} - \overline{D}^{-1/2}) \sqrt{T\Delta}(\widehat{\overline{\theta}}_{T\Delta} - \theta^*) 
\\ \nonumber && + D^{-1/2} \sqrt{T\Delta}(\widehat{\theta}_{T\Delta} - \widehat{\overline{\theta}}_{T\Delta}).
\end{eqnarray}
We introduce the sequence of random variables $I=\overline{D}^{-1/2} \sqrt{T\Delta}(\widehat{\overline{\theta}}_{T\Delta} - \theta^*).$ We can deduce by Expression (\ref{clteqstat}) as the final time $T \rightarrow + \infty$ that
\begin{eqnarray}
\label{lemstatp2}
I \cvdistrib  \xi.
\end{eqnarray}
We also introduce the sequence of random variables $II=(D^{-1/2} - \overline{D}^{-1/2}) \sqrt{T\Delta}(\widehat{\overline{\theta}}_{T\Delta} - \theta^*).$ We obtain by Expression (\ref{clteqstat}) and Lemma \ref{lemstatexistence} (ii) as the final time $T \rightarrow + \infty$ that
\begin{eqnarray}
\label{lemstatp3}
II \cvproba  0.
\end{eqnarray}
Moreover, we introduce the sequence of random variables $III=D^{-1/2} \sqrt{T\Delta}(\widehat{\theta}_{T\Delta} - \widehat{\overline{\theta}}_{T\Delta}).$ We get by Expression (\ref{clteqstat}) and Lemma \ref{lemstatexistence} (ii) as the final time $T \rightarrow + \infty$ that
\begin{eqnarray}
\label{lemstatp4}
III \cvproba  0.
\end{eqnarray}
Finally, we can prove Expression (\ref{lemstat1}) with the use of Expressions (\ref{lemstatp1}), (\ref{lemstatp2}), (\ref{lemstatp3}) and (\ref{lemstatp4}).  
\end{proof}

By an application of Lemma \ref{lemstat}, it remains to show the central limit theorem based on the stationary point process $\overline{N}$, namely Expression (\ref{clteqstat}), holds. This will rely on the technology of \cite{white2001asymptotic} (Theorem 5.17, p. 126). Since Condition (ii) from Theorem 5.17 in \cite{white2001asymptotic} requires that the time series $z$ is stationary, we construct a stationary time series $\overline{z}$. More specifically, we introduce the time series $\overline{z}=\{\overline{z}_{t\Delta} \}_{t \in \naturels_*}$ of dimension $n+1$ in which $\overline{z}_{t\Delta}'=(\overline{x}_{t\Delta}',\overline{u}_{t\Delta}')$ for any time index $t \in \naturels_*$.

In what follows, we introduce the notion of stationarity for a time series. This requires to define first a measure-preserving transformation. We set $g_1 (\omega) = B (\omega)$ in which $B$ is a measurable function from the sample space $\Omega$ to the space $\reels$. Then,
we can define the random variables $g_2(\omega) = B(D \omega)$, $g_3(\omega) = B(D^2\omega)$,
and so on, provided that $D$ is a measurable transformation. The random
variables constructed in this way are said to be random variables induced
by a measurable transformation. The following definition corresponds exactly to Definition 3.27 (p. 42) in \cite{white2001asymptotic}. 

\begin{definition}[measure-preserving transformation]
\label{defmeasprestrans}
The transformation $D : \Omega \rightarrow \Omega$ is measure preserving if it is measurable with respect to the $\sigma$-algebra $\sigalg$ and if $\proba(D^{-1} F) =
\proba(F)$ for any measured set $F \in \sigalg$.
\end{definition} 

The random variables induced by measure-preserving transformations
then have the property that $\proba (g_1 < b) = \proba (\omega : B(\omega) < b) = \proba ( \omega :
B(D\omega) \leq b) = \proba ( g_2 \leq b )$, namely they are identically distributed. In fact, such random variables have an even stronger property. We use the following
definition of stationary time series. The following definition corresponds exactly to Definition 3.28 (p. 43) in \cite{white2001asymptotic}. 

\begin{definition}[stationary time series]
\label{defstationary}
We define $G_1$ as the joint distribution function of the time series $g=\{g_{t} \}_{t \in \naturels_*}$ in which $g_t$ is a vector of dimension $q \times 1$ for any time $t \in \naturels_*$. We also define $G_{\tau + 1}$ as the joint distribution function of the time series $\{g_{\tau+t} \}_{t \in \naturels_*}$. The time series $g$
is said to be stationary if $G_1 = G_{\tau+1}$ for any index $\tau \in \naturels_*$.
\end{definition}

In other words, a time series is stationary in the sense of Definition \ref{defstationary} if the joint distribution of the
variables in the time series is identical, regardless of the date of the first
observation. We  now give some well-known properties about stationary time series. The following lemma corresponds exactly to Proposition 3.29 (p. 43) in \cite{white2001asymptotic}. The proof can be obtained by \cite{stout1974almost} (p. 169).

\begin{lemma}
\label{lemwhite2001prop329}
We introduce a random variable $B$ and a measure-preserving transformation $D$. We also introduce $g_1(\omega) = B(\omega), g_2(\omega) = B(D\omega), \ldots , g_t(\omega) = B(D^{t-1} \omega)$ for any $\omega \in \Omega$ and any time $t \in \naturels_*$. Then,
$g=\{g_{t} \}_{t \in \naturels_*}$ is a stationary time series in the sense of Definition \ref{defstationary}. 
 \end{lemma}

 \begin{proof}
[Proof of Lemma \ref{lemwhite2001prop329}]
The proof can be obtained by \cite{stout1974almost} (p. 169).
\end{proof}

A converse to this result is also available. The following lemma corresponds exactly to Proposition 3.30 (p. 43) in \cite{white2001asymptotic}. The proof can be obtained by \cite{stout1974almost} (p. 170).

\begin{lemma}
\label{lemwhite2001prop330}
We introduce a time series $g=\{g_{t} \}_{t \in \naturels_*}$ stationary in the sense of Definition \ref{defstationary}. Then, there exists a
measure-preserving transformation $D$ in the sense of Definition \ref{defmeasprestrans} defined on the stochastic basis $\mathbf{B}$ such that $g_1 (\omega)
= g_1 (\omega)$, $g_2 (\omega)
= g_1 (D \omega)$, $g_3 (\omega)
= g_1 (D^2 \omega), \ldots , g_t (\omega)
= g_1 (D^{t-1} \omega)$
for any $\omega \in \Omega$ and any time $t \in \naturels_*$. 
\end{lemma}

\begin{proof}
[Proof of Lemma \ref{lemwhite2001prop330}]
The proof can be obtained by \cite{stout1974almost} (p. 170).
\end{proof}

We now show in the following lemma that the time series $\overline{z}$ is stationary in the sense of Definition \ref{defstationary}. This statement can be proven by an application of Lemma \ref{lemstatexistence} (i). Indeed, the stationarity in the sense of continuous time stochastic processes implies the stationarity in the sense of time series.

\begin{lemma}
\label{lemstationary}
We assume that Condition \ref{condh} holds. Then, the time series $\overline{z}$ is stationary in the sense of Definition \ref{defstationary}.
\end{lemma}

\begin{proof}[Proof of Lemma \ref{lemstationary}]
This statement can be proven by an application of Lemma \ref{lemstatexistence} (i) with the definition of the time series $\overline{z}$.
\end{proof}

To obtain the law of large numbers and further the central limit theorem, we have to impose a restriction on
the dependence or memory of the time series. One such restriction is the
concept of ergodicity. The following definition corresponds exactly to Definition 3.33 (p. 44) in \cite{white2001asymptotic}.

\begin{definition}[ergodic time series]
\label{defergodicity}
We introduce a time series $g=\{g_{t} \}_{t \in \naturels_*}$ which is stationary in the sense of Definition \ref{defstationary} and the measure-preserving transformation $D$ obtained from
Lemma \ref{lemwhite2001prop330}. Then, the time series $g$ is said to be ergodic if for any measured set $F \in \sigalg$ and any measured set $G \in \sigalg$ we have
$$\lim_{T \rightarrow \infty} T^{-1} \sum_{t=1}^T \proba (F \cap D^t G) = \proba (F) \proba  (G).$$
\end{definition}

If the measured sets $F$ and $G$ were independent, then we would have $\proba (F \cap G) = \proba (F) \proba (G)$.
Thus, we can think of the event $D^t G$ as being the event $G$ shifted $t$ periods into the future. Since $\proba (D^t G) = \proba ( G)$ when the transformation $D$ is measure preserving, Definition \ref{defergodicity}
says that an ergodic time series satisfies for any measured sets $F$
and $G$ that the measured sets $F$ and $D^t G$ are independent on average in the limit. Thus, ergodicity
can be thought of as a form of average asymptotic independence. For
more on measure-preserving transformations, stationarity and ergodicity
the reader may consult \cite{doob1953stochastic} (pp. 167-185) and \cite{rosenblatt1978dependence}.

To show that the time series $\overline{z}$ is ergodic in the sense of Definition \ref{defergodicity}, we have to introduce dependence measures and mixing conditions. First, we define the measure of dependence between two $\sigma$-algebras $\mathcal{A} \subset \sigalg$ and $\mathcal{B} \subset \sigalg$ as
\be
\label{defmeasdep}
\alpha(\mathcal{A}, \mathcal{B}) = \sup_{A \in \mathcal{A}, B \in \mathcal{B}} \big| \mathbb{P}(A \cap B) - \mathbb{P}(A) \mathbb{P}(B)\big|.
\ee
Then, we define the dependence coefficient for any nonnegative integer $m \in \naturels$ as
\be
\alpha (m) = \sup_{j \in \naturels_*} \alpha(\mathcal{H}_j^g, \mathcal{H}_{j+m,\infty}^{g}).
\ee
Here, $\mathcal{H}_{j}^{g} = \mathcal{B} \{ g_t \text{ such that } t \in \naturels_* \text{ and } j \geq t\}$ is the forward looking $\sigma$-algebra generated by the time series $g$ until the time $j \in \naturels_*$. Moreover, $\mathcal{H}_{j,\infty}^{g} = \mathcal{B} \{ g_t \text{ such that } t \in \naturels_* \text{ and } j \leq t \}$ is the backward looking $\sigma$-algebra generated by the time series $g$ from the time $j \in \naturels$. The following definition introduces the notion of strongly mixing for time series which is also called $\alpha$-mixing. This definition was introduced in \cite{rosenblatt1956central}.

\begin{definition}[strongly mixing time series]
\label{defstronglymixing}
The time series $g=\{g_{t} \}_{t \in \naturels_*}$
 is said to be strongly mixing or $\alpha$-mixing if the dependence coefficient $\alpha(m) \rightarrow 0$ as the index $m \rightarrow \infty$.
\end{definition}

We define another measure of dependence between the two $\sigma$-algebras $\mathcal{A} \subset \sigalg$ and $\mathcal{B} \subset \sigalg$ as
\begin{equation}
\label{defmeasdeprho}
\rho(\mathcal{A}, \mathcal{B}) = \sup_{f_1 \in L^2(\mathcal{A}), f_2 \in L^2(\mathcal{B})} \big| \cor (f_1,f_2) \big|.
\end{equation}
Here, $L^2(\mathcal{A})$ denotes the space of square integrable random variables of the $\sigma$-algebra $\mathcal{A}$. In addition, $\cor (f_1,f_2)$ denotes the correlation between the two random variables $f_1$ and $f_2$. Then, we define the dependence coefficient for any nonnegative integer $m \in \naturels$ as
\begin{equation*}
\rho (m) = \sup_{j \in \naturels_*} \rho(\mathcal{H}_j^g, \mathcal{H}_{j+m,\infty}^{g}).
\end{equation*}
The following definition introduces the notion of $\rho$-mixing for time series.

\begin{definition}[$\rho$-mixing time series] 
\label{defrhomixing}
The time series $g=\{g_{t} \}_{t \in \naturels_*}$
 is said to be $\rho$-mixing if the dependence coefficient $\rho(m) \rightarrow 0$ as the index $m \rightarrow \infty$.
\end{definition}

The following lemma is a  well-known result which gives a bound between the two measures of dependence $\alpha$ and $\rho$ for any time series.
For a history and proof of this inequality, see Theorem 3.11 in \cite{bradley2007introduction} or Proposition 1 (p. 4) in \cite{doukhan1995mixing}.
\begin{lemma}
\label{lembound}
We have $4 \alpha(\mathcal{A}, \mathcal{B}) \leq \rho(\mathcal{A}, \mathcal{B})$ for any time series $g=\{g_{t} \}_{t \in \naturels_*}$, any $\sigma$-algebra $\mathcal{A} \subset \filt$ and any $\sigma$-algebra $\mathcal{B} \subset \filt$.
\end{lemma}

By the use of Lemma \ref{lembound}, we can deduce that $\rho$-mixing time series in the sense of Definition \ref{defrhomixing} implies strongly mixing time series in the sense of Definition \ref{defstronglymixing}. 

We now show that the time series $\overline{z}$ is $\rho$-mixing in the sense of Definition \ref{defrhomixing} in the following lemma. The proof extends the arguments from the proof of Lemma A.6 (p. 1834) in \cite{clinet2017statistical}  and the proof of Proposition C1 (i) in Supplement C (p. 56) of \cite{potiron2026mutually}.

\begin{lemma}
\label{lemrhomixing}
We assume that Condition \ref{condh} holds. Then, the time series $\overline{z}$ is $\rho$-mixing in the sense of Definition \ref{defrhomixing}.
\end{lemma}

\begin{proof}[Proof of Lemma \ref{lemrhomixing}]
This statement can be proven by an extension of the arguments in the proofs of Lemma A.6 (p. 1834) in \cite{clinet2017statistical} and Proposition C1 (i) in Supplement C (p. 56) of \cite{potiron2026mutually}.
\end{proof}

We now prove that the time series $\overline{z}$ is strongly mixing in the sense of Definition \ref{defrhomixing} in the following lemma. The proof is a direct consequence to Lemma \ref{lembound} and Lemma \ref{lemrhomixing}.

\begin{lemma}
\label{lemstrongmixing}
We assume that Condition \ref{condh} holds. Then, the time series $\overline{z}$ is strongly mixing in the sense of Definition \ref{defstronglymixing}.
\end{lemma}

\begin{proof}[Proof of Lemma \ref{lemstrongmixing}]
The proof is a direct consequence to Lemma \ref{lembound} and Lemma \ref{lemrhomixing}.
\end{proof}

The notion of mixing is a stronger memory requirement than the notion of
ergodicity for stationary time series. Given stationarity, mixing implies
ergodicity as the next result makes precise. The following lemma corresponds exactly to Proposition 3.44 (p. 48) in \cite{white2001asymptotic}. The proof can be obtained by \cite{rosenblatt1978dependence}. 

\begin{lemma}
\label{lemwhite2001prop344}
We introduce the time series $g=\{g_{t} \}_{t \in \naturels_*}$ which is stationary in the sense of Definition \ref{defstationary}. We assume that the time series $g$
 is strongly mixing in the sense of Definition \ref{defstronglymixing}, namely the correlation coefficient $\alpha (m) \rightarrow 0$ as the index $m \rightarrow \infty$. Then, the time series $g$ is ergodic in the sense of Definition \ref{defergodicity}. 
\end{lemma}

\begin{proof}
[Proof of Lemma \ref{lemwhite2001prop344}]
This can be obtained by \cite{rosenblatt1978dependence}. 
\end{proof} 

As we have that $\rho$-mixing implies strong mixing by Lemma \ref{lembound}, $\rho$-mixing time series are also ergodic. However, ergodic time series are not necessarily
mixing. Mixing is defined for time series that
are not necessarily strictly stationary, so mixing is more general in this
sense. For more on mixing and ergodicity, see \cite{rosenblatt1972uniform} and \cite{rosenblatt1978dependence}.

The next lemma shows that the time series $\overline{z}$ is stationary and ergodic. This corresponds to Condition (ii) from Theorem 5.17 in \cite{white2001asymptotic}.

\begin{lemma}
\label{lemergodic}
We assume that Condition \ref{condh} holds. Then, the time series $\overline{z}$ is stationary and ergodic in the sense of Definition \ref{defergodicity}.
\end{lemma}

\begin{proof}[Proof of Lemma \ref{lemergodic}]
First, we have that the time series $\overline{z}$ is strongly mixing in the sense of Definition \ref{defstronglymixing} by an application of Lemma \ref{lemstrongmixing}. We also have that the time series $\overline{z}$ is stationary in the sense of Definition \ref{defstationary} by Lemma \ref{lemstationary}. Thus, we can deduce that the time series $\overline{z}$ is ergodic in the sense of Definition \ref{defergodicity} by an application of Lemma \ref{lemwhite2001prop344}.
\end{proof}

We introduce $\esp[X | \sigalg]$ which denotes the conditional expectation of the random variable $X$ with respect to the $\sigma$-algebra $\sigalg$. Then, the conditional expectation $\esp[X | \sigalg]$ is a random variable which is measurable with respect to the $\sigma$-algebra $\filt$. Moreover, we introduce the big O which is defined through
$$f(t) = O\big(g(t)\big) \iff  |f(t)| \leq C |g(t)|$$ 
for any time $t \in \naturels_*$ and some constant $C \in \reels^+$ which does not depend on the time $t$. 

In order to give a convenient statement of an ergodic central limit theorem, we introduce the notion of an adapted mixingale time series. This corresponds exactly to Definition 5.15 (pp. 124-125) in \cite{white2001asymptotic}.

\begin{definition}[mixingale]
\label{defmixingale}
We introduce the time series $g=\{g_{t} \}_{t \in \naturels_*}$ which is adapted to the filtration $\mathcal{H}=\{\mathcal{H}_{t} \}_{t \in \naturels_*}$ and satisfies $\esp [g_t^2]
< \infty$ for any time $t \in \naturels_*$. Then, the time series $g$ is a mixingale adapted to the filtration $\mathcal{H}$ if there exist finite nonnegative sequences $\{ c_t \}_{t \in \naturels_*}$ and $\gamma = \{ \gamma_m \}_{m \in \naturels_*}$ such that $\gamma_m \rightarrow 0$ as $m \rightarrow \infty$ and
$$\big(\esp \big[ \esp [ g_t | \mathcal{H}_{t-m} ]^2 \big] \big)^{1/2} \leq c_t \gamma_m.$$
We say that $\gamma$ is of size $-a$ if $\gamma_m = O(m^{- a - \epsilon})$ for some positive real number $\epsilon \in \reels^+_*$.
\end{definition}

The notion of a mixingale is due to \cite{mcleish1974dependent}. In this
definition, the time series $g$ need not be stationary or ergodic, but may be heterogeneous. In the mixingale definition of \cite{mcleish1974dependent}, the time series
$g$ need not be adapted to the filtration $\mathcal{H}$. Nevertheless, we impose this because it simplifies matters nicely and
is sufficient for all our applications. As the name is intended to suggest,
mixingale time series have attributes of both mixing time series and martingale difference time series. They can be thought of as time series that behave
asymptotically like martingale difference time series. 

We now show in the following lemma that the time series $\{ \overline{x}^{(i)}_{t\Delta}\overline{u}_{t\Delta} \}_{t \in \naturels_*}$ is a mixingale adapted to the filtration $\mathcal{H}=\{\filt_{t\Delta} \}_{t \in \naturels_*}$ of size $-1$ for any index $i = 1, \ldots ,n$. Here, $v^{(i)}$ denotes the $i$th component of the vector $v$. This corresponds to Condition (iii) (a) from Theorem 5.17 in \cite{white2001asymptotic} in the particular case $p=1$. The proof extends the arguments from the proof of Theorem 4.6 (p. 1821) in \cite{clinet2017statistical}  and the proof of Theorem 1 (p. 331) in \cite{potiron2026mutually}.

\begin{lemma}
\label{lemmixingale}
We assume that Condition \ref{condh} holds. Then, the time series $\{ \overline{x}^{(i)}_{t\Delta}\overline{u}_{t\Delta} \}_{t \in \naturels_*}$ is a mixingale adapted to the filtration $\mathcal{H}=\{\filt_{t\Delta} \}_{t \in \naturels_*}$ of size $-1$ for any index $i = 1, \ldots ,n$ in the sense of Definition \ref{defmixingale}.
\end{lemma}

\begin{proof}[Proof of Lemma \ref{lemmixingale}]
This lemma can be proven with an extension of the arguments in the proof of Theorem 4.6 (p. 1821) in \cite{clinet2017statistical}  and the proof of Theorem 1 (p. 331) in \cite{potiron2026mutually}.
\end{proof}

If $x$ is a real number, a vector, a matrix or a tensor, we define the sum of the absolute values of its components as $|x| = \sum_i |x_i|$. We now show in the following lemma that $\esp \big[ | \overline{x}_{t\Delta}^{(i)} \overline{u}_{t\Delta} |^2 \big] < \infty$ for any index $i = 1, \ldots ,n$ and any time index $t \in \naturels_*$. This corresponds to Condition (iii) (b) from Theorem 5.17 in \cite{white2001asymptotic} in the particular case $p=1$.

\begin{lemma}
\label{lemmoment}
We assume that Condition \ref{condh} holds. Then, we have for any index $i = 1, \ldots ,n$ and any time index $t \in \naturels_*$ that 
$$\esp \big[ | \overline{x}_{t\Delta}^{(i)} \overline{u}_{t\Delta} |^2 \big] < \infty.$$
\end{lemma}

\begin{proof}[Proof of Lemma \ref{lemmoment}]

First, we obtain by Cauchy-Schwarz inequality for any index $i = 1, \ldots ,n$ and any time index $t \in \naturels_*$ that
\be
\esp \big[ | \overline{x}_{t\Delta}^{(i)} \overline{u}_{t\Delta} |^2 \big] & \leq & \sqrt{\esp \big[ | \overline{x}_{t\Delta}^{(i)}  |^4 \big] \esp \big[ | \overline{u}_{t\Delta} |^4 \big]}.
\ee
To complete the proof of Lemma \ref{lemmoment}, it remains to show for any index $i = 1, \ldots ,n$ and any time index $t \in \naturels_*$ that
\be
\label{lemmomentp1}
\esp \big[ | \overline{x}_{t\Delta}^{(i)}  |^4 \big] < \infty.
\ee
and for any time index $t \in \naturels_*$ that
\be
\label{lemmomentp2}
\esp \big[ | \overline{u}_{t\Delta} |^4 \big] < \infty.
\ee

We first prove Expression (\ref{lemmomentp1}). First, we have by Definition (\ref{defx}) for any index $i = 1, \ldots ,n$ and any time index $t \in \naturels_*$ that
\be
\label{lQ2016A1iip1} \esp[|\overline{x}_{t\Delta}^{(i)} |^{4}] & = & \esp \Big[ \Big| \Big(1, \int_{(t-1)\Delta}^{t\Delta} h_1(t-s) d\overline{N}_s, \ldots, \int_{(t-1)\Delta}^{t\Delta} h_{n-1}(t-s) d\overline{N}_s \Big)^{(i)} \Big|^{4} \Big].
\ee
Then, we get by Condition \ref{condh} \ref{condhhexp} for any index $i = 1, \ldots ,n$ and any time index $t \in \naturels_*$ that
\be
\label{lQ2016A1iip2}
&& \esp \Big[ \Big| \Big(1, \int_{(t-1)\Delta}^{t\Delta} h_1(t-s) d\overline{N}_s, \ldots, \int_{(t-1)\Delta}^{t\Delta} h_{n-1}(t-s) d\overline{N}_s \Big)^{(i)}  \Big|^{4} \Big]\\ \nonumber & = & \esp \Big[ \Big|  \Big(1, \int_{(t-1)\Delta}^{t\Delta} \exp \big(-b_1(t-s) \big) d\overline{N}_s, \ldots, \int_{(t-1)\Delta}^{t\Delta} \exp \big(-b_{n-1}(t-s) \big) d\overline{N}_s \Big)^{(i)}  \Big|^{4} \Big].
\ee
In addition, we introduce $b_- = \min (b_1, \ldots, b_{n-1})$ and we obtain for any index $i = 1, \ldots ,n$ and any time index $t \in \naturels_*$ that
\be
\label{lQ2016A1iip4}
&& \esp \Big[ \Big| \Big(1, \int_{(t-1)\Delta}^{t\Delta} \exp \big(-b_1(t-s) \big) d\overline{N}_s, \ldots, \int_{(t-1)\Delta}^{t\Delta} \exp \big(-b_{n-1}(t-s) \big) d\overline{N}_s \Big)^{(i)} \Big|^{4} \Big]\\ \nonumber & \leq & \esp \Big[ \Big|  \Big(1, \int_{(t-1)\Delta}^{t\Delta} \exp \big(-b_-(t-s) \big) d\overline{N}_s, \ldots, \int_{(t-1)\Delta}^{t\Delta} \exp \big(-b_{-}(t-s) \big) d\overline{N}_s \Big)^{(i)} \Big|^{4} \Big].
\ee
Finally, Expression (\ref{lemmomentp1}) can be shown with the martingale definition from Lemma \ref{lemstatexistence} (iii). 

We prove now Expression (\ref{lemmomentp2}). First, we have by Definition (\ref{defuoverline}) for any time index $t \in \naturels_*$ that
\be
\label{lQ2016A1iip21}
\esp \big[ | \overline{u}_{t\Delta} |^4 \big] & = & \esp[|\overline{N}_{(t\Delta)^-} - \overline{N}_{(t-1)\Delta} - \theta^* \overline{x}_{t\Delta} |^{4}].
\ee
Then, we obtain by the triangle inequality for any time index $t \in \naturels_*$ that
\be
\label{lQ2016A1iip22}
\esp[|\overline{N}_{(t\Delta)^-} - \overline{N}_{(t-1)\Delta} - \theta^* \overline{x}_{t\Delta} |^{4}] & \leq &  \esp[(|\overline{N}_{(t\Delta)^-} -\overline{N}_{(t-1)\Delta}| + | \overline{x}'_{t\Delta}\theta^* |)^{4}].
\ee
Moreover, we get by Condition \ref{condh} \ref{condhparameterspace} that there exists a real positive number $\theta^+ \in \reels_*^+$ for any time index $t \in \naturels_*$ satisfying
\be
\label{lQ2016A1iip23}
\esp[(|\overline{N}_{(t\Delta)^-} -\overline{N}_{(t-1)\Delta}| + | \theta^* \overline{x}_{t\Delta} |)^{4}] & \leq & \esp[(|\overline{N}_{(t\Delta)^-} -\overline{N}_{(t-1)\Delta}| + \theta^+ |  \overline{x}_{t\Delta} |)^{4}].
\ee
Finally, Expression (\ref{lemmomentp2}) can be shown with the martingale definition from Lemma \ref{lemstatexistence} (iii) and Expression (\ref{lemmomentp1}).
\end{proof}

We now show in the following lemma that the variance matrix $\overline{V}_{T\Delta}$ is positive definite uniformly over the time index $T \in \naturels_*$. This corresponds to Condition (iii) (c) from Theorem 5.17 in \cite{white2001asymptotic}.

\begin{lemma}
\label{lemmatrix}
We assume that Condition \ref{condh} holds. Then, we have that the variance matrix $\overline{V}_{T\Delta}$ is positive definite uniformly over the time index $T \in \naturels_*$.
\end{lemma}

\begin{proof}[Proof of Lemma \ref{lemmatrix}]
From Definition  (\ref{defVToverline}), we have that $\overline{V}_{T\Delta}$ is a variance matrix. Thus, the matrix $\overline{V}_{T\Delta}$ is positive definite by variance matrix properties. From the martingale definition in Lemma \ref{lemstatexistence} (iii), we can further deduce that the matrix $\overline{V}_{T\Delta}$ is positive definite uniformly over the time index $T \in \naturels_*$. 
\end{proof}

We now show in the following lemma that $\esp \big[ | \overline{x}_{t\Delta}^{(i)}  |^2 \big] < \infty$ for any index $i = 1, \ldots ,n$ and any time index $t \in \naturels_*$. This corresponds to Condition (iv) (a) from Theorem 5.17 in \cite{white2001asymptotic} in the particular case $p=1$. The proof follows the arguments from the proof of Lemma \ref{lemmoment}.

\begin{lemma}
\label{lemmoment2}
We assume that Condition \ref{condh} holds. Then, we have that $\esp \big[ | \overline{x}_{t\Delta}^{(i)}  |^2 \big] < \infty$ for any index $i = 1, \ldots ,n$ and any time index $t \in \naturels_*$.
\end{lemma}

\begin{proof}[Proof of Lemma \ref{lemmoment2}]
The proof follows the arguments from the proof of Lemma \ref{lemmoment}.
\end{proof}

We denote by $E_d = E \times E \times \ldots \times E$ the $d$ times product space of any space $E$. We also define the null vector of the space $\reels_d$ as $0_d=(0,\ldots,0)$. We now show in the following lemma that the  asymptotic variance matrix $\overline{M}$ is positive definite. This corresponds to Condition (iv) (b) from Theorem 5.17 in \cite{white2001asymptotic}.

\begin{lemma}
\label{lemmatrix2}
We assume that Condition \ref{condh} holds. Then, we have that the asymptotic variance matrix $\overline{M}$ is positive definite.
\end{lemma}

\begin{proof}[Proof of Lemma \ref{lemmatrix2}]
To verify that the matrix $\overline{M}$ is positive definite, we can check the definition of positive definite matrix. For any vector $v \in \reels_d$ such that $v \neq 0_d$, we have by Definition  (\ref{defMoverline}) that
\be
\label{lemmatrix2p1}
v'\overline{M} v= v'\esp [\overline{x}_{T\Delta} \overline{x}_{T\Delta}']v.
\ee
Then, we obtain by linearity of the expectation value that
\be
\label{lemmatrix2p2}
v'\esp [\overline{x}_{T\Delta} \overline{x}_{T\Delta}']v= \esp [v'\overline{x}_{T\Delta} \overline{x}_{T\Delta}'v].
\ee
This can be reexpressed as
\be
\label{lemmatrix2p3}
\esp [v'\overline{x}_{T\Delta} \overline{x}_{T\Delta}'v]= \esp [(v'\overline{x}_{T\Delta})^2].
\ee
Moreover, we get by Definition (\ref{defxoverline}) that
\be
\label{lemmatrix2p4}
\esp [(v'\overline{x}_{T\Delta})^2] > 0.
\ee
Finally, we have shown the the matrix $\overline{M}$ is positive definite by Expressions (\ref{lemmatrix2p1}), (\ref{lemmatrix2p2}), (\ref{lemmatrix2p3}) and (\ref{lemmatrix2p4}).
\end{proof}

We show the central limit theorem based on the stationary point process $\overline{N}$  (\ref{clteqstat}) holds in the following lemma. This relies on the technology of \cite{white2001asymptotic} (Theorem 5.17, p. 126). 

\begin{lemma}
\label{lemcltstat}
We assume that Condition \ref{condh} holds. Then, we have that the central limit theorem of the parametric estimation procedure based on ordinary least squares with the stationary point process $\overline{N}$ as the final time $T \rightarrow + \infty$, namely
\begin{eqnarray}
\nonumber
\overline{D}^{-1/2} \sqrt{T\Delta}(\widehat{\overline{\theta}}_{T\Delta} - \theta^*) \cvdistrib  \xi.
\end{eqnarray}
\end{lemma}

\begin{proof}[Proof of Lemma \ref{lemcltstat}]

We verify that the conditions from Theorem 5.17 in \cite{white2001asymptotic} are satisfied. First, Condition (i) holds since we reexpress the Hawkes process as a standard linear regression model of the form (\ref{deflr}). Then, Condition (ii) is satisfied by an application of Lemma \ref{lemstationary} and Lemma \ref{lemergodic}. Condition (iii) (a) holds by Lemma \ref{lemmixingale}. In addition, Condition (iii) (c) holds with the use of Lemma \ref{lemmoment}. Moreover, Condition (iii) (d) is obtained by Lemma \ref{lemmatrix}. Condition (iv) (a) holds with the use of Lemma \ref{lemmoment2}. Finally, Condition (iv) (b) is obtained by Lemma \ref{lemmatrix2}.

\end{proof}

We now give the proof of Theorem \ref{thclt}. This is a direct consequence to Lemma \ref{lemstat} and Lemma \ref{lemcltstat}.

\begin{proof}[Proof of Theorem \ref{thclt}]

This is a direct consequence to Lemma \ref{lemstat} and Lemma \ref{lemcltstat}.

\end{proof}

We now give the proof of Theorem \ref{thconsistency}. This is an application of \cite{white2001asymptotic} (Theorem 6.21, p. 159). The proof follows by an extension of the proof of Theorem \ref{thclt}. 

\begin{proof}[Proof of Theorem \ref{thconsistency}]

The proof follows by an application of \cite{white2001asymptotic} (Theorem 6.21, p. 159) with an extension of the proof of Theorem \ref{thclt}.

\end{proof}

Moreover, we establish the proof of Theorem \ref{thcltest}. The main idea of the proof is to apply the technology of \cite{white2001asymptotic} (Theorem 5.17, p. 126) with Theorem \ref{thclt} and Theorem \ref{thconsistency}.  

\begin{proof}[Proof of Theorem \ref{thcltest}]

First, we can deduce by Theorem \ref{thclt} that the conditions (i), (ii), (iii) and (iv) from Theorem 5.17 (p. 126) in \cite{white2001asymptotic} are satisfied. In addition, we can deduce by Theorem \ref{thconsistency} that the condition (v) from Theorem 5.17 (p. 126) in \cite{white2001asymptotic} is satisfied. Thus, we can use Theorem 5.17 (p. 126) in \cite{white2001asymptotic}.

\end{proof}

Finally, we detail the proof of Corollary \ref{corwald}. The proof of the convergence under the null hypothesis $H_0$ is obtained by applying the delta method (see Theorem 20.8 (p. 297) in \cite{van1998asymptotic}) to Theorem \ref{thclt} and applying Slutsky's theorem to Theorem \ref{thconsistency}.

\begin{proof}[Proof of Corollary \ref{corwald}]
First, we have by Definition (\ref{defW}) that the Wald test statistic is equal to
$$S_{T\Delta} =
T\big(R \widehat{\theta}_{T\Delta} -r\big)' \big(R\widehat{D}_{T\Delta} R' \big)^{-1}\big(R\widehat{\theta}_{T\Delta}-r\big).$$
We start under the null hypothesis $H_0$. We focus on the function $g_S(\theta) =
T\big(R \theta -r\big)' \big(RD R' \big)^{-1}\big(R\theta-r\big)$. Then, we can show that the conditions from Theorem 20.8 (p. 297) in  \cite{van1998asymptotic} applied on the function $g_S(\theta)$ hold. In particular, we can show the Hadamard differentiability. Moreover, an application of Theorem 20.8 (p. 297) in  \cite{van1998asymptotic} on the function $g_S(\theta)$ to Theorem \ref{thclt} yields that the random variable $T\big(R \widehat{\theta}_{T\Delta} -r\big)' \big(R D R' \big)^{-1}\big(R\widehat{\theta}_{T\Delta}-r\big)$ converges in distribution to a chi-squared random variable with $q$ degrees of freedom under the null hypothesis $H_0$ as the final time $T \rightarrow + \infty$.  In addition, the test statistic $S_{T\Delta}$ converges in distribution to a chi-squared random variable with $q$ degrees of freedom under the null hypothesis $H_0$ as the final time $T \rightarrow + \infty$ by an application of Slutsky's theorem to Theorem \ref{thconsistency}. Finally, we can show that the test statistic $S_{T\Delta}$ is also consistent under the alternative hypothesis  $H_1$, namely we have $\proba ( S_{T\Delta} > Q(p) \mid H_1 ) \rightarrow 1$ for any probability $p \in (0,1)$ as the final time $T \rightarrow + \infty$, by extending the arguments used under the null hypothesis $H_0$.
\end{proof}

\end{document}